\documentclass[11pt, letterpaper]{article}

\usepackage{amsmath, amssymb, amsthm, amsfonts, bbm, commath}
\usepackage{calc}
\usepackage[inline]{enumitem}
\usepackage{tikz}
\usetikzlibrary{arrows.meta}
\usepackage{subcaption}
\usepackage{authblk}
\usepackage{graphicx}
\usepackage{hyperref}
\usepackage{xcolor}

\title{
	Optimal Record and Replay under Causal Consistency
	\thanks{
		This research is supported in part by National Science Foundation award 1409416, and Toyota InfoTechnology Center. Any opinions, findings, and conclusions or recommendations expressed here are those of the authors and do not necessarily reflect the views of the funding agencies or the U.S. government.
	}
}

\author[1]{
	Russell L. Jones
}

\author[2]{
	Muhammad S. Khan
}

\author[3]{
	Nitin H. Vaidya
}

\affil[1]{
	Department of Electrical and Computer Engineering
}

\affil[2]{
	Department of Computer Science
}

\affil[ ]{
	University of Illinois at Urbana-Champaign \protect\linebreak
	\texttt{\{rjones27,mskhan6\}@illinois.edu} \protect\linebreak
}

\affil[3]{
	Department of Computer Science \protect\linebreak
	Georgetown University \protect\linebreak
	\texttt{nv198@georgetown.edu}
}

\newtheorem{theorem}{Theorem}[section]
\newtheorem{lemma}[theorem]{Lemma}

\newtheorem{definition}[theorem]{Definition}

\newtheorem{claim}[theorem]{Claim}
\newtheorem{observation}[theorem]{Observation}

\renewenvironment{proof}{\noindent{\bf Proof:} \hspace*{1mm}}{
	\hspace*{\fill} $\Box$ }
\newenvironment{proof_of}[1]{\noindent {\bf Proof of #1:}
	\hspace*{1mm}}{\hspace*{\fill} $\Box$ }
\newenvironment{proof_claim}{\begin{quotation} \noindent{\bf Proof:}}{
		\hspace*{\fill} $\diamond$ \end{quotation}}

\newcommand{ \DRO }[1]{ <_{ DRO(#1) } }
\newcommand{ \SCO }[1]{ \SCOi{}{#1} }
\newcommand{ \SCOi }[2]{ <_{ SCO_{#1}(#2) } }
\newcommand{ \SWO }[1]{ <_{ SWO(#1) } }

\newcommand{ \PO }{ <_{ PO } }
\newcommand{ \WO }{ <_{WO} }

\newcommand{ \V }{ \mathcal{V} }
\newcommand{ \pV }{ \mathcal{U} }
\newcommand{ \R }{ \mathcal{ R } }
\newcommand{ \A }{ \mathcal{ A } }

\newcommand{ \op }[2]{ {o}^{#1}_{#2}}
\newcommand{ \opr }[2]{ {r}^{#1}_{#2}}
\newcommand{ \opw }[2]{ {w}^{#1}_{#2}}

\newcommand{ \edge }[1]{ \del[0]{ #1 } }

\newcommand{ \trans }[1]{ \widehat{ #1 } }

\newcommand{ \cupdot }{ 
	\ensuremath{ \, \mathaccent \cdot \cup \, } 
}

\begin{document}
	\maketitle

	\begin{abstract}
		\normalsize
We investigate the minimum record needed to replay executions of processes
that share causally consistent memory.
For a version of causal consistency,
we identify optimal records under both offline and online recording setting.
Under the offline setting, a central authority has information about every process' view of the execution and can decide what information to record for each process.
Under the online setting, each process has to decide on the record at runtime as the operations are observed.


	\end{abstract}

	\section{Introduction} \label{section introduction}
In this paper, we explore {\em Record and Replay} (RnR) of multi-process applications where
processes communicate via shared memory.
Record and replay (RnR) mechanisms aim to allow parallel program debugging to proceed as follows.
The programmer runs the program, and potentially observes incorrect behavior.
The programmer then re-runs the program, while more closely watching the program state, and
attempts to discover where a program bug may have occurred.
However, even when a parallel program is re-executed with the same input, different
executions of the program may proceed differently, due to 
non-determinism introduced by the uncertainty in the delays incurred in performing various operations.
Thus, the observed bug may not re-occur during re-run, making it quite difficult to discover the cause of the original problem.
Record and Replay (RnR) aims to solve this problem by creating a {\em record} during the original execution, and using it during replay to guarantee that the re-run produces the same outcomes as the original execution.
In other words, while the original execution may be non-deterministic, the replay using the record
eliminates the non-determinism as desired.

There can be many sources of non-determinism in parallel programs.
For example user inputs, readings from sensors, random coin flips, etc.
However, in this paper we focus specifically on the non-determinism allowed by the shared memory consistency models in the read-write memory model.
For a given program, the shared memory consistency model defines a space of allowed executions possible when the program is run.
By creating a record during an execution and enforcing it in the replay, this space is further restricted hence reducing the inherent non-determinism.
The goal is to record enough from the original execution so as to reproduce the same outcomes in the replay.

The work in this paper is motivated by the trade-off between
the consistency model for shared memory and the amount of information
that must be recorded to facilitate a replay. 
A stronger consistency model imposes more constraints
on the execution, resulting in a smaller space of allowed executions. Intuitively, a stronger
consistency model should require a smaller record to resolve the non-determinism during replay.
In Section \ref{section causal consistency} we present an example execution to illustrate that
this intuition is indeed correct.
In prior work, Netzer \cite{Netzer1993Optimal} identified the minimum record necessary for RnR under
the \emph{sequential consistency} model \cite{LamportSequential}. The 
computer architecture research 
community has also
investigated RnR systems under various consistency models, for example \cite{Dunlap}, \cite{Torrellas}, and \cite{Lee:2010:REO:1736020.1736031}.
See also a survey by Chen et. al. \cite{Chen:2015:DRS:2830539.2790077}.
However, to the best of our knowledge, only Netzer's work \cite{Netzer1993Optimal} has addressed identification of \emph{minimum} record for RnR under read-write memory model.

This paper builds on Netzer's work to address the minimum record for correct replay under
{\em causal consistency}. Whether a certain record is necessary and sufficient for replay depends
on several factors, as discussed next. Lee et. al. \cite{Lee:2010:REO:1736020.1736031} have also
discussed a classification of RnR strategies 
and Chen el. al. \cite{Chen:2015:DRS:2830539.2790077} provide a taxonomy of deterministic replay schemes
.
\begin{enumerate}[label=\arabic*), wide, labelwidth=!, labelindent=0pt]
	\item \emph{How faithful should the replay be to the original execution?}
To understand the different scenarios that are plausible, let us consider an implementation
of shared memory. Suppose that each process maintains a local replica of the shared variables.
When a process writes to a shared variable, the new value is propagated to other processes
via update messages. The new value is eventually written at each replica, while ensuring
that the consistency model is obeyed. Figure \ref{figure how faithful replay}(a) illustrates an execution of two processes
that implement sequential consistency. In this case, in the original execution, $x$ is updated
to equal 1 due to the write operation $\opw{}{1}(x=1)$ by process $1$, and then $y$ is updated
to 2 due to the write operation $\opw{}{2}(y=2)$ by process $2$. Subsequently, process $1$ reads $y$ as 2 with the read operation $\opr{}{1}(y=2)$.
Figures \ref{figure how faithful replay}(b) and (c) show two possible replays of the execution in Figure \ref{figure how faithful replay}(a). Observe that, while the read returns the same value in both replays, the order in which the variables are updated
is different in the replay in Figure \ref{figure how faithful replay}(b) than the original execution. On the other hand, the
replay in Figure \ref{figure how faithful replay}(c) performs the updates in an identical order as in the original execution.

\begin{figure*}
	\centering
	{
		\begin{tikzpicture}
		\node at (-2, 2) {Process 1:};
		\node at (-2, 1) {Process 2:};
		\node (wx1) at (0, 2) {$\opw{}{1}(x=1)$};
		\node (wy2) at (3, 1) {$\opw{}{2}(y=2)$};
		\node (ry2) at (6, 2) {$\opr{}{1}(y=2)$};
		
		\path[->]
		(wx1) edge (wy2)
		(wy2) edge (ry2)
		;
		\end{tikzpicture}
		\subcaption{Original Execution}
		\hrule
		\vspace{\floatsep}
		\begin{tikzpicture}
		\node at (-2, 2) {Process 1:};
		\node at (-2, 1) {Process 2:};
		\node (wx1) at (3, 2) {$\opw{}{1}(x=1)$};
		\node (wy2) at (0, 1) {$\opw{}{2}(y=2)$};
		\node (ry2) at (6, 2) {$\opr{}{1}(y=2)$};
		
		\path[->]
		(wy2) edge (wx1)
		(wx1) edge (ry2)
		;
		\end{tikzpicture}
		\subcaption{Replay 1}
		\hrule
		\vspace{\floatsep}
		\begin{tikzpicture}
		\node at (-2, 2) {Process 1:};
		\node at (-2, 1) {Process 2:};
		\node (wx1) at (0, 2) {$\opw{}{1}(x=1)$};
		\node (wy2) at (3, 1) {$\opw{}{2}(y=2)$};
		\node (ry2) at (6, 2) {$\opr{}{1}(y=2)$};
		
		\path[->]
		(wx1) edge (wy2)
		(wy2) edge (ry2)
		;
		\end{tikzpicture}
		\subcaption{Replay 2}
		\hrule
		\vspace{\floatsep}
	}

	\caption{An example illustrating how different replays may reproduce the same read values. The arrows show the order in which the updates are propagated and read.}
	\label{figure how faithful replay}
\end{figure*}
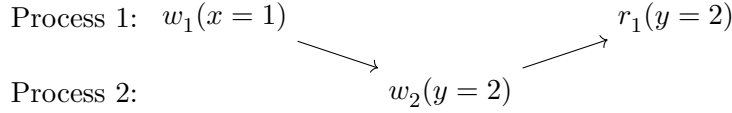
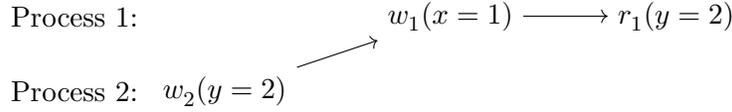
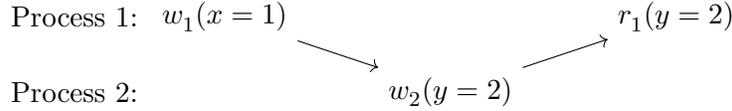

Depending on whether we must reproduce the replay as in Figure \ref{figure how faithful replay}(b), or allow a replay as
in Figure \ref{figure how faithful replay}(c), the minimum record necessary will be different.
As one may expect, the record required for replay in Figure \ref{figure how faithful replay}(b) is smaller, since
the replay is not as faithful as that in Figure \ref{figure how faithful replay}(c).
Netzer's minimum record \cite{Netzer1993Optimal} for sequential consistency
allows the replay in Figure \ref{figure how faithful replay}(b), which
ensures that all the reads and writes to the same variable occur in the same order during replay as in the original execution.
However, the updates to different variables may not necessarily occur in the same order
during replay as in the original execution.

At a minimum, the read operations in the replay must return the same values as the corresponding read operations in the original execution.
This ensures that the program state for each process, and so the output, in the replay is the same as the one in the original execution (i.e., the same branches are taken in both the executions as the next step to be performed by a process depends on the current program state and the values read from shared memory) and so the replay is indistinguishable to the high-level user from the original execution.
We discuss the exact formal model for this work in Section \ref{section RnR model}.


	\item \emph{At what level of abstraction is the RnR system implemented?}
	The abstraction level where the RnR system operates influences what can and needs to be recorded.
	For instance, if the shared memory is implemented via message passing, then,
	for the purpose of RnR, we may treat this as a message-passing system and record
	messages rather than shared memory operations.
	In this case, the RnR system can be viewed
	as residing below the shared memory implementation.

	Alternatively, the RnR system may operate at the library level where the low level details, including interactions with the shared memory, are abstracted via the provided libraries.
	The RnR system is only allowed to record interactions with the APIs of the given libraries.
	We refer the reader to \cite{Chen:2015:DRS:2830539.2790077} Section 4.3 for a more detailed explanation of different abstract levels.

	In this paper, our focus is on RnR for the shared memory.
	In our model, the RnR system resides on top of the shared memory layer so that the inner workings of
	the shared memory are abstracted while the interactions with the shared memory, via the read and write operations on shared variables, is exposed.
	In this case, we assume that
	the RnR module may observe, at each process, the reads of that process
	and the writes of all the processes.

	\item \emph{Offline versus online recording.}
	In the offline setting, the RnR module is provided with a completed execution in its entirety, and can use this information to obtain a record that suffices for a correct replay.
	In the online setting, each process has its own RnR module that observes
	the execution incrementally, and must decide incrementally what information 
	must be recorded.
	The online record can be useful when, for example, the replay proceeds in 
	tandem with the original execution for redundancy purposes.
	Netzer's result \cite{Netzer1993Optimal} applies to both the offline and online setting for sequential consistency.
	In this paper, we consider both offline and online settings in the context of 
	causal consistency.
\end{enumerate}

A summary of our contributions is presented in Table
\ref{table strong causal summary}.
In this work, we present the optimal record for a version of causal consistency 
which we call \emph{strong causal consistency}.
This is formally defined in Section \ref{section consistency} and is followed 
by many practical implementations of causal consistency.
We consider both the RnR model for replay as in Figure
\ref{figure how faithful replay}(b) and as in Figure
\ref{figure how faithful replay}(c).
These are defined formally in Section \ref{section RnR model}.
Sequential consistency was considered by Netzer \cite{Netzer1993Optimal}.
We consider the first RnR model in Section \ref{section optimal record}.
In Sections \ref{section strong causal consistency} and 
\ref{section strong causal consistency online} we present the optimal records 
for strong causal consistency for the offline and online scenarios respectively.
The question of optimal record for causal consistency is still open and we 
discuss this in Section \ref{section causal consistency}.
We consider the second RnR model in Section \ref{section optimal record DRO} 
with optimal record for the offline case of strong causal consistency given in 
Section \ref{section strong causal consistency DRO} and the one for causal 
consistency discussed in Section \ref{section causal consistency DRO}.
We finish the paper with a discussion in Section \ref{section discussion}, 
along with some open problems.

\newcommand\T{\rule{0pt}{2.6ex}}       
\newcommand\B{\rule[-1.2ex]{0pt}{0pt}} 

\begin{table*}
\small
\centering
\begin{tabular}{ |c|c|c| }
	\hline
	&	\multicolumn{2}{c|}{\textbf{Replay as in}}\T\B\\
	\hline
	&
	\textbf{Figure \ref{figure how faithful replay}(b) (resolves }	&
	\textbf{Figure \ref{figure how faithful replay}(c) (resolves }
	\T\B\\
	\textbf{Model Setting}	&
	\textbf{entire views identically)}	&
	\textbf{data races identically)}
	\T\B\\
	\hline
	\begin{tabular}{ c }
		Sequential Consistency \cite{LamportSequential} - Offline Recording
	\end{tabular}
	&	Similar to Netzer \cite{Netzer1993Optimal}	&	Netzer 
	\cite{Netzer1993Optimal}\T\B\\
	\hline
	\begin{tabular}{ c }
		Sequential Consistency \cite{LamportSequential} - Online Recording
	\end{tabular}
	&	Similar to Netzer \cite{Netzer1993Optimal}	&	Netzer 
	\cite{Netzer1993Optimal}\T\B\\
	\hline
	\begin{tabular}{ c }
		Strong Causal Consistency - Offline Recording
	\end{tabular}
	&	\textbf{This work}	&	\textbf{This work}	\T\B\\
	\hline
	\begin{tabular}{ c }
		Strong Causal Consistency - Online Recording
	\end{tabular}
	&	\textbf{This work}	&	Future work	\T\B\\
	\hline
	\begin{tabular}{ c }
		Causal Consistency \cite{Ahamad1995} - Offline Recording
	\end{tabular}
	&	Open	&	Open	\T\B\\
	\hline
	\begin{tabular}{ c }
		Causal Consistency \cite{Ahamad1995} - Online Recording
	\end{tabular}
	&	Open	&	Open	\T\B\\
	\hline
\end{tabular}
\caption{A summary of RnR results.}
\label{table strong causal summary}
\end{table*}

	\section{Preliminaries} \label{section preliminaries}
A \emph{relation} $R$ on a set $O$ is a set of tuples $\edge{ a, b }$ such that $a, b \in O$.
We use the notation $a <_{R} b$ if $\edge{ a, b } \in R$.
We denote $a \le_{R} b$ if either $a <_{R} b$ or $a = b$.
An irreflexive, antisymmetric, and transitive relation is called a \emph{partial order}.
A partial order $R$ on a set $O$ is a \emph{total order} if for any $a, b \in O$, either $a <_{R} b$ or $b <_{R} a$.
A partial order can be represented by a directed acyclic graph which is closed under transitivity.
For two relations $A$ and $B$ on a set $O$, we say that $A$ \emph{respects} $B$ if $B \subseteq A$.
We use the notation $A \mid O'$ to restrict the relation $A$ on set $O$ to a subset $O' \subseteq O$.
$\trans{A}$ denotes the (unique) transitive reduction of the partial order $A$ and $a \lessdot_A b$ denotes $\edge{ a, b } \in \trans{A}$.
We use $A \cup B$ to denote the union, with the transitive closure, of relations $A$ and $B$, and $A \cupdot B$ to denote the disjoint union of $A$ and $B$.
For example, consider two partial orders $A$ and $B$ on the set $\set{a, b}$, given by $A = \set{ \edge{ a, b } }$ and $B = \set{ \edge{ b, a } }$.
Then, $A \cup B = \set{ \edge{ a, b }, \edge{ b, a }, \edge{ a, a }, \edge{ b, b } }$ while $A \cupdot B = \set{ \edge{ a, b }, \edge{ b, a } }$.
Observe that union and disjoint union of two partial orders may not be a partial order, as the previous example shows.

We borrow some notation by Steinke and Nutt \cite{Steinke2004Unified} for shared memory formalism.
The shared memory consists of a set of variables $X$ and supports two operations, read and write.
We use $\opw{}{}$ for writes, $\opr{}{}$ for reads, and $\op{}{}$ when the operation can be either read or write.
We use a subscript for process identifier or leave it blank if it is unspecified.
If the variable and the corresponding value read/written is relevant, we specify it in parenthesis.
For example, $\opw{}{i}(x = 1)$ denotes a write of value $1$ to variable $x$ performed by process $i$ and $\op{}{j}(y)$ denotes an operation performed by process $j$ that can either be a read or write to variable $y$.
Formally, an \emph{operation} is a $4$-tuple $( op, i, x, id )$ where $op$ is $\operatorname{r}$ for read and $\operatorname{w}$ for write, $i$ is the unique identifier of the process that performed the operation, $x$ is the (shared) variable on which the operation was performed, and $id$ is the unique identifier of the operation.
This notation allows for wild-card entries, e.g. $( \operatorname{w}, i, *, * )$ is the set of all writes executed by process $i$.
Observe that we do not specify the values in the notation.
We assume that each write operation writes a unique value\footnote{Since the unique write values have a one-to-one correspondence with the unique identifiers of the respective write operations, therefore formally specifying the write values is redundant.}.
The values read by read operations may vary between executions, but each read operation reads a value written by some write.

All operations in $( *, i, *, * )$ are totally ordered.
We denote this total order by $PO(i)$.
The disjoint union of these is the \emph{program order} given by $PO = \cupdot_{i} PO(i)$.
This is the order on operations implied by the program text.
In figures representing total orders, we draw operations from left to right as they appear in the total order.
For example, Figure \ref{figure views causal weaker example}(a) draws the program order for two processes, $1$ and $2$.
The two total orders, $PO(1)$ and $PO(2)$, corresponding to processes $1$ and $2$ respectively, are drawn from left to right.

We model the distributed system as a network of processes that communicate with each other via reads and writes to the shared memory.
Each process comes with a program that specifies the operations to be executed 
and the order in which they should be executed.
Formally, a shared memory \emph{system} is a set of processes $P$, a set of operations $O$, a program order $PO$ on $O$, a set of shared variables $X$, and a shared memory $\Pi$.
An execution is the result of processes running their programs on a shared memory system where each read operation returns a value written by some write operation.
\begin{definition}[Writes-to]
	Given an execution, a write operation $w$ \emph{writes-to} a read operation $r$, denoted $w \mapsto r$, if $w$ and $r$ are on the same variable and $r$ returns the value written by $w$.
\end{definition}

We reason about executions as a collection of read and write operations on shared variables.
We do not distinguish any operation as special, e.g. synchronization operation, but view all operations to the shared memory uniformly.
This is the same as Netzer's model \cite{Netzer1993Optimal}.
		\subsection*{Assumptions about Programs}
In general, programs are dynamic where the next operation to be executed 
depends on the current program state.
Our model requires reproducing the execution faithfully; at the 
very least all read operations must return the same values.
Since we consider deterministic programs, that read the same values from the 
shared memory via the corresponding read operations, therefore we claim, 
without proof, that program at each process will execute 
the same operations in the same order in both the original execution and the 
replay.
A similar result is shown in \cite{Mellor-Crummey:1989:DAL:916540} for a 
different setting.
So we assume that the program order $PO$ is fixed.

One standard practice for writing concurrent programs is to ensure that they 
are properly synchronized such that they are data race free \cite{AdveHill}.
This guarantees sequential semantics for such programs under most concurrent 
languages and multiprocessors.
We do \underline{not} make any such assumptions since
\begin{enumerate}[label=\arabic*),wide,labelwidth=!,labelindent=0pt,topsep=0pt]
	\item
	we do not distinguish any operation as special, e.g. synchronization 
	operations,

	\item
	one of the aims of this work is to replay programs for debugging purposes, so 
	assuming that the programmer has written the program correctly is a dangerous 
	assumption, and

	\item
	the guarantee of sequential semantics for data race free programs is for a 
	different consistency model (cache consistency) and it does not hold for 
	causal consistency.
\end{enumerate}

	\section{Shared Memory Consistency} \label{section consistency}
For an execution, a \emph{view} $V$ on a set of operations $O' \subseteq O$ is a total order on $O'$ such that each read $r \in O'$ returns the last value written to the corresponding variable in $V$.
For a view $V$, the \emph{data-race order} is given by $DRO(V) = \bigcup_{x \in X} V \mid (*, *, x, *)$.
Reasoning about allowed executions under a shared memory consistency model relies on existence of some collection of views $\V$ that satisfy some properties, depending on the shared memory consistency model.
We say that $\V$ \emph{explains} the execution under the consistency model.
For example, causal consistency \cite{Ahamad1995} requires existence of per-process views that satisfy causality, which is the union (with the transitive closure) of the writes-to relation and the program order.
Formally, we use the definition by Steinke and Nutt \cite{Steinke2004Unified}.

\begin{definition}[{Write-read-write Order [Steinke and Nutt \cite{Steinke2004Unified}]}] \label{definition write read write order}
	Given an execution with a writes-to relation $\mapsto$, two writes, $\opw{1}{} \in (\operatorname{w}, *, *, *)$ and $\opw{2}{} \in (\operatorname{w}, *, *, *)$, are ordered by \emph{write-read-write order}, $\edge{ \opw{1}{},  \opw{2}{} } \in {WO}$, if there exists a read operation $\opr{}{} \in (\operatorname{r}, *, *, *)$ such that $\opw{1}{} \mapsto \opr{}{} <_{PO} \opw{2}{}$.
\end{definition}

\begin{definition}[{Causal Consistency [Steinke and Nutt \cite{Steinke2004Unified}]}] \label{definition causal consistency}
	An execution is \emph{causally consistent} if there exists a set of views $\V = \set{ V_i }_{i \in P}$ such that, for every process $i$,
	\begin{itemize}[topsep=0pt, itemsep=0pt]
		\item $V_i$ is a view on the set of operations $( *, i, *, * ) \cup ( \operatorname{w}, *, *, * )$, and
		\item $V_i$ respects ${WO} \cup \del{{PO} | ( *, i, *, * ) \cup ( \operatorname{w}, *, *, * )}$.
	\end{itemize}
	A shared memory $\Pi$ is \emph{causally consistent} if every execution run on $\Pi$ is causally consistent.
\end{definition}

Note that, by definition, each view $V_i$ already respects the writes-to relation restricted to $( *, i, *, * ) \cup ( \operatorname{w}, *, *, * )$ since, by definition of a view, each read returns the last value written to the corresponding variable in $V_i$.
Note also that read operations are only observed by the processes that perform them while write operations are observed by every process.
We work with a version of causal consistency which we call \emph{strong causal consistency}.
This model is motivated by an implementation of causal consistency via lazy replication \cite{Ladin:1992:PHA:138873.138877}.
Ladin et. al. \cite{Ladin:1992:PHA:138873.138877} use vector timestamps to ensure that a write operation $w_i$ from process $i$ is only committed locally when all write operations in $w_i$'s history, as summarized by $w_i$'s vector timestamp, have been observed.
Many practical systems use vector timestamps to determine order of operations and detect conflicts in systems with weak consistency guarantees (e.g. Dynamo \cite{DeCandia:2007:DAH:1323293.1294281}, COPS \cite{Lloyd:2011:DSE:2043556.2043593}, and Bayou \cite{Terry:1995:MUC:224056.224070}) although these systems have conflict resolution schemes which make their actual consistency guarantees stronger than strong causal consistency (see also Section \ref{section discussion}).
Formally, we define strong causal consistency as follows.

\begin{definition} [Strong Causal Order] \label{definition LRO strong causal order}
	Given a set of views $\V = \set{ V_i }_{i \in P}$, two writes, $\opw{1}{} \in (\operatorname{w}, *, *, *)$ and $\opw{2}{i} \in (\operatorname{w}, i, *, *)$, are ordered by \emph{strong causal order}, $\edge{ \opw{1}{}, \opw{2}{i} } \in { SCO( \V ) }$, if $\edge{ \opw{1}{}, \opw{2}{i} } \in  { V_i } $.
\end{definition}

This is stronger than the write-read-write order $WO$ since two writes $\opw{1}{} \in ( \operatorname{w}, *, *, * )$ and $\opw{2}{i} \in (\operatorname{w}, i, *, *)$ are ordered by $WO$ if and only if $\opw{1}{}$ has been read by process $i$ before it performs $\opw{2}{i}$.
However, $\opw{1}{}$ has to be merely observed by process $i$ for the two operations to be ordered by strong causal order.
Intuitively, this corresponds to causality when each write operation observed is immediately read.

\begin{definition} [Strong Causal Consistency] \label{definition LRO strong causal consistency}
	An execution is \emph{strongly causal consistent} if there exists a set of views $\V = \set{ V_i }_{i \in P}$ such that, for every process $i$,
	\begin{itemize}[topsep=0pt, itemsep=0pt]
		\item $V_i$ is a view on the set of operations $( *, i, *, * ) \cup ( \operatorname{w}, *, *, * )$, and
		\item $V_i$ respects ${SCO( \V )} \cup \del{{PO} | ( *, i, *, * ) \cup ( \operatorname{w}, *, *, * )}$.
	\end{itemize}
	A shared memory $\Pi$ is \emph{strongly causal consistent} if every execution run on $\Pi$ is strongly causal consistent.
\end{definition}

\begin{figure*}
	\centering
	{
		\begin{tikzpicture}
		\node at (0, 2.5) {Process 1:};
		\node at (0, 1) {Process 2:};
		
		\node (w1x) at (2, 2.5) {$\opw{}{1}(x=1)$};
		\node (r11x) at (5, 2.5) {$\opr{1}{1}(x=1)$};
		\node (r1y) at (8, 2.5) {$\opr{}{1}(y=4)$};
		\node (w1y) at (11, 2.5) {$\opw{}{1}(y=3)$};
		\node (r21x) at (14, 2.5) {$\opr{2}{1}(x=1)$};
		
		\node (w2x) at (2, 1) {$\opw{}{2}(x=2)$};
		\node (r12x) at (5, 1) {$\opr{1}{2}(x=2)$};
		\node (w2y) at (8, 1) {$\opw{}{2}(y=4)$};
		\node (r2y) at (11, 1) {$\opr{}{2}(y=3)$};
		\node (r22x) at (14, 1) {$\opr{2}{2}(x=2)$};
		
		\path[->]
		(w1x) edge node[above] {$PO$} (r11x)
		(r11x) edge node[above] {$PO$} (r1y)
		(r1y) edge node[above] {$PO$} (w1y)
		(w1y) edge node[above] {$PO$} (r21x)
		
		(w2x) edge node[below] {$PO$} (r12x)
		(r12x) edge node[below] {$PO$} (w2y)
		(w2y) edge node[below] {$PO$} (r2y)
		(r2y) edge node[below] {$PO$} (r22x)
		
		
		(w2y) edge[line width=1.2pt, bend right=10] node[above left] {$WO$} (w1y)
		;
		\end{tikzpicture}
		\subcaption{A two process program with read and write values for one 
		possible causally consistent execution.}
		\hrule
		\vspace{\floatsep}
		%
		\centering
		\begin{tikzpicture}
		\node at (-2, 3) {$V_1$:};
		\node (V1w2x) at (0, 3) {$\opw{}{2}(x)$};
		\node (V1w1x) at (2, 3) {$\opw{}{1}(x)$};
		\node (V1r11x) at (4, 3) {$\opr{1}{1}(x)$};
		\node (V1w2y) at (6, 3) {$\opw{}{2}(y)$};
		\node (V1r1y) at (8, 3) {$\opr{}{1}(y)$};
		\node (V1w1y) at (10, 3) {$\opw{}{1}(y)$};
		\node (V1r21x) at (12, 3) {$\opr{2}{1}(x)$};
		
		\node at (-2, 1) {$V_2$:};
		\node (V2w1x) at (0, 1) {$\opw{}{1}(x)$};
		\node (V2w2x) at (2, 1) {$\opw{}{2}(x)$};
		\node (V2r12x) at (4, 1) {$\opr{1}{2}(x)$};
		\node (V2w2y) at (6, 1) {$\opw{}{2}(y)$};
		\node (V2w1y) at (8, 1) {$\opw{}{1}(y)$};
		\node (V2r2y) at (10, 1) {$\opr{}{2}(y)$};
		\node (V2r22x) at (12, 1) {$\opr{2}{2}(x)$};
		
		\path[->]
		(V1w1x) edge[dashed] (V1r11x)
		(V1w1x) edge[dashed, bend left=10] (V1r21x)
		(V2w2x) edge[dashed] (V2r12x)
		(V2w2x) edge[dashed, bend left=15] (V2r22x)
		(V1w2y) edge[dashed] (V1r1y)
		(V2w1y) edge[dashed] (V2r2y)
		
		(V1w2y) edge[line width=1.2pt, bend right=20] node[below] {$WO$} (V1w1y)
		(V2w2y) edge[line width=1.2pt] node[above] {$WO$} (V2w1y)
		
		(V2w1x) edge[bend right=20] node[above] {$PO$} (V2w1y)
		(V2r2y) edge node[below] {$PO$} (V2r22x)
		
		(V1w1y) edge node[below] {$PO$} (V1r21x)
		(V1w2x) edge[bend right=15] node[below] {$PO$} (V1w2y)
		;
		\end{tikzpicture}
		\subcaption{A set of views that explain the execution given in (a) under 
		causal consistency. The values of read and write operations have been 
		omitted with the dotted edges giving the writes-to relation.}
		\hrule
		\vspace{\floatsep}
	}
	\caption{An execution which is causally consistent but not strongly causal 
	consistent.}
	\label{figure views causal weaker example}
\end{figure*}
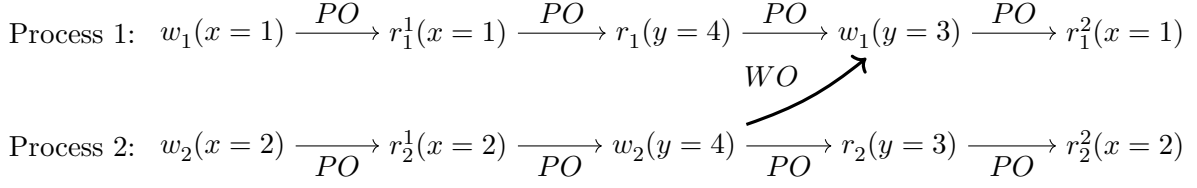
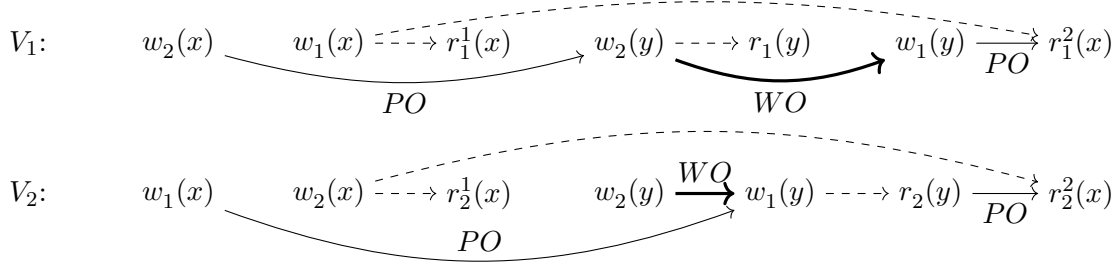

Observe that strong causal consistency does not violate the write-read-write order and thus it is at least as strong as causal consistency.
In fact, it is strictly stronger than causal consistency.
Figure \ref{figure views causal weaker example} shows a causally consistent execution of a two process program.
The read and write values are given in Figure \ref{figure views causal weaker example}(a).
Figure \ref{figure views causal weaker example}(b) gives a set of views that explains this execution under causal consistency.
The values of read and write operations have been omitted with the dotted edges giving the writes-to relation.
Some obvious $PO$ edges have also been omitted to avoid clutter.
We reason that no set of views can explain the execution under strong causal consistency.
Observe that ordering $\edge{ \opw{}{2}(x), \opw{}{1}(x) } \in V_1$ implies an $SCO(\V)$ edge that must be respected by $V_2$.
Therefore, any set of views that explain the execution under strong causal consistency must have either $\edge{ \opw{}{2}(x), \opw{}{1}(x) } \in V_2$ or $\edge{ \opw{}{1}(x), \opw{}{2}(x) } \in V_1$.
We show that none of these is possible.

For the first case, note that $\opw{}{1}(x) \PO \opw{}{1}(y) \mapsto \opr{}{2}(y) \PO \opr{2}{2}(x)$.
Therefore $\opw{}{1}(x)$ can not be placed after $\opr{2}{2}(x)$ in $V_2$.
Now if $\opw{}{1}(x)$ is placed after $\opw{}{2}(x)$ in $V_2$, then $\opr{2}{2}(x)$ does not return the last value written to $x$ in $V_2$.
This violates the definition of a view.

For the second case, we have that $\opw{}{2}(x) \PO \opw{}{2}(y) \WO \opw{}{1}(y) \PO \opr{2}{1}(x)$.
Therefore $\opw{}{2}(x)$ can not be placed after $\opr{2}{1}(x)$ in $V_1$.
Now if $\opw{}{2}(x)$ is placed after $\opw{}{1}(x)$ in $V_1$, then $\opr{2}{1}(x)$ does not return the last value written to $x$ in $V_1$.
Again, this violates the definition of a view.

		\subsection*{Compiler and Hardware Optimizations}
In real world systems, many optimizations are applied to the provided program by both the compiler at compile time and the hardware at runtime.
The shared memory consistency model ensures that these optimizations are such that the guarantees provided are still maintained by these optimizations.
For example, consider a uniprocessor and a shared memory consistency model that guarantees a view consistent with the program order implied by the written program.
The compiler and hardware optimizations may result in operations being executed out of order in apparent violation of the  program order constraints.
However, the resulting execution can still be explained by the \emph{existence} of a view (or views) where the operations are executed exactly as specified by the program order.
Using view based definitions of shared memory consistency models allows us to abstract these implementation details.
Therefore we allow all optimizations to be applied to the given program as long 
as the relevant shared memory consistency guarantees are satisfied.

	\section{RnR Model} \label{section RnR model}
For replaying executions, we assume that the per-process views are provided to the RnR system.
The RnR system uses the views to determine the record.
In case of online recording, the views are provided to the RnR system incrementally, as and when new operations occur that affect the views.
Now let us illustrate how this requirement may be implemented in practice.
Consider a shared memory implementation wherein each process has a copy of the shared variable and the shared memory is implemented via message passing.
Then the shared memory adds a write operation to process $i$'s view when the local copy of the corresponding variable is updated at process $i$.
Similarly a read  by process $i$ is added to process $i$'s view when the local copy is read.

The RnR system will record some edges from each view (i.e. on each process) and the replay execution is only allowed views that enforce these records.
Note that we do not place any restriction on how the record is enforced.
We assume that any set of views can explain the replay as long as it extends the record and is consistent under the shared memory consistency model.
Formally, we define two RnR models with different fidelities.
Under the first model, the RnR system is allowed to record any edge from each view and we require that the replay reproduces the per-process views exactly as in the original execution.
Under the second model, the RnR system is only allowed to record data races from each view and we only require that the data races are resolved identically in the replay.

%
%
%
\noindent
\textbf{RnR Model 1:}
Given a set of views $\V = \set{V_i}_{i \in P}$, $\R = \set{R_i}_{i \in P}$ is a \emph{{record}} of $\V$ if each $R_i \subseteq V_i$.
An execution is a \emph{{replay}} of $\R$ if there exists a set of views $\V' = \set[0]{ V'_i }_{ i \in P }$ that explain the execution under the consistency model and each $V'_i$ respects $R_i$.
We say that $\V'$ \emph{{certifies}} the replay to be valid for $\R$.
A record $\R$ of a set of views $\V$ is \emph{{good}} if, for any replay of $\R$, under the same consistency model, any set of views $\V' = \set[0]{ V'_i }_{ i \in P }$ that certifies the replay to be valid for $\R$ must have $V'_i = V_i$ for all $i \in P$ (i.e. only $\V$ certifies the replay to be valid for $\R$).

\noindent
\textbf{RnR Model 2:}
Given a set of views $\V = \set{V_i}_{i \in P}$, $\R = \set{R_i}_{i \in P}$ is a \emph{{record}} of $\V$ if each $R_i \subseteq DRO( V_i )$.
An execution is a \emph{{replay}} of $\R$ if there exists a set of views $\V' = \set[0]{ V'_i }_{ i \in P }$ that explain the execution under the consistency model and each $V'_i$ respects $R_i$.
We say that $\V'$ \emph{{certifies}} the replay to be valid for $\R$.
A record $\R$ of a set of views $\V$ is \emph{{good}} if, for any replay of $\R$, under the same consistency model, any set of views $\V' = \set[0]{ V'_i }_{ i \in P }$ that certifies the replay to be valid for $\R$ must have $DRO(V'_i) = DRO(V_i)$ for all $i \in P$.

The second replay model is the same as the one considered by Netzer \cite{Netzer1993Optimal}.
Observe that for each record, there exists at least one replay, specifically the original execution.
Note that RnR Model 1 forces all writes to appear in the same order for a process' view as they did in the original execution, which is different than Netzer's model in \cite{Netzer1993Optimal}.
This may seem expensive since reordering writes to different variables can result in performance optimizations while still returning the same values for reads and allowing the program state in the replay to progress the same as in the original execution.
RnR Model 2 allows writes to different variables to be executed in different order, which is the same as Netzer's model in \cite{Netzer1993Optimal}.
But for RnR Model 1 we require that each process' point of view with respect to the order of events must be indistinguishable between the original execution and the replay.

In contrast to the discussion at the end of Section \ref{section consistency}, the optimizations for the replay execution may be more restrictive than those for the original execution.
Exactly what optimizations are allowed in the replay execution versus the original execution depends on the shared memory consistency model as well as the replay system implementation.
In this work, we do not discuss replay systems, their implementations, or how they may enforce the provided record.
So we do not discuss the optimizations during the replay.

	\section{Optimal Records for RnR Model 1} \label{section optimal record}
		\subsection{Offline Record for Strong Causal Consistency} \label{section strong causal consistency}
In this section we consider offline record for strong causal consistency.
In this case the entire set of per-process views $\V = \set{V_i}_{i \in P}$ is made available to the RnR system.
The RnR system determines the record that must be saved.
If the RnR system decides to record the entire views $V_i$ for every process $i$, then this would be sufficient to reproduce the original execution exactly.
However, this is wasteful since the transitive reduction $\trans{V}_i$ for each process $i$ would also achieve the same result.

We first give intuition on what edges from each $\trans{V}_i$ do not need to be recorded before formalizing it in Theorem \ref{theorem:RnRStrongCausalConsistencySufficiency}.
Fix a process $i$.
Since $PO$ is fixed and independent of executions the RnR system does not have to record these edges in $V_i$ as they are guaranteed by the consistency model.
Now consider two write operations $\opw{1}{} \in (\operatorname{w}, *, *, *)$ and $\opw{2}{j} \in (\operatorname{w}, j, *, *)$, for $j \ne i$, such that $\edge{ \opw{1}{}, \opw{2}{j} } \in SCO(\V)$.
If process $j$ correctly orders the two operations $\edge{ \opw{1}{}, \opw{2}{j} }$ in the replay, then this edge will be guaranteed by the consistency model, due to strong causal order, and process $i$ does not need to record it.
Such edges are captured by the following definition.

\begin{definition}
	Given a set of views $\V = \set{ V_i }_{i \in P}$, the relation $SCO_i( \V )$, for a process $i \in P$, is defined as follows.
	Two writes, $\opw{1}{} \in (\operatorname{w}, *, *, *)$ and $\opw{2}{j} \in (\operatorname{w}, j, *, *)$, are ordered $\edge{ \opw{1}{}, \opw{2}{j} } \in { SCO_i( \V ) }$, if $\edge{ \opw{1}{}, \opw{2}{j} } \in SCO(\V)$ and $j \ne i$.
\end{definition}

Observe that the subscript distinguishes the relation $SCO_i( \V )$ from $SCO( \V )$ (Definition \ref{definition LRO strong causal order}) which is a partial order for strongly causal executions.
We now present an example to illustrate another set of edges that do not need to be recorded, although they are not directly guaranteed by the consistency model.
Consider the following execution on three processes and a set of views that explains it under strong causal consistency (Figure \ref{figure B_i example}).
Process $1$ performs the write $\opw{}{1} \in ( \operatorname{w}, 1, *, *)$, process $2$ performs $\opw{}{2} \in ( \operatorname{w}, 2, *, * )$, and process $3$ does not perform any operations.
Now process $1$ orders $\opw{}{1} <_{V_1} \opw{}{2}$, process $2$ orders $\opw{}{2} <_{V_2} \opw{}{1}$, and process $3$ orders $\opw{}{1} <_{V_3} \opw{}{2}$.
It can be easily verified that this set of views satisfies Definition \ref{definition LRO strong causal consistency} of strong causal consistency where both $PO$ and $SCO( \V )$ are empty.
Now note that if process $3$ records $\opw{}{1} <_{R_3} \opw{}{2}$, process $1$ does not need to record its order of the two operations.
The reason is that any possible set of views $\V' = \set[0]{ V'_i }_{ i \in P }$, that certify a replay to be valid for $\R$, will have $V'_3$ order $\opw{}{1} <_{V'_3} \opw{}{2}$.
So if process $1$ orders $\opw{}{2} <_{V'_1} \opw{}{1}$, this will create an $SCO(\V')$ edge $\opw{}{2} \SCO{\V'} \opw{}{1}$.
Since $V'_3$ respects $SCO(\V')$, therefore process $3$ will order $\opw{}{2} <_{V'_3} \opw{}{1}$.
This conflicts with the recorded edge $\opw{}{1} <_{R_3} \opw{}{2}$.
Thus, such a set of views can not certify a replay execution to be valid for $\R$.
The set of such edges is captured by the following relation.

\begin{definition} \label{definition LRO Bi}
	Given a set of views $\V = \set{ V_i }_{i \in P}$, the relation $B_i( \V )$, for a process $i \in P$, is defined as follows.
	Two writes, $\opw{1}{i} \in (\operatorname{w}, i, *, *)$ and $\opw{2}{j} \in (\operatorname{w}, j, *, *)$ such that $i \ne j$, are ordered $\edge{ \opw{1}{i}, \opw{2}{j} } \in {B_i(\V)}$ if $\edge{ \opw{1}{i}, \opw{2}{j} } \in {V_i}$ and there exists a process $k \ne i, j$ such that $\edge{ \opw{1}{i}, \opw{2}{j} } \in {V_k}$.
\end{definition}

\begin{figure}
	\centering
	\begin{tikzpicture}
	\node at (-1, 4) {$V_1$:};
	\node at (-1, 3) {$V_2$:};
	\node at (-1, 2) {$V_3$:};
	\node (V1w1) at (0, 4) {$\opw{}{1}$};
	\node (V1w2) at (3, 4) {$\opw{}{2}$};
	\node (V2w1) at (3, 3) {$\opw{}{1}$};
	\node (V2w2) at (0, 3) {$\opw{}{2}$};
	\node (V3w1) at (0, 2) {$\opw{}{1}$};
	\node (V3w2) at (3, 2) {$\opw{}{2}$};
	\path[->]
	(V1w1) edge (V1w2)
	(V2w2) edge (V2w1)
	(V3w1) edge[color=red] node[above] {$R_3$} (V3w2)
	;
	\end{tikzpicture}
	
	\vspace{\floatsep}
	\vspace{\floatsep}
	
	\begin{tikzpicture}
	\node at (5, 4) {$V'_1$:};
	\node at (5, 3) {$V'_2$:};
	\node at (5, 2) {$V'_3$:};
	\node (V'1w1) at (9, 4) {$\opw{}{1}$};
	\node (V'1w2) at (6, 4) {$\opw{}{2}$};
	\node (V'2w1) at (9, 3) {$\opw{}{1}$};
	\node (V'2w2) at (6, 3) {$\opw{}{2}$};
	\node (V'3w1) at (9, 2) {$\opw{}{1}$};
	\node (V'3w2) at (6, 2) {$\opw{}{2}$};
	\path[->]
	(V'1w2) edge node[above] {$SCO(\V')$} (V'1w1)
	(V'2w2) edge node[above] {$SCO(\V')$} (V'2w1)
	(V'3w1) edge[bend left=20, color=red] node[below] {$R_3$} (V'3w2)
	(V'3w2) edge node[above] {$SCO(\V')$} (V'3w1)
	;
	\end{tikzpicture}
	\hrule
	\vspace{\floatsep}
	\caption{$\set{V_i}_{i=1}^{3}$ explains a strongly causal execution and 
		$\set[0]{V'_i}_{i=1}^{3}$ explains an invalid replay. Process $1$ orders 
		$\opw{}{1} <_{V_1} \opw{}{2}$ in the replay which would force process $3$ 
		to 
		violate the record.}
	\label{figure B_i example}
\end{figure}

Informally, in any set of views $\V'$ that explain a replay of $\R$, setting $\edge{ \opw{2}{j}, \opw{1}{i} } \in {V'_i}$ will create an $SCO(\V')$ edge $\edge{ \opw{2}{j}, \opw{1}{i} }$ which will conflict with $V'_k$.
The following theorem states that for every process $i$ it suffices to record all edges in $\trans{V}_i$, except those in $SCO_i( \V )$, $PO$, or $B_i( \V )$.

\begin{theorem} \label{theorem:RnRStrongCausalConsistencySufficiency}
	Consider a set of views $\V = \set{ V_i }_{i \in P}$ that explain a strongly causal consistent execution.
	For each process $i \in P$, let $R_i = \trans{V}_i \setminus \del{ SCO_i( \V ) \cupdot PO \cupdot B_i(\V) }$.
	Then, $\R = \set{R_i}_{i \in P}$ is a good record of $\V$.
\end{theorem}



The formal proof of the theorem is given in Appendix \ref{section proofs}.
We first show that the strong causal order and the $B_i$'s are preserved in the replay (Lemma \ref{lemma replay preserves SCO and B_i}).
The proof then proceeds by arguing that, for every process $i$, each path in $\trans{V}_i$ is reproduced correctly in the replay.
We refer the reader to Appendix \ref{section proofs} for the details.
The following theorem states that, for every process $i$, each edge in $\trans{V}_i \setminus \del{ SCO_i( \V ) \cupdot PO \cupdot B_i(\V) }$ is necessary for a good record under strong causal consistency.

\begin{theorem} \label{theorem:RnRStrongCausalConsistencyNecessity}
	Consider a set of views $\V = \set{ V_i }_{i \in P}$ that explain a strongly causal consistent execution.
	For any good record $\R = \set{R_i}_{i \in P}$ of $\V$, for any process $i \in P$ and any two operations $\op{1}{}, \op{2}{} \in ( *, i, *, * ) \cup ( w, *, *, * )$, if $\edge{ \op{1}{}, \op{2}{} } \in \trans{V}_i \setminus \del{ PO \cupdot SCO_i(\V) \cupdot B_i(\V) }$, then $\edge{ \op{1}{}, \op{2}{} } \in {R_i}$.
\end{theorem}

The formal proof of the theorem is presented in Appendix \ref{section proofs}.
We show that if any two operations $\op{1}{}, \op{2}{}$ are such that, for some process $i$, $\edge{ \op{1}{}, \op{2}{} } \in \trans{V}_i \setminus \del{ PO \cupdot SCO_i(\V) \cupdot B_i(\V) }$ but $\edge{ \op{1}{}, \op{2}{} }$ is not recorded, then we can swap the two operations during the replay without violating consistency or replay constraints.
This violates the definition of a good record.
Theorems \ref{theorem:RnRStrongCausalConsistencySufficiency} and \ref{theorem:RnRStrongCausalConsistencyNecessity} show that the record $\R = \set{R_i}_{i \in P}$ such that $R_i = \trans{V}_i \setminus \del{ SCO_i( \V ) \cupdot PO \cupdot B_i(\V) }$ is both sufficient and necessary for a correct replay under strong causal consistency.
		\subsection{Online Record for Strong Causal Consistency}
			\label{section strong causal consistency online}
We now look at the optimal record in an online setting.
Consider the following implementation of shared memory.
Each process keeps a copy of every shared variable in $X$.
Processes exchange messages to propagate their writes to shared variables.
Based on the received messages, each process updates the current value of its copy of the shared variables.
At any point in the execution, a read on variable $x$ at process $i$ returns the current value of $x$ stored at $i$.
We abstract this perspective of shared memory as follows.
Each process has a fixed set of read and write operations $(*, i, *, *)$ that it executes in their local order $PO(i)$ by communicating with the shared memory.
Executing an operation may take arbitrarily long and the process may spend arbitrarily long time to execute the next operation but each process only executes one operation at a time.
Via the shared memory, a process $i$ \emph{{observes}} its own 
operations and write operations from other processes one at a time.
The order in which these operations are observed give rise to the view $V_i$.
More formally, the execution proceeds in time steps.
At each time step in the execution, a unique\footnote{Uniqueness of the process makes the model simpler. If more than one process observes an operation at a given time step, we can separate this into multiple time steps ordered by the process identifiers.} process $i$ observes an operation from $(*, i, *, *) \cup (\operatorname{w}, *, *, *)$ and adds it to its view $V_i$.

The online record algorithm proceeds as follows.
Suppose process $i$ wants to record $\edge{ \op{1}{}, \op{2}{} } \in V_i$.
Then, process $i$ must record $\edge{ \op{1}{}, \op{2}{} }$ at the time when it observes $\op{2}{}$.
In the online setting, process $i$ has limited information about views of other processes at any given time in the execution.
How much does process $i$ know?
We assume that, at most, process $i$ has access to the history of other processes brought with the observed operation.
More precisely, at any time in the execution, if process $i$ is 
\emph{{aware}} that $\edge{ \op{1}{}, {\op{2}{}} } \in V_j$, for some 
process $j \ne i$, then process $i$ must have already observed $\op{3}{j} \in 
(*, j, *, *)$ such that $\op{1}{} <_{V_j} \op{2}{} \le_{V_j} \op{3}{j}$.
As discussed in Sections \ref{section introduction} and \ref{section RnR model}, the recording proceeds without information about the internal workings of the shared memory.
However, we assume that the RnR system is aware of the shared memory guarantees.
More precisely, for strong causal consistency, we assume that any process $i$ can check if $\edge{ \op{1}{}, {\op{2}{}} } \in SCO(\V)$ and also if $\edge{ \op{1}{}, {\op{2}{}} } \in PO$.
For a given execution $\V = \set{ V_i }_{i \in P}$, we say that a record $\R = 
\set{ R_i }_{ i \in P }$ is an \emph{{online record}} of $\V$ if $\R$ 
can be recorded in this manner.


Recall from Theorems $\ref{theorem:RnRStrongCausalConsistencySufficiency}$ and $\ref{theorem:RnRStrongCausalConsistencyNecessity}$ that for any process $i$, $R_i = \trans{V}_i \setminus \del{ SCO_i( \V ) \cupdot PO \cupdot B_i(\V) }$ is both sufficient and necessary in the offline setting.
Therefore, if the recording unit can detect, for an edge $\edge{ \op{1}{}, {\op{2}{}} } \in \trans{V}_i$, if it is one of $SCO_i( \V )$, $PO$, or $B_i( \V )$, then the optimal record in the online setting would match exactly the one in the offline scenario.
However, it turns out that the membership of $\edge{ \op{1}{}, {\op{2}{}} }$ in $B_i( \V )$ cannot be checked by the recording unit online.
This is formalized in Theorems \ref{theorem strong causal online sufficiency} and \ref{theorem strong causal online necessity} which state that for each process $i$, $R_i = \trans{V}_i \setminus \del{ SCO_i( \V ) \cupdot PO }$ is both sufficient and necessary in the online setting.
The formal proofs are presented in Appendix \ref{section proofs}.

\begin{theorem} \label{theorem strong causal online sufficiency}
	Consider a set of views $\V = \set{ V_i }_{i \in P}$ that explain a strongly causal consistent execution.
	For each process $i \in P$, let $R_i = \trans{V}_i \setminus \del{ SCO_i( \V ) \cupdot PO }$.
	Then, $\R = \set{R_i}_{i \in P}$ is a good \textbf{online} record of $\V$.
\end{theorem}

\begin{theorem} \label{theorem strong causal online necessity}
	Consider a set of views $\V = \set{ V_i }_{i \in P}$ that explain a strongly causal consistent execution.
	For any good \textbf{online} record $\R = \set{R_i}_{i \in P}$ of $\V$, for any process $i \in P$ and any two operations $\op{1}{}, \op{2}{} \in ( *, i, *, * ) \cup ( w, *, *, * )$, if $\edge{ \op{1}{}, \op{2}{} } \in \trans{V}_i \setminus PO \cupdot SCO_i(\V)$, then $\edge{ \op{1}{}, \op{2}{} } \in {R_i}$.
\end{theorem}

		\subsection{Causal Consistency} \label{section causal consistency}
Causal consistency (Definition \ref{definition causal consistency}) imposes less restrictions on views that can explain an execution as compared to strong causal consistency.
As discussed in Section \ref{section introduction}, we expect a smaller record for strong causal consistency than causal consistency.
Indeed consider a simple execution on two processes and two operations where process $1$ performs $\opw{}{1}$ and process $2$ performs $\opw{}{2}$.
Consider the set of views given in Figure \ref{figure causal vs strong} that explains this execution under both causal and strong causal consistency.
Under strong causal consistency, only process $1$ has to record $\edge{ 
\opw{}{2}, \opw{}{1} }$.
However, since causal consistency imposes no restrictions in this particular 
example, a good record for causal consistency will require process $2$ to 
record $\edge{ \opw{}{2}, \opw{}{1} }$ as well.

\begin{figure}
	\centering
	
	\begin{tikzpicture}
	\node at (-1, 6) {$V_1$:};
	\node at (-1, 5) {$V_2$:};
	
	\node (V1w2) at (0, 6) {$\opw{}{2}$};
	\node (V1w1) at (3, 6) {$\opw{}{1}$};
	
	\node (V2w2) at (0, 5) {$\opw{}{2}$};
	\node (V2w1) at (3, 5) {$\opw{}{1}$};
	
	\path[->]
	(V1w2) edge[color=red] node[above] {$R_1$} (V1w1)
	
	(V2w2) edge node[above] {$SCO_2(\V)$} (V2w1)
	;
	\end{tikzpicture}
	
	\vspace{\floatsep}
	\vspace{\floatsep}
	
	\begin{tikzpicture}
	\node at (5, 6) {$V'_1$:};
	\node at (5, 5) {$V'_2$:};
	
	\node (V'1w2) at (6, 6) {$\opw{}{2}$};
	\node (V'1w1) at (9, 6) {$\opw{}{1}$};
	
	\node (V'2w1) at (6, 5) {$\opw{}{1}$};
	\node (V'2w2) at (9, 5) {$\opw{}{2}$};
	
	\path[->]
	(V'1w2) edge[color=red] node[above] {$R_1$} (V'1w1)
	
	(V'2w1) edge (V'2w2)
	;
	\end{tikzpicture}
	
	\hrule
	\vspace{\floatsep}
	\caption{A simple example where the required record is smaller for strong 
	causal consistency. $\set[0]{ V'_1, V'_2 }$ can certify a replay under causal 
	consistency but not under strong causal consistency.}
	\label{figure causal vs strong}
\end{figure}

The question of what is the optimal record for causal consistency is still open.
We give a simple counterexample that shows that the natural strategy following the scheme of strong causal consistency does not work.
More concretely, consider a set of views $\V = \set{V_i}_{i \in P}$ that explain a causally consistent execution.
For each process $i$, let $R_i = \trans{V}_i \setminus \del{ WO \cupdot PO }$.
We give a simple four process example that shows that 
$\R = \set{R_i}_{i \in P}$ is not a good record of $\V$.
The program for this example is given in Figures \ref{figure causal consistency counter example original execution} and \ref{figure causal consistency counter example replay}.
Figure \ref{figure causal consistency counter example original execution} gives the writes-to relation in bold edges, for the original execution of the program, as well as a set of views $\V$ that explains the execution.
The red edges represent the recorded edges, as specified above.
Figure \ref{figure causal consistency counter example replay} gives one 
possible replay where the reads return the default values for the variables (so 
that the writes-to relation is empty), as well as a set of views $\V'$ that 
certifies the replay to be valid for the given record.

\begin{figure*}
	\centering
	\begin{tabular}{c|c}
		\begin{tikzpicture}
		\node at (0, 3) {Process 1:};
		\node at (0, 2) {Process 2:};
		\node at (0, 1) {Process 3:};
		\node at (0, 0) {Process 4:};
		
		\node (w1) at (1.5, 3) {$\opw{}{1}(x)$};
		\node (r2) at (1.5, 2) {$\opr{}{2}(x)$};
		\node (w2) at (4, 2) {$\opw{}{2}(x)$};
		\node (w3) at (1.5, 1) {$\opw{}{3}(y)$};
		\node (r4) at (1.5, 0) {$\opr{}{4}(y)$};
		\node (w4) at (4, 0) {$\opw{}{4}(y)$};
		
		\path[->]
		(r2) edge node[above] {$PO$} (w2)
		(r4) edge node[above] {$PO$} (w4)
		(w1) edge[line width=1.4pt] (r2)
		(w3) edge[line width=1.4pt] (r4)
		;
		\end{tikzpicture}
		&
		\begin{tikzpicture}
		\node at (-1, 6) {$V_1$:};
		\node at (-1, 5) {$V_2$:};
		\node at (-1, 4) {$V_3$:};
		\node at (-1, 3) {$V_4$:};
		
		\node (V1w1) at (0, 6) {$\opw{}{1}(x)$};
		\node (V1w3) at (2.25, 6) {$\opw{}{3}(y)$};
		\node (V1w4) at (4.5, 6) {$\opw{}{4}(y)$};
		\node (V1w2) at (6.75, 6) {$\opw{}{2}(x)$};
		
		\node (V2w1) at (0, 5) {$\opw{}{1}(x)$};
		\node (V2w3) at (2.25, 5) {$\opw{}{3}(y)$};
		\node (V2w4) at (4.5, 5) {$\opw{}{4}(y)$};
		\node (V2r2) at (6.75, 5) {$\opr{}{2}(x)$};
		\node (V2w2) at (9, 5) {$\opw{}{2}(x)$};
		
		\node (V3w3) at (0, 4) {$\opw{}{3}(y)$};
		\node (V3w1) at (2.25, 4) {$\opw{}{1}(x)$};
		\node (V3w2) at (4.5, 4) {$\opw{}{2}(x)$};
		\node (V3w4) at (6.75, 4) {$\opw{}{4}(y)$};
		
		\node (V4w3) at (0, 3) {$\opw{}{3}(y)$};
		\node (V4w1) at (2.25, 3) {$\opw{}{1}(x)$};
		\node (V4w2) at (4.5, 3) {$\opw{}{2}(x)$};
		\node (V4r4) at (6.75, 3) {$\opr{}{4}(y)$};
		\node (V4w4) at (9, 3) {$\opw{}{4}(y)$};
		\path[->]
		(V1w1) edge[color=red] node[above] {$R_1$} (V1w3)
		(V1w3) edge node[above] {$WO$} (V1w4)
		(V1w4) edge[color=red] node[above] {$R_1$} (V1w2)
		
		(V2w1) edge[color=red] node[above] {$R_2$} (V2w3)
		(V2w3) edge node[above] {$WO$} (V2w4)
		(V2w4) edge[color=red] node[above] {$R_2$} (V2r2)
		(V2r2) edge node[above] {$PO$} (V2w2)
		
		(V3w3) edge[color=red] node[above] {$R_3$} (V3w1)
		(V3w1) edge node[above] {$WO$} (V3w2)
		(V3w2) edge[color=red] node[above] {$R_3$} (V3w4)
		
		(V4w3) edge[color=red] node[above] {$R_4$} (V4w1)
		(V4w1) edge node[above] {$WO$} (V4w2)
		(V4w2) edge[color=red] node[above] {$R_4$} (V4r4)
		(V4r4) edge node[above] {$PO$} (V4w4)
		;
		\end{tikzpicture}
		\\
	\end{tabular}
	\hrule
	\vspace{\floatsep}
	\caption{A 4 process program where the bold edges represent the writes-to 
	relation for a possible execution. The set of views $\set{V_i}_{i=1}^{4}$ 
	explains this execution. The recorded edges are given in red.}
	\label{figure causal consistency counter example original execution}
\end{figure*}

\begin{figure*}
	\centering
	\begin{tabular}{c|c}
		\begin{tikzpicture}
		\node at (0, 3) {Process 1:};
		\node at (0, 2) {Process 2:};
		\node at (0, 1) {Process 3:};
		\node at (0, 0) {Process 4:};
		
		\node (w1) at (1.5, 3) {$\opw{}{1}(x)$};
		\node (r2) at (1.5, 2) {$\opr{}{2}(x)$};
		\node (w2) at (4, 2) {$\opw{}{2}(x)$};
		\node (w3) at (1.5, 1) {$\opw{}{3}(y)$};
		\node (r4) at (1.5, 0) {$\opr{}{4}(y)$};
		\node (w4) at (4, 0) {$\opw{}{4}(y)$};
		
		\path[->]
		(r2) edge node[above] {$PO$} (w2)
		(r4) edge node[above] {$PO$} (w4)
		;
		\end{tikzpicture}
		&
		\begin{tikzpicture}
		\node at (-1, 6) {$V'_1$:};
		\node at (-1, 5) {$V'_2$:};
		\node at (-1, 4) {$V'_3$:};
		\node at (-1, 3) {$V'_4$:};
		
		\node (V1w1) at (4.5, 6) {$\opw{}{1}(x)$};
		\node (V1w3) at (6.75, 6) {$\opw{}{3}(y)$};
		\node (V1w4) at (0, 6) {$\opw{}{4}(y)$};
		\node (V1w2) at (2.25, 6) {$\opw{}{2}(x)$};
		
		\node (V2w1) at (6.75, 5) {$\opw{}{1}(x)$};
		\node (V2w3) at (9, 5) {$\opw{}{3}(y)$};
		\node (V2w4) at (0, 5) {$\opw{}{4}(y)$};
		\node (V2r2) at (2.25, 5) {$\opr{}{2}(x)$};
		\node (V2w2) at (4.5, 5) {$\opw{}{2}(x)$};
		
		\node (V3w3) at (4.5, 4) {$\opw{}{3}(y)$};
		\node (V3w1) at (6.75, 4) {$\opw{}{1}(x)$};
		\node (V3w2) at (0, 4) {$\opw{}{2}(x)$};
		\node (V3w4) at (2.25, 4) {$\opw{}{4}(y)$};
		
		\node (V4w3) at (6.75, 3) {$\opw{}{3}(y)$};
		\node (V4w1) at (9, 3) {$\opw{}{1}(x)$};
		\node (V4w2) at (0, 3) {$\opw{}{2}(x)$};
		\node (V4r4) at (2.25, 3) {$\opr{}{4}(y)$};
		\node (V4w4) at (4.5, 3) {$\opw{}{4}(y)$};
		\path[->]
		(V1w1) edge[color=red] node[above] {$R_1$} (V1w3)
		(V1w2) edge node[above] {} (V1w1)
		(V1w4) edge[color=red] node[above] {$R_1$} (V1w2)
		
		(V2w1) edge[color=red] node[above] {$R_2$} (V2w3)
		(V2w2) edge node[above] {} (V2w1)
		(V2w4) edge[color=red] node[above] {$R_2$} (V2r2)
		(V2r2) edge node[above] {$PO$} (V2w2)
		
		(V3w3) edge[color=red] node[above] {$R_3$} (V3w1)
		(V3w4) edge node[above] {} (V3w3)
		(V3w2) edge[color=red] node[above] {$R_3$} (V3w4)
		
		(V4w3) edge[color=red] node[above] {$R_4$} (V4w1)
		(V4w4) edge node[above] {} (V4w3)
		(V4w2) edge[color=red] node[above] {$R_4$} (V4r4)
		(V4r4) edge node[above] {$PO$} (V4w4)
		;
		\end{tikzpicture}
		\\
	\end{tabular}
	\hrule
	\vspace{\floatsep}
	\caption{A possible replay of the execution in Figure \ref{figure causal 
			consistency counter example original execution} where the reads return 
			the 
		default values. The set of views $\set[0]{V'_i}_{i=1}^{4}$ certify that 
		this 
		replay is valid for the record from Figure \ref{figure causal consistency 
			counter example original execution}.}
	\label{figure causal consistency counter example replay}
\end{figure*}

Observe that $\V' \ne \V$.
There are two $WO$ edges $\edge{ \opw{}{1}, \opw{}{2} }$ and $\edge{ \opw{}{3}, \opw{}{4} }$ in the original execution while $WO'$, the write-read-write order for the replay, is empty.
Note that, in this example, not only do the views differ, but the reads return the wrong values in the replay as well.

The example replay execution is causally consistent, but it has the strange property that processes do not commit their writes locally before informing other processes.
For example, consider $\opw{}{2}$ and $\opw{}{4}$.
We have $\edge{ \opw{}{4}, \opw{}{2} } \in V_2$ but $\edge{ \opw{}{2}, \opw{}{4} } \in V_4$.
Both process $2$ and $4$ observed the other process's write before they saw their own; one of these processes distributed it's write to the other, then observed the other process's write, then committed it's own write.
This does not violate causality because neither process had read the other process' write (note, however, that this does violate strong causality).
Consider the setting where each process keeps a copy of each variable and the shared memory is implemented via message passing.
Then either process $2$ or process $4$ sends messages for its write before writing the local copy of the corresponding variable.
Such an execution would not be possible if each process always wrote to their local copy of the variable first and then sent the relevant messages to other processes.

	\section{Optimal Records for RnR Model 2} \label{section optimal record DRO}
		\subsection{Offline Record for Strong Causal Consistency}
			\label{section strong causal consistency DRO}
In this section we consider offline record for strong causal consistency.
In this case, as in Section \ref{section strong causal consistency}, the entire set of per-process views $\V = \set{V_i}_{i \in P}$ are made available to the RnR system which then determines the record that must be saved.
We define strong write order inductively as below.
It will be important for the optimal record for this RnR model.

\begin{definition} [Strong Write Order] \label{definition:strongWriteOrder}
	Given a set of views $\V = \set{ V_i }_{i \in P}$, two writes, $\opw{1}{} \in (\operatorname{w}, *, *, *)$ and $\opw{2}{i} \in ( \operatorname{w}, i, *, * )$, are ordered
	\begin{enumerate} [labelwidth=!, labelindent=0pt, topsep=0pt, itemsep=0pt]
		\item
		$\edge{\opw{1}{}, \opw{2}{i}} \in { SWO^1( \V ) }$ if $\edge{\opw{1}{}, \opw{2}{i}} \in { DRO(V_i) \cup \del{PO | (*, i, *, *) \cup (\operatorname{w}, *, *, *)} }$,

		\item
		$\edge{\opw{1}{}, \opw{2}{i}} \in { SWO^k( \V ) }$ if $\edge{\opw{1}{}, \opw{2}{i}} \in { DRO(V_i) \cup SWO^{k-1}( \V ) \cup \del{PO | (*, i, *, *) \cup (\operatorname{w}, *, *, *)} }$.
	\end{enumerate}
	We say that $\opw{1}{}$ and $\opw{2}{i}$ are ordered by \emph{strong write order}, $\edge{\opw{1}{}, \opw{2}{i}} \in { SWO( \V ) }$, if $\edge{\opw{1}{}, \opw{2}{i}} \in { SWO^k( \V ) }$ for some $k$.
	Furthermore, if $\edge{\opw{1}{}, \opw{2}{i}} \in { SWO( \V ) }$, then for every process $j \ne i$, we say that $\edge{\opw{1}{}, \opw{2}{i}} \in { SWO_j( \V ) }$.
\end{definition}

Note that for strongly causal consistent executions the strong write order is a subset of strong causal order.
Hence strong write order is a partial order for strongly causal consistent executions.
In contrast with RnR Model 1, we are only allowed to record $DRO$ edges.
Intuitively, $SWO$ captures those $SCO$ edges that can be used to influence the views of other processes under this model.
The base case captures those edges that will be forced on every process if process $i$ reproduces $DRO(V_i)$ faithfully.
The inductive case captures those edges that would be forced on every process if the previous level is forced and if process $i$ reproduces $DRO(V_i)$ faithfully.
However note that, in contrast with the RnR Model 1, $SWO$ may influence some relations than cannot be recorded.

The following definition will be useful in presenting the optimal records.

\begin{definition} \label{definition DRO Ai}
	Given a set of views $\V = \set{ V_i }_{i \in P}$, the relation $A_i(\V)$, for a process $i \in P$, is defined as $A_i(\V) = DRO(V_i) \cup SWO_i(\V) \cup \del{PO | (*, i, *, *) \cup (\operatorname{w}, *, *, *)}$.
	Furthermore $\A(\V) = \set{ A_i(\V) }_{i \in P}$.
\end{definition}

\begin{observation} \label{observation DRO Ai same as SWO}
	Consider a set of views $\V = \set{ V_i }_{i \in P}$ that explain a strongly causal execution and two writes, $\opw{1}{} \in (\operatorname{w}, *, *, *)$ and $\opw{2}{i} \in ( \operatorname{w}, i, *, * )$.
	Then $\edge{\opw{1}{}, \opw{2}{i}} \in A_i(\V)$ if and only if $\edge{\opw{1}{}, \opw{2}{i}} \in SWO( \V )$.
\end{observation}

Note that this implies that $A_i(\V) \supseteq SWO(\V)$, for all $i \in P$, as follows.
Each edge in $SWO(\V)$ is either a $SWO_i(\V)$ edge or a ${SWO(\V) \setminus SWO_i(\V)}$ edge.
Observation \ref{observation DRO Ai same as SWO} implies $\del{SWO(\V) \setminus SWO_i(\V)} \subseteq A_i(\V)$ and $SWO_i(\V) \subseteq A_i(\V)$ by Definition \ref{definition DRO Ai}.

\begin{proof}
	\begin{itemize}[labelwidth=!, labelindent=0pt, topsep=0pt, itemsep=0pt]
		\item[$\Rightarrow$]
		Suppose $\edge{\opw{1}{}, \opw{2}{i}} \in A_i(\V)$.
		Then $\edge{\opw{1}{}, \opw{2}{i}} \in SWO(\V)$ by Definition \ref{definition:strongWriteOrder}.

		\item[$\Leftarrow$]
		Suppose $\edge{\opw{1}{}, \opw{2}{i}} \in SWO^k(\V)$ for some $k > 0$.
		We proceed by induction on $k$.
		For the base case, we have that $\edge{\opw{1}{}, \opw{2}{i}} \in { DRO(V_i) \cup \del{PO | (*, i, *, *) \cup (\operatorname{w}, *, *, *)} }$ and so $\edge{\opw{1}{}, \opw{2}{i}} \in A_i(\V)$.
		For the inductive step, we have that $\edge{\opw{1}{}, \opw{2}{i}} \in DRO(V_i) \cup SWO^{k-1}(\V) \cup \del{PO | (*, i, *, *) \cup (\operatorname{w}, *, *, *)}$.
		Now $SWO^{k-1}(\V) = SWO^{k-1}_i(\V) \cupdot \del[1]{ SWO^{k-1}(\V) \setminus SWO^{k-1}_i(\V) }$.
		Observe that $\del[1]{ SWO^{k-1}(\V) \setminus SWO^{k-1}_i(\V) } \subseteq A_i(\V)$ by the inductive hypothesis.
		Furthermore $SWO^{k-1}_i(\V) \subseteq A_i(\V)$ by Definition \ref{definition DRO Ai}.
		Since $A_i(\V)$ is closed under transitivity, the result follows.
	\end{itemize}
\end{proof}

Similar to the record for RnR Model 1 in Section \ref{section strong causal consistency}, we wish to capture the effect of reordering two operations on the $SWO$ that violates the views of some other process.
More specifically, for two operations $\op{1}{} \in (*, *, *, *)$ and $\op{2}{} \in (*, *, *, *)$ such that $\op{1}{} \DRO{V_i} \op{2}{}$ for some process $i$, reordering them as $\op{2}{} \DRO{V_i} \op{1}{}$, may introduce some $SWO(\V)$ edges that violate some other process's view.
The following two definitions capture this notion.
\begin{definition} \label{definition DRO Ci}
	Given a set of views $\V = \set{ V_i }_{i \in P}$, a process $i \in P$, and two operations $\op{1}{} \in (*, *, *, *)$ and $\opw{2}{} \in (\operatorname{w}, *, *, *)$, the relation $C_i( \V, \op{1}{}, \opw{2}{} )$ is defined inductively as follows.
	\begin{enumerate} [labelwidth=!, labelindent=0pt, topsep=0pt, itemsep=0pt]
		\item
		Two write operations $\opw{3}{} \in (\operatorname{w}, *, *, *)$ and $\opw{4}{i} \in (\operatorname{w}, i, *, *)$ are ordered $\edge{\opw{3}{}, \opw{4}{i}} \in C_i^1( \V, \op{1}{}, \opw{2}{} )$ if
			\begin{enumerate} [labelwidth=!, labelindent=0pt, topsep=0pt, itemsep=0pt]
				\item
				$\op{1}{} \le_{A_i(\V)} \opw{4}{i}$, and

				\item
				$\opw{3}{} \le_{A_i(\V)} \opw{2}{}$.
			\end{enumerate}

		\item
		Two write operations $\opw{3}{} \in (\operatorname{w}, *, *, *)$ and $\opw{4}{i'} \in (\operatorname{w}, i', *, *)$ are ordered $\edge{\opw{3}{}, \opw{4}{i'}} \in C_i^k( \V, \op{1}{}, \opw{2}{} )$ if there exist two write operations $\opw{5}{} \in (\operatorname{w}, *, *, *)$ and $\opw{6}{} \in (\operatorname{w}, *, *, *)$ such that
		\begin{enumerate} [labelwidth=!, labelindent=0pt, topsep=0pt, itemsep=0pt]
			\item
			$\edge{\opw{5}{}, \opw{6}{}} \in C_i^{k-1}( \V, \op{1}{}, \opw{2}{} )$,

			\item
			$\opw{3}{} \le_{A_{i'}(\V) \cup C_i^{k-1}( \V, \op{1}{}, \opw{2}{} )} \opw{5}{}$, and

			\item
			$\opw{6}{} \le_{A_{i'}(\V)} \opw{4}{i'}$.
		\end{enumerate}
	\end{enumerate}
	Two write operations $\opw{3}{} \in (\operatorname{w}, *, *, *)$ and $\opw{4}{} \in (\operatorname{w}, *, *, *)$ are ordered $\edge{\opw{3}{}, \opw{4}{}} \in C_i( \V, \op{1}{}, \opw{2}{} )$ if $\edge{\opw{3}{}, \opw{4}{}} \in C^k_i( \V, \op{1}{}, \opw{2}{} )$ for some $k \ge 1$.
\end{definition}

\begin{definition} \label{definition DRO Bi}
	Given a set of views $\V = \set{ V_i }_{i \in P}$, the relation $B_i( \V )$, for a process $i \in P$, is defined as follows.
	Two operations on the same variable $x$, $\op{1}{} \in (*, *, x, *)$ and $\opw{2}{} \in (\operatorname{w}, *, x, *)$, are ordered $\edge{\op{1}{}, \opw{2}{}} \in {B_i(\V)}$ if
	\begin{enumerate} [labelwidth=!, labelindent=0pt, topsep=0pt, itemsep=0pt]
		\item
		$\edge{\op{1}{}, \opw{2}{}} \in DRO(V_i)$, and

		\item
		there exists a process $m \in P$ such that either
		\begin{enumerate} [labelwidth=!, labelindent=0pt, topsep=0pt, itemsep=0pt]
			\item $m \ne i$ and ${A_m(\V)} \cupdot C_i( \V, \op{1}{}, \opw{2}{} )$ has a cycle, or
			\item $m = i$ and $\del[1]{A_m(\V) \setminus \set[0]{ \edge{\op{1}{}, \opw{2}{}} } } \cupdot C_i( \V, \op{1}{}, \opw{2}{} )$ has a cycle.
		\end{enumerate}
	\end{enumerate}
\end{definition}

Informally, in any set of views $\V'$ that explain a replay of $\R$, setting $\edge{ \opw{2}{}, \op{1}{} } \in DRO(V'_i)$ will create a $SWO(\V')$ edge which will conflict with $A_m(\V')$.
The Appendix \ref{section observations} contains some useful observations which are needed in the proofs later on.
The following theorem states that for every process $i$ it suffices to record all edges in $\trans{A}_i(\V)$, except those in $SWO_i(\V)$, $PO$, or $B_i(\V)$.

\begin{theorem} \label{theorem DRO RnRStrongCausalConsistencySufficiency}
	Consider a set of views $\V = \set{ V_i }_{i \in P}$ that explain a strongly causal consistent execution.
	For each process $i \in P$, let $R_i = \trans{A}_i(\V) \setminus \del{ SWO_i(\V) \cupdot PO \cupdot B_i(\V)}$.
	Then, $\R = \set{R_i}_{i \in P}$ is a good record of $\V$.
\end{theorem}

The formal proof of the theorem is given in Appendix \ref{section proofs 2}.
It proceeds similarly to the proof of Theorem \ref{theorem:RnRStrongCausalConsistencySufficiency} but is significantly more complicated.
The following theorem states that, for every process $i$, each edge in $\trans{A}_i(\V) \setminus \del{ SWO_i[ \V ] \cupdot PO \cupdot B_i(\V) }$ is necessary for a good record under strong causal consistency.

\begin{theorem} \label{theorem DRO RnRStrongCausalConsistencyNecessity}
	Consider a set of views $\V = \{ V_i \}_{i \in P}$ that explain a strongly causal consistent execution.
	For any good record $\R = \set{R_i}_{i \in P}$ of $\V$, for any process $i \in P$ and any two operations $\op{1}{}, \op{2}{} \in ( *, i, *, * ) \cup ( w, *, *, * )$, if $\edge{ \op{1}{}, \op{2}{} } \in \trans{A}_i(\V) \setminus PO \cupdot SWO_i( \V ) \cupdot B_i(\V)$, then $\edge{ \op{1}{}, \op{2}{} } \in R_i$.
\end{theorem}

The formal proof of the theorem is presented in Appendix \ref{section proofs 2}.
We follow the same strategy as proof of Theorem \ref{theorem:RnRStrongCausalConsistencyNecessity} and show that if any two operations $\op{1}{}, \op{2}{}$ are such that, for some process $i$, $\edge{ \op{1}{}, \op{2}{} } \in \trans{A}_i \setminus PO \cupdot SWO_i(\V) \cupdot B_i(\V)$ but $\edge{ \op{1}{}, \op{2}{} }$ is not recorded, then we can swap the two operations during the replay without violating consistency or replay constraints.
This violates the definition of a good record.
Theorems \ref{theorem DRO RnRStrongCausalConsistencySufficiency} and \ref{theorem DRO RnRStrongCausalConsistencyNecessity} show that the record $\R = \set{R_i}_{i \in P}$ such that $R_i = \trans{A}_i \setminus \del{ SWO_i( \V ) \cupdot PO \cupdot B_i(\V) }$ is both sufficient and necessary for a correct replay under strong causal consistency.
		\subsection{Causal Consistency} \label{section causal consistency DRO}
The question of what is the optimal record for causal consistency is still open 
for RnR Model 2 as well.
Similar to Section \ref{section causal consistency}, we give a counterexample 
that shows that the natural strategy following 
the scheme of strong causal consistency does not work.
More concretely, consider a set of views $\V = \set{V_i}_{i \in P}$ that 
explain a causally consistent execution.
For each process $i$, let $A_i = DRO(V_i) \cup WO \cup \del{PO | (*, i, *, *) 
\cup (\operatorname{w}, *, *, *)}$ and
$R_i = \trans{A}_i \setminus \del{ WO \cupdot PO }$.
We give a simple four process example that shows that $\R = \set{R_i}_{i \in P}$ 
is not a good record of $\V$.
The program for this example is given in Figures
\ref{figure causal consistency DRO counter example original execution} and
\ref{figure causal consistency DRO counter example replay}.
Figure \ref{figure causal consistency DRO counter example original execution} also gives 
the writes-to relation in bold edges, for the original execution.
The writes-to relation is empty for the replay, as shown in Figure
\ref{figure causal consistency DRO counter example replay}.
Figure \ref{figure causal consistency DRO counter example original views} gives 
a set of views $\V$ that explains the original execution.
The red edges represent the recorded edges.
Figure \ref{figure causal consistency DRO counter example replay views} gives one 
possible replay where the reads return the default values for the variables (so 
that the writes-to relation is empty), as well as a set of views $\V'$ that 
certifies the replay to be valid for the given record.

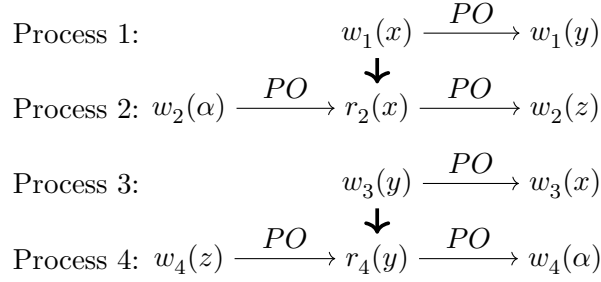
\begin{figure*}
	\centering
	\begin{tikzpicture}
	\node at (0, 3) {Process 1:};
	\node at (0, 2) {Process 2:};
	\node at (0, 1) {Process 3:};
	\node at (0, 0) {Process 4:};
	
	\node (w1x) at (4, 3) {$\opw{}{1}(x)$};
	\node (w1y) at (6.5, 3) {$\opw{}{1}(y)$};
	
	\node (w2a) at (1.5, 2) {$\opw{}{2}(\alpha)$};
	\node (r2x) at (4, 2) {$\opr{}{2}(x)$};
	\node (w2z) at (6.5, 2) {$\opw{}{2}(z)$};
	
	\node (w3y) at (4, 1) {$\opw{}{3}(y)$};
	\node (w3x) at (6.5, 1) {$\opw{}{3}(x)$};
	
	\node (w4z) at (1.5, 0) {$\opw{}{4}(z)$};
	\node (r4y) at (4, 0) {$\opr{}{4}(y)$};
	\node (w4a) at (6.5, 0) {$\opw{}{4}(\alpha)$};
	
	\path[->]
	(w1x) edge node[above] {$PO$} (w1y)
	
	(w2a) edge node[above] {$PO$} (r2x)
	(r2x) edge node[above] {$PO$} (w2z)
	
	(w3y) edge node[above] {$PO$} (w3x)
	
	(w4z) edge node[above] {$PO$} (r4y)
	(r4y) edge node[above] {$PO$} (w4a)
	
	(w1x) edge[line width=1.4pt] (r2x)
	(w3y) edge[line width=1.4pt] (r4y)
	;
	\end{tikzpicture}
	\hrule
	\vspace{\floatsep}
	\caption{A 4 process program where the bold edges represent the writes-to relation for a 
		possible execution.}
	\label{figure causal consistency DRO counter example original execution}
\end{figure*}

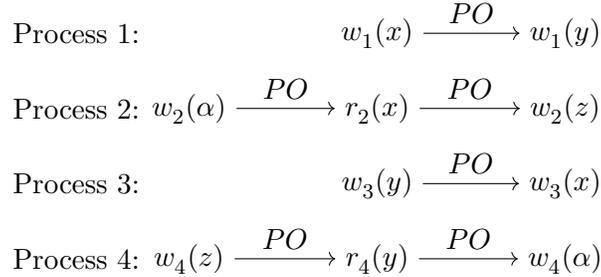
\begin{figure*}
\centering
\begin{tikzpicture}
\node at (0, 3) {Process 1:};
\node at (0, 2) {Process 2:};
\node at (0, 1) {Process 3:};
\node at (0, 0) {Process 4:};

\node (w1x) at (4, 3) {$\opw{}{1}(x)$};
\node (w1y) at (6.5, 3) {$\opw{}{1}(y)$};

\node (w2a) at (1.5, 2) {$\opw{}{2}(\alpha)$};
\node (r2x) at (4, 2) {$\opr{}{2}(x)$};
\node (w2z) at (6.5, 2) {$\opw{}{2}(z)$};

\node (w3y) at (4, 1) {$\opw{}{3}(y)$};
\node (w3x) at (6.5, 1) {$\opw{}{3}(x)$};

\node (w4z) at (1.5, 0) {$\opw{}{4}(z)$};
\node (r4y) at (4, 0) {$\opr{}{4}(y)$};
\node (w4a) at (6.5, 0) {$\opw{}{4}(\alpha)$};

\path[->]
(w1x) edge node[above] {$PO$} (w1y)

(w2a) edge node[above] {$PO$} (r2x)
(r2x) edge node[above] {$PO$} (w2z)

(w3y) edge node[above] {$PO$} (w3x)

(w4z) edge node[above] {$PO$} (r4y)
(r4y) edge node[above] {$PO$} (w4a)
;
\end{tikzpicture}
\hrule
\vspace{\floatsep}
\caption{
	A possible replay of the execution in Figure \ref{figure causal consistency 
	DRO 
		counter example original execution} where the reads return the default 
		values.
}
\label{figure causal consistency DRO counter example replay}
\end{figure*}

\begin{figure*}
	\centering
	\begin{tabular}{c}
		\begin{tikzpicture}
		\node at (1.5, 4) {$V_1$:};
		\node at (3, 4) {$\opw{}{1}(x)$};
		\node at (4.5, 4) {$\opw{}{1}(y)$};
		\node at (6, 4) {$\opw{}{3}(y)$};
		\node at (7.5, 4) {$\opw{}{4}(z)$};
		\node at (9, 4) {$\opw{}{4}(\alpha)$};
		\node at (10.5, 4) {$\opw{}{2}(\alpha)$};
		\node at (12, 4) {$\opw{}{2}(z)$};
		\node at (13.5, 4) {$\opw{}{3}(x)$};
		
		\node at (1.5, 1.5) {$\trans{A}_1(\V)$:};
		\node at (3, 3) {P1:};
		\node at (3, 2) {P2:};
		\node at (3, 1) {P3:};
		\node at (3, 0) {P4:};
		
		\node (w1x) at (4, 3) {$\opw{}{1}(x)$};
		\node (w1y) at (6, 3) {$\opw{}{1}(y)$};
		
		\node (w2a) at (14, 2) {$\opw{}{2}(\alpha)$};
		\node (w2z) at (16, 2) {$\opw{}{2}(z)$};
		
		\node (w3y) at (8, 1) {$\opw{}{3}(y)$};
		\node (w3x) at (12, 1) {$\opw{}{3}(x)$};
		
		\node (w4z) at (8, 0) {$\opw{}{4}(z)$};
		\node (w4a) at (12, 0) {$\opw{}{4}(\alpha)$};
		
		\path[->]
		(w1x) edge node[above] {$PO$} (w1y)
		(w1y) edge[color=red] node[above right] {$R_1$} (w3y)
		(w3y) edge node[above right] {$WO$} (w4a)
		(w4z) edge node[below] {$PO$} (w4a)
		(w4a) edge[color=red] node[below right] {$R_1$} (w2a)
		(w3y) edge[out=15,in=165] node[above] {$PO$} (w3x)
		(w2a) edge node[above] {$PO$} (w2z)
		;
		\end{tikzpicture}
		\\
		\hline
		\begin{tikzpicture}
		\node at (1.5, 4) {$V_2$:};
		\node at (3, 4) {$\opw{}{1}(x)$};
		\node at (4.5, 4) {$\opw{}{1}(y)$};
		\node at (6, 4) {$\opw{}{3}(y)$};
		\node at (7.5, 4) {$\opw{}{4}(z)$};
		\node at (9, 4) {$\opw{}{4}(\alpha)$};
		\node at (10.5, 4) {$\opw{}{2}(\alpha)$};
		\node at (12, 4) {$\opr{}{2}(x)$};
		\node at (13.5, 4) {$\opw{}{2}(z)$};
		\node at (15, 4) {$\opw{}{3}(x)$};
		
		\node at (1.5, 1.5) {$\trans{A}_2(\V)$:};
		\node at (3, 3) {P1:};
		\node at (3, 2) {P2:};
		\node at (3, 1) {P3:};
		\node at (3, 0) {P4:};
		
		\node (w1x) at (4, 3) {$\opw{}{1}(x)$};
		\node (w1y) at (6, 3) {$\opw{}{1}(y)$};
		
		\node (w2a) at (12, 2) {$\opw{}{2}(\alpha)$};
		\node (r2x) at (14, 2) {$\opr{}{2}(x)$};
		\node (w2z) at (16, 2) {$\opw{}{2}(z)$};
		
		\node (w3y) at (8, 1) {$\opw{}{3}(y)$};
		\node (w3x) at (16, 1) {$\opw{}{3}(x)$};
		
		\node (w4z) at (8, 0) {$\opw{}{4}(z)$};
		\node (w4a) at (10, 0) {$\opw{}{4}(\alpha)$};
		
		\path[->]
		(w1x) edge node[above] {$PO$} (w1y)
		(w1y) edge[color=red] node[above right] {$R_2$} (w3y)
		(w3y) edge node[above right] {$WO$} (w4a)
		(w4z) edge node[below] {$PO$} (w4a)
		(w4a) edge[color=red] node[below right] {$R_2$} (w2a)
		(r2x) edge[color=red] node[below left] {$R_2$} (w3x)
		(w2a) edge node[above] {$PO$} (r2x)
		(r2x) edge node[above] {$PO$} (w2z)
		;
		\end{tikzpicture}
		\\
		\hline
		\begin{tikzpicture}
		\node at (1.5, 4) {$V_3$:};
		\node at (3, 4) {$\opw{}{3}(y)$};
		\node at (4.5, 4) {$\opw{}{3}(x)$};
		\node at (6, 4) {$\opw{}{1}(x)$};
		\node at (7.5, 4) {$\opw{}{2}(\alpha)$};
		\node at (9, 4) {$\opw{}{2}(z)$};
		\node at (10.5, 4) {$\opw{}{4}(z)$};
		\node at (12, 4) {$\opw{}{4}(\alpha)$};
		\node at (13.5, 4) {$\opw{}{1}(y)$};
		
		\node at (1.5, 1.5) {$\trans{A}_3(\V)$:};
		\node at (3, 3) {P1:};
		\node at (3, 2) {P2:};
		\node at (3, 1) {P3:};
		\node at (3, 0) {P4:};
		
		\node (w1x) at (8, 3) {$\opw{}{1}(x)$};
		\node (w1y) at (16, 3) {$\opw{}{1}(y)$};
		
		\node (w2a) at (8, 2) {$\opw{}{2}(\alpha)$};
		\node (w2z) at (12, 2) {$\opw{}{2}(z)$};
		
		\node (w3y) at (4, 1) {$\opw{}{3}(y)$};
		\node (w3x) at (6, 1) {$\opw{}{3}(x)$};
		
		\node (w4z) at (14, 0) {$\opw{}{4}(z)$};
		\node (w4a) at (16, 0) {$\opw{}{4}(\alpha)$};
		
		\path[->]
		(w1x) edge node[above] {$PO$} (w1y)
		(w3x) edge[color=red] node[above left] {$R_3$} (w1x)
		(w1x) edge node[below left] {$WO$} (w2z)
		(w4z) edge node[above] {$PO$} (w4a)
		(w2z) edge[color=red] node[above right] {$R_3$} (w4z)
		(w2a) edge[out=-15, in=195] node[below] {$PO$} (w2z)
		(w3y) edge node[above] {$PO$} (w3x)
		;
		\end{tikzpicture}
		\\
		\hline
		\begin{tikzpicture}
		\node at (1.5, 4) {$V_4$:};
		\node at (3, 4) {$\opw{}{3}(y)$};
		\node at (4.5, 4) {$\opw{}{3}(x)$};
		\node at (6, 4) {$\opw{}{1}(x)$};
		\node at (7.5, 4) {$\opw{}{2}(\alpha)$};
		\node at (9, 4) {$\opw{}{2}(z)$};
		\node at (10.5, 4) {$\opw{}{4}(z)$};
		\node at (12, 4) {$\opr{}{4}(y)$};
		\node at (13.5, 4) {$\opw{}{4}(\alpha)$};
		\node at (15, 4) {$\opw{}{1}(y)$};
		
		\node at (1.5, 1.5) {$\trans{A}_4(\V)$:};
		\node at (3, 3) {P1:};
		\node at (3, 2) {P2:};
		\node at (3, 1) {P3:};
		\node at (3, 0) {P4:};
		
		\node (w1x) at (8, 3) {$\opw{}{1}(x)$};
		\node (w1y) at (16, 3) {$\opw{}{1}(y)$};
		
		\node (w2a) at (8, 2) {$\opw{}{2}(\alpha)$};
		\node (w2z) at (10, 2) {$\opw{}{2}(z)$};
		
		\node (w3y) at (4, 1) {$\opw{}{3}(y)$};
		\node (w3x) at (6, 1) {$\opw{}{3}(x)$};
		
		\node (w4z) at (12, 0) {$\opw{}{4}(z)$};
		\node (r4y) at (14, 0) {$\opr{}{4}(y)$};
		\node (w4a) at (16, 0) {$\opw{}{4}(\alpha)$};
		
		\path[->]
		(w1x) edge node[above] {$PO$} (w1y)
		(w3x) edge[color=red] node[above left] {$R_4$} (w1x)
		(w1x) edge node[above right] {$WO$} (w2z)
		(w4z) edge node[above] {$PO$} (r4y)
		(r4y) edge node[above] {$PO$} (w4a)
		(w2z) edge[color=red] node[above right] {$R_4$} (w4z)
		(w2a) edge node[below] {$PO$} (w2z)
		(w3y) edge node[above] {$PO$} (w3x)
		(r4y) edge[color=red] node[below right] {$R_4$} (w1y)
		;
		\end{tikzpicture}
		\\
	\end{tabular}
	\hrule
	\vspace{\floatsep}
	\caption{The set of views $\set{V_i}_{i=1}^{4}$ explains the execution in Figure 
		\ref{figure causal consistency DRO counter example original execution}.
		$\trans{A}_i(\V)$ for $i = 1, 2, 3, 4$ are also given with the recorded edges drawn in 
		red.}
	\label{figure causal consistency DRO counter example original views}
\end{figure*}

\begin{figure*}
\centering
\begin{tabular}{c}
	\begin{tikzpicture}
	\node at (1.5, 4) {$V'_1$:};
	\node at (3, 4) {$\opw{}{4}(z)$};
	\node at (4.5, 4) {$\opw{}{4}(\alpha)$};
	\node at (6, 4) {$\opw{}{2}(\alpha)$};
	\node at (7.5, 4) {$\opw{}{2}(z)$};
	\node at (9, 4) {$\opw{}{1}(x)$};
	\node at (10.5, 4) {$\opw{}{1}(y)$};
	\node at (12, 4) {$\opw{}{3}(y)$};
	\node at (13.5, 4) {$\opw{}{3}(x)$};
	
	\node at (1.5, 1.5) {$\trans{A}_1(\V')$:};
	\node at (3, 3) {P1:};
	\node at (3, 2) {P2:};
	\node at (3, 1) {P3:};
	\node at (3, 0) {P4:};
	
	\node (w1x) at (10, 3) {$\opw{}{1}(x)$};
	\node (w1y) at (12, 3) {$\opw{}{1}(y)$};
	
	\node (w2a) at (8, 2) {$\opw{}{2}(\alpha)$};
	\node (w2z) at (10, 2) {$\opw{}{2}(z)$};
	
	\node (w3y) at (14, 1) {$\opw{}{3}(y)$};
	\node (w3x) at (16, 1) {$\opw{}{3}(x)$};
	
	\node (w4z) at (4, 0) {$\opw{}{4}(z)$};
	\node (w4a) at (6, 0) {$\opw{}{4}(\alpha)$};
	
	\path[->]
	(w1x) edge node[above] {$PO$} (w1y)
	(w1y) edge[color=red] node[above right] {$R_1$} (w3y)
	(w4z) edge node[below] {$PO$} (w4a)
	(w4a) edge[color=red] node[below right] {$R_1$} (w2a)
	(w2a) edge node[above] {$PO$} (w2z)
	(w3y) edge node[above] {$PO$} (w3x)
	;
	\end{tikzpicture}
	\\
	\hline
	\begin{tikzpicture}
	\node at (1.5, 4) {$V'_2$:};
	\node at (3, 4) {$\opw{}{4}(z)$};
	\node at (4.5, 4) {$\opw{}{4}(\alpha)$};
	\node at (6, 4) {$\opw{}{2}(\alpha)$};
	\node at (7.5, 4) {$\opr{}{2}(x)$};
	\node at (9, 4) {$\opw{}{2}(z)$};
	\node at (10.5, 4) {$\opw{}{1}(x)$};
	\node at (12, 4) {$\opw{}{1}(y)$};
	\node at (13.5, 4) {$\opw{}{3}(y)$};
	\node at (15, 4) {$\opw{}{3}(x)$};
	
	\node at (1.5, 1.5) {$\trans{A}_2(\V')$:};
	\node at (3, 3) {P1:};
	\node at (3, 2) {P2:};
	\node at (3, 1) {P3:};
	\node at (3, 0) {P4:};
	
	\node (w1x) at (10, 3) {$\opw{}{1}(x)$};
	\node (w1y) at (12, 3) {$\opw{}{1}(y)$};
	
	\node (w2a) at (8, 2) {$\opw{}{2}(\alpha)$};
	\node (r2x) at (10, 2) {$\opr{}{2}(x)$};
	\node (w2z) at (12, 2) {$\opw{}{2}(z)$};
	
	\node (w3y) at (14, 1) {$\opw{}{3}(y)$};
	\node (w3x) at (16, 1) {$\opw{}{3}(x)$};
	
	\node (w4z) at (4, 0) {$\opw{}{4}(z)$};
	\node (w4a) at (6, 0) {$\opw{}{4}(\alpha)$};
	
	\path[->]
	(w1x) edge node[above] {$PO$} (w1y)
	(w1y) edge[color=red] node[above right] {$R_2$} (w3y)
	(w4z) edge node[below] {$PO$} (w4a)
	(w4a) edge[color=red] node[below right] {$R_2$} (w2a)
	(r2x) edge[color=red, out=-45, in=200] node[below] {$R_2$} (w3x)
	(w2a) edge node[above] {$PO$} (r2x)
	(r2x) edge node[above] {$PO$} (w2z)
	(w3y) edge node[above] {$PO$} (w3x)
	;
	\end{tikzpicture}
	\\
	\hline
	\begin{tikzpicture}
	\node at (1.5, 4) {$V'_3$:};
	\node at (3, 4) {$\opw{}{2}(\alpha)$};
	\node at (4.5, 4) {$\opw{}{2}(z)$};
	\node at (6, 4) {$\opw{}{4}(z)$};
	\node at (7.5, 4) {$\opw{}{4}(\alpha)$};
	\node at (9, 4) {$\opw{}{3}(y)$};
	\node at (10.5, 4) {$\opw{}{3}(x)$};
	\node at (12, 4) {$\opw{}{1}(x)$};
	\node at (13.5, 4) {$\opw{}{1}(y)$};
	
	\node at (1.5, 1.5) {$\trans{A}_3(\V')$:};
	\node at (3, 3) {P1:};
	\node at (3, 2) {P2:};
	\node at (3, 1) {P3:};
	\node at (3, 0) {P4:};
	
	\node (w1x) at (14, 3) {$\opw{}{1}(x)$};
	\node (w1y) at (16, 3) {$\opw{}{1}(y)$};
	
	\node (w2a) at (4, 2) {$\opw{}{2}(\alpha)$};
	\node (w2z) at (6, 2) {$\opw{}{2}(z)$};
	
	\node (w3y) at (10, 1) {$\opw{}{3}(y)$};
	\node (w3x) at (12, 1) {$\opw{}{3}(x)$};
	
	\node (w4z) at (8, 0) {$\opw{}{4}(z)$};
	\node (w4a) at (10, 0) {$\opw{}{4}(\alpha)$};
	
	\path[->]
	(w1x) edge node[below] {$PO$} (w1y)
	(w3x) edge[color=red] node[above left] {$R_3$} (w1x)
	(w4z) edge node[above] {$PO$} (w4a)
	(w2z) edge[color=red] node[above right] {$R_3$} (w4z)
	(w2a) edge node[below] {$PO$} (w2z)
	(w3y) edge node[above] {$PO$} (w3x)
	;
	\end{tikzpicture}
	\\
	\hline
	\begin{tikzpicture}
	\node at (1.5, 4) {$V'_4$:};
	\node at (1.5, 4) {$V'_3$:};
	\node at (3, 4) {$\opw{}{2}(\alpha)$};
	\node at (4.5, 4) {$\opw{}{2}(z)$};
	\node at (6, 4) {$\opw{}{4}(z)$};
	\node at (7.5, 4) {$\opr{}{4}(y)$};
	\node at (9, 4) {$\opw{}{4}(\alpha)$};
	\node at (10.5, 4) {$\opw{}{3}(y)$};
	\node at (12, 4) {$\opw{}{3}(x)$};
	\node at (13.5, 4) {$\opw{}{1}(x)$};
	\node at (15, 4) {$\opw{}{1}(y)$};
	
	\node at (1.5, 1.5) {$\trans{A}_4(\V')$:};
	\node at (3, 3) {P1:};
	\node at (3, 2) {P2:};
	\node at (3, 1) {P3:};
	\node at (3, 0) {P4:};
	
	\node (w1x) at (14, 3) {$\opw{}{1}(x)$};
	\node (w1y) at (16, 3) {$\opw{}{1}(y)$};
	
	\node (w2a) at (4, 2) {$\opw{}{2}(\alpha)$};
	\node (w2z) at (6, 2) {$\opw{}{2}(z)$};
	
	\node (w3y) at (8, 1) {$\opw{}{3}(y)$};
	\node (w3x) at (10, 1) {$\opw{}{3}(x)$};
	
	\node (w4z) at (8, 0) {$\opw{}{4}(z)$};
	\node (r4y) at (10, 0) {$\opr{}{4}(y)$};
	\node (w4a) at (12, 0) {$\opw{}{4}(\alpha)$};
	
	\path[->]
	(w1x) edge node[above] {$PO$} (w1y)
	(w3x) edge[color=red] node[above left] {$R_4$} (w1x)
	(w4z) edge node[above] {$PO$} (r4y)
	(r4y) edge node[below] {$PO$} (w4a)
	(w2z) edge[color=red] node[below left] {$R_4$} (w4z)
	(w2a) edge node[below] {$PO$} (w2z)
	(w3y) edge node[above] {$PO$} (w3x)
	(r4y) edge[color=red] node[below right] {$R_4$} (w1y)
	;
	\end{tikzpicture}
	\\
\end{tabular}
\hrule
\vspace{\floatsep}
\caption{
	The set of views $\set[0]{V'_i}_{i=1}^{4}$ certifies that the replay in Figure
	\ref{figure causal consistency DRO counter example replay} is valid for the 
	record from Figure \ref{figure causal consistency DRO counter example 
	original views}.
	$\trans{A}_i(\V')$ for $i = 1, 2, 3, 4$ are also given with the recorded 
	edges drawn 
	in red.}
\label{figure causal consistency DRO counter example replay views}
\end{figure*}
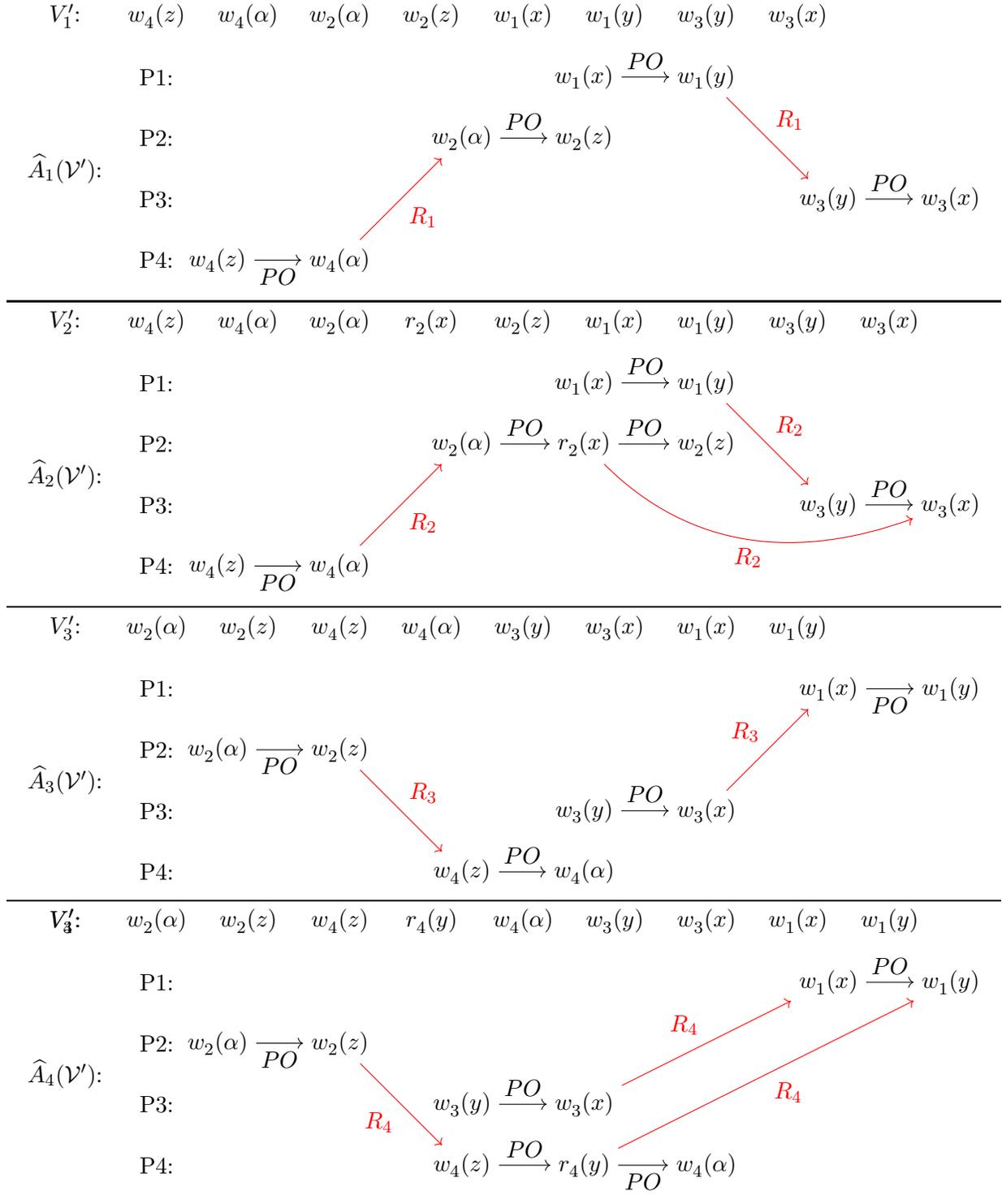

There are two $WO$ edges $\edge{ \opw{}{1}, \opw{}{2} }$ and $\edge{ \opw{}{3}, 
\opw{}{4} }$ in the original execution while $WO'$, the write-read-write order 
for the replay, is empty.
Note that, in this example, not only do the views differ, but the reads return 
the wrong values in the replay as well.

	\section{Discussion and Open Problems} \label{section discussion}
In this work we have looked at the optimal record for RnR under strong causal consistency, a strengthened version of causal consistency followed by practical implementations of causally consistent shared memory \cite{DeCandia:2007:DAH:1323293.1294281}, \cite{Ladin:1992:PHA:138873.138877}, \cite{Lloyd:2011:DSE:2043556.2043593}, \cite{Terry:1995:MUC:224056.224070}.
Table \ref{table strong causal summary} provides a summary of RnR results.
The optimal record for causal consistency is still an open problem.
In Section \ref{section causal consistency} and Section
\ref{section causal consistency DRO} we showed that a simple strategy following 
the scheme of 
strong causal consistency does not work for either RnR Model 1 or RnR Model 2.

As discussed in Section \ref{section introduction}, to the best of our knowledge, only one other work by Netzer \cite{Netzer1993Optimal} looks at optimal record for RnR.
However, Netzer considered sequential consistency and his setting 
is the same as RnR Model 2 where 
the objective is to record only data races so that all data races\footnote{Two operations form a data race if they are on the same variable and at least one of them is a write.} are resolved.
Another interesting setting is if the RnR system is allowed to record any edge in the views but the objective is to resolve all data races.
We have not yet looked at this setting, which we leave open to investigate in a future work.

We have not discussed how the record is enforced during replay.
For example, a simple strategy could be to simply wait for an operation until all its dependencies in the record have been observed.
This may not work with every record since the replay may be forced to choose between a record constraint and a consistency constraint.
We leave addressing this question to a future work.

Another problem of interest is to look at optimal RnR for weaker models.
Cache consistency is defined as sequential consistency on a per variable basis.

\begin{definition}
	An execution is \emph{cache consistent} if there exists a set of views $\V = \set{ V_{x} }_{x \in X}$ such that, for every variable $x$,
	\begin{itemize}[topsep=0pt, itemsep=0pt]
		\item $V_{x}$ is a view on the set of operations $( *, *, x, * )$, and
		\item $V_{x}$ respects $\del{{PO} | ( *, *, x, * )}$.
	\end{itemize}
	A shared memory $\Pi$ is \emph{cache consistent} if every execution run on $\Pi$ is cache consistent.
\end{definition}

For this definition, the optimal record follows from Netzer's result on sequential consistency \cite{Netzer1993Optimal}.
However, this assumes that per variable views are available to be recorded.
From the per process perspective, Steinke and Nutt \cite{Steinke2004Unified} have an alternate equivalent definition which sees cache consistency as providing per process views.
We refer the reader to Theorem B.8 in \cite{Steinke2004Unified} for this alternate definition of cache consistency.
What does the optimal record look like in this setting?
Cache consistency is implemented by virtually all commercial multiprocessors.

Cache consistency is incomparable to causal consistency.
What does the optimal record look like for a system that ensures both cache and causal consistency?
With the per process view of cache consistency it is easy to define cache+causal consistency by combining the restrictions on per process views.
In causal consistency views for two different processes may diverge so that after all operations have been observed, the two processes may have different values for the same shared variable.
Real world distributed systems provide some sort of conflict resolution on top of causal consistency to alleviate this problem \cite{DeCandia:2007:DAH:1323293.1294281}, \cite{Lloyd:2011:DSE:2043556.2043593}, \cite{Terry:1995:MUC:224056.224070}.
This results in ``eventual'' consistency where the different processes are eventually in agreement on the value of the shared variables, if all updates are stopped.
When this is implemented via a simple last writer wins rule, this is equivalent to all processes agreeing on the per variable ordering of write operations \cite{johnson1975maintenance}, i.e. cache consistency.

It would be interesting to experimentally evaluate how the theoretically optimum record performs on real systems, as opposed to the naive solution.
We leave that investigation open to a future work.

	\nocite{Jones2016Thesis}
	\bibliographystyle{abbrv}
	\bibliography{bib}

\begin{thebibliography}{10}

\bibitem{AdveHill}
S.~V. Adve and M.~D. Hill.
\newblock Weak ordering-a new definition.
\newblock In {\em [1990] Proceedings. The 17th Annual International Symposium
  on Computer Architecture}, pages 2--14, May 1990.

\bibitem{Ahamad1995}
M.~Ahamad, G.~Neiger, J.~E. Burns, P.~Kohli, and P.~W. Hutto.
\newblock Causal memory: definitions, implementation, and programming.
\newblock {\em Distributed Computing}, 9(1):37--49, Mar 1995.

\bibitem{Chen:2015:DRS:2830539.2790077}
Y.~Chen, S.~Zhang, Q.~Guo, L.~Li, R.~Wu, and T.~Chen.
\newblock Deterministic replay: A survey.
\newblock {\em ACM Comput. Surv.}, 48(2):17:1--17:47, Sept. 2015.

\bibitem{DeCandia:2007:DAH:1323293.1294281}
G.~DeCandia, D.~Hastorun, M.~Jampani, G.~Kakulapati, A.~Lakshman, A.~Pilchin,
  S.~Sivasubramanian, P.~Vosshall, and W.~Vogels.
\newblock Dynamo: Amazon's highly available key-value store.
\newblock {\em SIGOPS Oper. Syst. Rev.}, 41(6):205--220, Oct. 2007.

\bibitem{Dunlap}
G.~W. Dunlap, D.~G. Lucchetti, M.~A. Fetterman, and P.~M. Chen.
\newblock Execution replay of multiprocessor virtual machines.
\newblock In {\em Proceedings of the Fourth ACM SIGPLAN/SIGOPS International
  Conference on Virtual Execution Environments}, VEE '08, pages 121--130, New
  York, NY, USA, 2008. ACM.

\bibitem{Torrellas}
N.~Honarmand and J.~Torrellas.
\newblock Relaxreplay: Record and replay for relaxed-consistency
  multiprocessors.
\newblock {\em SIGARCH Comput. Archit. News}, 42(1):223--238, Feb. 2014.

\bibitem{johnson1975maintenance}
P.~Johnson and R.~Thomas.
\newblock Maintenance of duplicate databases.
\newblock 1975.

\bibitem{Jones2016Thesis}
R.~L. Jones.
\newblock Record and replay under relaxed consistency.
\newblock Master's thesis, University of Illinois at Urbana-Champaign, 2016.

\bibitem{Ladin:1992:PHA:138873.138877}
R.~Ladin, B.~Liskov, L.~Shrira, and S.~Ghemawat.
\newblock Providing high availability using lazy replication.
\newblock {\em ACM Trans. Comput. Syst.}, 10(4):360--391, Nov. 1992.

\bibitem{LamportSequential}
Lamport.
\newblock How to make a multiprocessor computer that correctly executes
  multiprocess programs.
\newblock {\em IEEE Transactions on Computers}, C-28(9):690--691, Sept 1979.

\bibitem{Lee:2010:REO:1736020.1736031}
D.~Lee, B.~Wester, K.~Veeraraghavan, S.~Narayanasamy, P.~M. Chen, and J.~Flinn.
\newblock Respec: Efficient online multiprocessor replay via speculation and
  external determinism.
\newblock In {\em Proceedings of the Fifteenth Edition of ASPLOS on
  Architectural Support for Programming Languages and Operating Systems},
  ASPLOS XV, pages 77--90, New York, NY, USA, 2010. ACM.

\bibitem{Lloyd:2011:DSE:2043556.2043593}
W.~Lloyd, M.~J. Freedman, M.~Kaminsky, and D.~G. Andersen.
\newblock Don't settle for eventual: Scalable causal consistency for wide-area
  storage with {COPS}.
\newblock In {\em Proceedings of the Twenty-Third ACM Symposium on Operating
  Systems Principles}, SOSP '11, pages 401--416, New York, NY, USA, 2011. ACM.

\bibitem{Mellor-Crummey:1989:DAL:916540}
J.~M. Mellor-Crummey.
\newblock {\em Debugging and Analysis of Large-scale Parallel Programs}.
\newblock PhD thesis, 1989.
\newblock AAI9019825.

\bibitem{Netzer1993Optimal}
R.~H.~B. Netzer.
\newblock Optimal tracing and replay for debugging shared-memory parallel
  programs.
\newblock {\em SIGPLAN Not.}, 28(12):1--11, Dec. 1993.

\bibitem{Steinke2004Unified}
R.~C. Steinke and G.~J. Nutt.
\newblock A unified theory of shared memory consistency.
\newblock {\em J. ACM}, 51(5):800--849, Sept. 2004.

\bibitem{Terry:1995:MUC:224056.224070}
D.~B. Terry, M.~M. Theimer, K.~Petersen, A.~J. Demers, M.~J. Spreitzer, and
  C.~H. Hauser.
\newblock Managing update conflicts in {B}ayou, a weakly connected replicated
  storage system.
\newblock In {\em Proceedings of the Fifteenth ACM Symposium on Operating
  Systems Principles}, SOSP '95, pages 172--182, New York, NY, USA, 1995. ACM.

\end{thebibliography}

	\appendix
	\section{Proofs for Section \ref{section optimal record}} \label{section proofs}
\begin{lemma} \label{lemma replay preserves SCO and B_i}
	Consider a set of views $\V = \set[0]{ V_i }_{i \in P}$ that explain a strongly causal consistent execution.
	For each process $i \in P$, let $R_i = \trans{V}_i \setminus \del{ SCO_i( \V ) \cupdot PO \cupdot B_i(\V) }$.
	Then, for any set of views $\V' = \set[0]{ V'_i }_{i \in P}$ that certify a strongly causal consistent replay to be valid under $\R = \set[0]{R_i}_{i \in P}$, we have that
	\begin{enumerate}[label=(\alph*),topsep=0pt]
		\item $SCO( \V' ) \supseteq SCO( \V )$, and
		\item $V'_i \supseteq B_i( {\V} )$ for every process $i \in P$.
	\end{enumerate}
\end{lemma}

\begin{proof_of}{Lemma \ref{lemma replay preserves SCO and B_i}(a)}
	Consider any arbitrary set of views $\V' = \set[0]{ V'_i }_{i \in P}$ that certify a strongly causal consistent replay to be valid for $\R$.
	We will call a write operation $\opw{}{}$ \emph{bad} if there exists a write operation ${\opw{}{}}'$ such that $\edge{ {\opw{}{}}', \opw{}{} } \in SCO( \V )$ but $\edge{ {\opw{}{}}', \opw{}{} } \not \in SCO( \V' )$.
	Recall from Definitions \ref{definition LRO strong causal order}-\ref{definition LRO strong causal consistency} that $SCO( {\V} )$ orders only write operations and is a partial order for strongly causal consistent executions.
	Consider any bad write operation, WLOG executed on process $1$, $\opw{2}{1} \in ( \operatorname{w}, 1, *, * )$, which is minimal with respect to $SCO( \V )$; i.e. for every write operation ${\opw{}{}}' \SCO{ {\V} } \opw{2}{1}$, we have that ${\opw{}{}}'$ is not bad.
	We proceed via contradiction.

	Since $\opw{2}{1}$ is a bad write operation, so there exists a write operation $\opw{1}{} \in ( \operatorname{w}, *, *, * )$ such that $\edge{ \opw{1}{}, \opw{2}{1} } \in SCO( \V )$ and $\edge{ \opw{1}{}, \opw{2}{1} } \not \in SCO( \V' )$.
	Consider a path $\rho$ from $\opw{1}{}$ to $\opw{2}{1}$ in $\trans{V}_1$ (such a path must exist since $\opw{1}{} \SCO{ {\V} } \opw{2}{1} \Rightarrow \opw{1}{} <_{V_1} \opw{2}{1}$) given by $\opw{1}{} = \op{\rho, 0}{} \lessdot_{V_1} \op{\rho, 1}{} \lessdot_{V_1} \op{\rho, 2}{} \lessdot_{V_1} \dots \lessdot_{V_1} \op{\rho, k}{} = \opw{2}{1}$.
	If $\edge{ \op{\rho, j}{}, \op{\rho, j+1}{} } \in {V'_1}$ for every $j \in \sbr{0, k-1}$, then $\edge{ \opw{1}{}, \opw{2}{1} } \in V'_1$ and so $\edge{ \opw{1}{}, \opw{2}{1} } \in SCO(\V')$ by Definition \ref{definition LRO strong causal order} which is a contradiction.
	So there exists a $j \in \sbr{0, k-1}$ such that $\edge{ \op{\rho, j}{}, \op{\rho, j+1}{} } \not \in V'_1$.
	
	Consider the smallest $j \in [0, k-1]$ such that $\edge{ \op{\rho, j}{}, \op{\rho, j+1}{} } \not \in V'_1$.
	Therefore $\opw{1}{} \le_{V'_1} \op{\rho, j}{}$.
	There are 4 cases to consider.
	\begin{enumerate}[label=\underline{Case \arabic*:},wide,topsep=0pt,labelindent=0pt]
		\item $\edge{ \op{\rho, j}{}, \op{\rho, j+1}{} } \in V_1 \setminus \del{ SCO_1(\V) \cupdot PO \cupdot B_1( \V ) }$.
		Then\\ $\edge{ \op{\rho, j}{}, \op{\rho, j+1}{} } \in R_1$ and $V'_1$ respects $R_1$ since $\V'$ certifies a replay to be valid for $\R$.
		Thus $\edge{ \op{\rho, j}{}, \op{\rho, j+1}{} } \in V'_1$, a contradiction.

		\item $\edge{ \op{\rho, j}{}, \op{\rho, j+1}{} } \in PO$.
		Then $V'_1$ respects $PO$ due to consistency and $PO$ is independent of executions.
		Thus $\edge{ \op{\rho, j}{}, \op{\rho, j+1}{} } \in V'_1$, a contradiction.

		\item $\edge{ \op{\rho, j}{}, \op{\rho, j+1}{} } \in SCO_1( {\V} )$.
		Then both $\op{\rho, j}{}$ and $\op{\rho, j+1}{}$ must be write operations.
		There are now two cases to consider.
		\begin{enumerate}[label=\underline{Case \roman*:},topsep=0pt,labelwidth=\widthof{Case ii:a},leftmargin=!,itemindent=\widthof{aaaaaaa}]
			\item
			$j < k-1$.
			Then $\op{\rho, j+1}{} \ne \opw{2}{1}$.
			Observe that $\edge{ \op{\rho, j+1}{}, \opw{2}{1} } \in V_1$ and so $\edge{ \op{\rho, j+1}{}, \opw{2}{1} } \in SCO(\V)$ by Definition \ref{definition LRO strong causal order}.
			Therefore, by the minimality of $\opw{2}{1}$, we have that $\op{\rho, j+1}{}$ is not a bad write.
			Thus $\edge{ \op{\rho, j}{}, \op{\rho, j+1}{} } \in SCO( {\V'} )$.
			Since $V'_1$ respects $SCO( {\V'} )$, therefore $\edge{ \op{\rho, j}{}, \op{\rho, j+1}{} } \in V'_1$, a contradiction.
			
			\item
			$j = k - 1$.
			So $\op{\rho, j+1}{} = \opw{2}{1}$ and $\edge{ \op{\rho, j}{}, \opw{2}{1} } \in SCO_1( {\V} )$.
			From Definition \ref{definition LRO Bi} we have that $\opw{2}{1}$ is not executed on process $1$, a contradiction to the initial assumption that $\opw{2}{1} \in (\operatorname{w}, 1, *, *)$.
		\end{enumerate}

		\item $\edge{ \op{\rho, j}{}, \op{\rho, j+1}{} } \in B_1( \V )$.
		Then by Definition \ref{definition LRO Bi}, $\op{\rho, j}{} \in (\operatorname{w}, 1, *, *)$ is a write operation on process $1$.
		Therefore $\edge{ \op{\rho, j}{}, \opw{2}{1} } \in PO$ and we get that $\opw{1}{} \le_{V'_1} \op{\rho, j}{} <_{V'_1} \opw{2}{1}$, a contradiction.
	\end{enumerate}
	In all cases, we get the desired contradiction.
\end{proof_of}
\\

\begin{proof_of}{Lemma \ref{lemma replay preserves SCO and B_i}(b)}
	Consider any arbitrary set of views $\V' = \set[0]{ V'_i }_{i \in P}$ that certify a strongly causal consistent replay to be valid for $\R$.
	We will call a write operation $\opw{1}{i} \in (\operatorname{w}, i, *, *)$, executed on a process $i$, \emph{bad} if there exists a write operation $\opw{2}{} \in (\operatorname{w}, *, *, *)$ such that $\edge{ \opw{1}{i}, \opw{2}{} } \in B_i( \V )$ but $\edge{ \opw{2}{}, \opw{1}{i} } \in V'_i$ (note that $B_i( \V )$ orders only write operations from Definition \ref{definition LRO Bi}).
	Recall that $SCO( \V )$ is a partial order for strongly causal consistent executions.
	Consider any bad write operation $\opw{1}{i} \in (\operatorname{w}, i, *, *)$ which is maximal with respect to $SCO( \V )$; i.e. for every write operation ${\opw{}{}}' >_{SCO(\V)} \opw{1}{i}$, we have that ${\opw{}{}}'$ is not bad.
	We proceed via contradiction.
	
	Since $\opw{1}{i}$ is a bad write operation, so there exists a write operation $\opw{2}{} \in (\operatorname{w}, *, *, *)$ such that $\edge{ \opw{1}{i}, \opw{2}{} } \in B_i( \V )$ and $\edge{ \opw{2}{}, \opw{1}{i} } \in V'_i$.
	Therefore, $\edge{ \opw{2}{}, \opw{1}{i} } \in SCO( \V' )$.
	By Definition \ref{definition LRO Bi}, there exists a process, WLOG process $1 \ne i$, such that $\edge{ \opw{1}{i}, \opw{2}{} } \in V_1$.
	If $\edge{ \opw{1}{i}, \opw{2}{} } \in V'_1$ then $V'_1$ does not respect $SCO( \V' )$, a contradiction since $\V'$ explains a strongly causal consistent execution.
	Therefore $\edge{ \opw{1}{i}, \opw{2}{} } \not \in V'_1$.
	Consider a path $\rho$ from $\opw{1}{i}$ to $\opw{2}{}$ in $\trans{V}_1$ given by $\opw{1}{i} = \op{\rho, 0}{} \lessdot_{V_1} \op{\rho, 1}{} \lessdot_{V_1} \op{\rho, 2}{} \lessdot_{V_1} \dots \lessdot_{V_1} \op{\rho, k}{} = \opw{2}{}$.
	If $\edge{ \op{\rho, j}{}, \op{\rho, j+1}{} } \in {V'_1}$ for every $j \in \sbr{0, k-1}$, then $\edge{ \opw{1}{i}, \opw{2}{} } \in V'_1$ which is a contradiction.
	So there exists a $j \in \sbr{0, k-1}$ such that $\edge{ \op{\rho, j}{}, \op{\rho, j+1}{} } \not \in V'_1$.
	
	Consider the smallest $j \in [0, k-1]$ such that $\edge{ \op{\rho, j}{}, \op{\rho, j+1}{} } \not \in V'_1$.
	Therefore $\opw{1}{i} \le_{V'_1} \op{\rho, j}{}$.
	There are 4 cases to consider.
	\begin{enumerate}[label=\underline{Case \arabic*:},wide,topsep=0pt,labelindent=0pt]
		\item $\edge{ \op{\rho, j}{}, \op{\rho, j+1}{} } \in V_1 \setminus \del{ SCO_1(\V) \cupdot PO \cupdot B_1( \V ) }$.
		Then \\$\edge{ \op{\rho, j}{}, \op{\rho, j+1}{} } \in R_1$ and $V'_1$ respects $R_1$ since $\V'$ certifies a replay to be valid for $\R$.
		Thus $\edge{ \op{\rho, j}{}, \op{\rho, j+1}{} } \in V'_1$, a contradiction.
		
		\item $\edge{ \op{\rho, j}{}, \op{\rho, j+1}{} } \in PO$.
		Then $V'_1$ respects $PO$ due to consistency and $PO$ is independent of executions.
		Thus $\edge{ \op{\rho, j}{}, \op{\rho, j+1}{} } \in V'_1$, a contradiction.
		
		\item $\edge{ \op{\rho, j}{}, \op{\rho, j+1}{} } \in SCO_1( {\V} )$.
		Then $V'_1$ respects $SCO( \V' )$ due to consistency and $SCO( \V' ) \supseteq SCO( \V )$ by Lemma \ref{lemma replay preserves SCO and B_i}(a).
		Thus $\edge{ \op{\rho, j}{}, \op{\rho, j+1}{} } \in V'_1$, a contradiction.
		
		\item $\edge{ \op{\rho, j}{}, \op{\rho, j+1}{} } \in B_1( \V )$.
		By Definition \ref{definition LRO Bi}, $\op{\rho, j}{} \in ( \operatorname{w}, 1, *, * )$ is a write operation on process $1$.
		Recall that $\opw{1}{i} \le_{V'_1} \op{\rho, j}{}$.
		If $\opw{1}{i} = \op{ \rho, j }{}$, then $i = 1$, which contradicts the initial assumption that $i \ne 1$.
		Thus $\opw{1}{i} \ne \op{\rho, j}{}$ and $\edge{ \opw{1}{i}, \op{\rho, j}{} } \in V'_1$ so that by Definition \ref{definition LRO strong causal order} $\edge{ \opw{1}{i}, \op{\rho, j}{} } \in SCO( \V' )$.
		Therefore, by the maximality of $\opw{1}{i}$, we have that $\op{\rho, j}{}$ is not a bad write.
		Thus $\edge{ \op{\rho, j}{}, \op{\rho, j+1}{} } \in V'_1$, a contradiction.
	\end{enumerate}
	In all cases, we get the desired contradiction.
	So we have that $\edge{ \opw{1}{i}, \opw{2}{} } \in V'_1$ but $\edge{ \opw{2}{}, \opw{1}{i} } \in SCO( \V' )$, which is a contradiction since $V'_1$ respects $SCO( \V' )$.
\end{proof_of}
\\

\begin{proof_of}{Theorem \ref{theorem:RnRStrongCausalConsistencySufficiency}}
	Consider any arbitrary set of views $\V' = \set{ V'_i }_{i \in P}$ that certify a strongly causal consistent replay to be valid for $\R$.
	We show that $\V' = \V$.
	More precisely, we show that for any process $i$ and any two operations $\op{1}{}, \op{2}{} \in ( *, *, *, * )$ such that $\edge{ \op{1}{}, \op{2}{} } \in V_i$ we must have $\edge{ \op{1}{}, \op{2}{} } \in V'_i$.
	Consider any arbitrary process $i$.
	We have that
	\begin{itemize}[topsep=0pt]
		\item $V'_i$ respects ${R_i}$, since $\V'$ certifies a replay to be valid for $\R$.
		\item $V'_i$ respects $SCO_i( \V ) \cup \del{PO \mid ( *, i, *, * ) \cup ( w, *, *, * )} \cup B_i( \V )$ due to consistency and Lemma \ref{lemma replay preserves SCO and B_i}.
	\end{itemize}
	Consider a $\op{1}{} \op{2}{}$-path $\rho$ in $\trans{V}_i$ given by $\op{1}{} = \op{\rho, 0}{} \lessdot_{V_i} \op{\rho, 1}{} \lessdot_{V_i} \op{\rho, 2}{} \lessdot_{V_i} \dots \lessdot_{V_i} \op{\rho, k}{} = \op{2}{}$.
	By construction of $\trans{V}_i$, each edge is either a $R_i$ edge or a ${PO}$ edge or a $SCO_i( {\V} )$ edge or a $B_i( \V )$ edge.
	Thus $\op{1}{} = \op{\rho, 0}{} <_{V'_i} \op{\rho, 1}{} <_{V'_i} \op{\rho, 2}{} <_{V'_i} \dots <_{V'_i} \op{\rho, k}{} = \op{2}{}$ and $\edge{ \op{1}{}, \op{2}{} } \in {V'_i}$, as required.
\end{proof_of}
\\

\begin{proof_of}{Theorem \ref{theorem:RnRStrongCausalConsistencyNecessity}}
	Assume for the sake of contradiction that there exists a good record $\R$ of $\V$, a process, WLOG process $1$, and two operations $\op{1}{}, \op{2}{} \in ( *, 1, *, * ) \cup ( \operatorname{w}, *, *, * )$ such that $\edge{ \op{1}{}, \op{2}{} } \in \trans{V}_i \setminus PO \cupdot SCO_i(\V) \cupdot B_i(\V)$ and $\edge{ \op{1}{}, \op{2}{} } \not \in {R_i}$.
	%
	Then, we construct a set of views $\V'$, that differs from $\V$, but certifies a strongly causal replay to be valid for $\R$, i.e. $\V'$ explains a strongly causal execution and extends the record $\R$.
	This violates the definition of a good record (see Section \ref{section RnR model}).
	We construct $\V'$ from $\V$ as follows.
	Let $V'_1 := \del[1]{ V_1 \setminus \set[0]{ \edge{ \op{1}{}, \op{2}{} } } } \cupdot \set[0]{ \edge{ \op{2}{}, \op{1}{} } } $.
	For each $i > 1$, set $V'_i = V_i$.
	There are two things to be shown:
	\begin{enumerate} [label=\arabic*),topsep=0pt]
		\item each $V'_i$ is a total order (so that it is indeed a view), and 
		\item $\V'$ certifies a strongly causal replay to be valid for $\R$, i.e., satisfies properties for both strong causal consistency and replay.
	\end{enumerate}

	We first show that for each $i \in P$, $V'_i$ is a total order.
	Since $V'_i = V_i$ for $i > 1$, we focus on $V'_1$.
	Suppose $V'_1$ is not a total order.
	$V'_1$ orders all operations in $(\operatorname{w}, *, *, *) \cup (*, 1, *, *)$ by construction.
	So we must have introduced a cycle in $V'_1$.
	This implies that there is a $\op{1}{} \op{2}{}$-path in $V_1 \setminus \set[0]{ \edge{ \op{1}{}, \op{2}{} } }$.
	Let this $\op{1}{} \op{2}{}$-path $\rho$ be given by $\op{1}{} = \op{\rho, 0}{} <_{V_i} \op{\rho, 1}{} <_{V_i} \dots  <_{V_i} \op{\rho, k}{} = \op{2}{}$.
	Note that since $\trans{V}_1$ preserves all paths in $V_1$, so there must be a $\op{\rho, j}{} \op{\rho, j + 1}{}$-path $P_j$ in $\trans{V}_1$ for every $j \in \sbr{0, k-1}$.
	Note also that these paths do not include the edge $\edge{ \op{1}{}, \op{2}{} }$ because $V_1$ is acyclic.
	So there is a $\op{1}{} \op{2}{}$-path in $\trans{V}_1$ that does not use the edge $\edge{ \op{1}{}, \op{2}{} }$ given by $\bigcup_{j = 0}^{k-1} P_j$.
	Hence, the edge $\edge{ \op{1}{}, \op{2}{} }$ can be removed from $\trans{V}_1$ while preserving all paths in $V_1$.
	This contradicts the fact that $\trans{V}_1$ is the (unique) transitive reduction of $V_1$.

	We now show that $\V'$ certifies a strongly causal replay to be valid for $\R$.
	More precisely, we show that, for each process $i$,
	\begin{enumerate}[label=\arabic*),topsep=0pt]
		\item $V'_i$ respects $R_i$, and
		\item $V'_i$ respects ${SCO( \V' )} \cup \del{{PO} | ( *, i, *, * ) \cup ( \operatorname{w}, *, *, * )}$.
	\end{enumerate}
	Observe that for each $i > 1$, $V'_i = V_i$ and so $V'_i$ respects $R_i \subseteq V_i$ and $PO \mid (*, i, *, *) \cup (\operatorname{w}, *, *, *)$.
	For $i = 1$, recall that $\edge{ \op{1}{}, \op{2}{} } \not \in R_1$ and $\edge{ \op{1}{}, \op{2}{} } \not \in PO$, both of which are independent of $\V'$.
	Therefore $V'_1$ respect $R_1$ and $PO \mid (*, 1, *, *) \cup (\operatorname{w}, *, *, *)$ as well.
	So it is left to show that each $V'_i$ respects $SCO( \V' )$.
	There are 4 cases to consider.
	\begin{enumerate}[label=\underline{Case \arabic*:},wide,topsep=0pt,labelindent=0pt]
		\item Either $\op{1}{} \in (\operatorname{r}, 1, *, *)$ or $\op{2}{} \in (\operatorname{r}, 1, *, *)$.
		Since strong causal order only orders write operations (Definition \ref{definition LRO strong causal order}) so $SCO( \V' ) = SCO( \V )$.
		Therefore, for each $i \in P$, $V'_i$ respects $SCO( \V' )$.

		\item $\op{2}{} \in (\operatorname{w}, 1, *, *)$ and $\op{1}{} \not \in (\operatorname{w}, 1, *, *)$.
		Then $\edge{ \op{1}{}, \op{2}{} } \in SCO( \V )$ and therefore $SCO( \V' ) = SCO( \V ) \setminus \set[0]{ \edge{ \op{1}{}, \op{2}{} } }$.
		Since $SCO( \V ) \supset SCO( \V' )$, therefore, for every $i > 1$, $V'_i$ respects $SCO( \V' )$.
		$V'_1$ respects $SCO( \V ) \setminus \set[0]{ \edge{ \op{1}{}, \op{2}{} } }$ by construction.

		\item $\op{1}{} \in (\operatorname{w}, 1, *, *)$ and $\op{2}{} \not \in (\operatorname{w}, 1, *, *)$.
		Then $SCO( \V' ) = SCO( \V ) \cup \set[0]{ \edge{ \op{2}{}, \op{1}{} } }$ by Definition \ref{definition LRO strong causal order}.
		WLOG $\op{2}{} \in (\operatorname{w}, 2, *, *)$.
		Since $\edge{ \op{1}{}, \op{2}{} } \not \in SCO_1( \V )$, therefore $\edge{ \op{2}{}, \op{1}{} } \in V_2$.
		We have that $V'_2 = V_2$ respects $SCO( \V ) \cup \set[0]{ \edge{ \op{2}{}, \op{1}{} } } = SCO( \V' )$.
		Since $\edge{ \op{1}{}, \op{2}{} } \not \in B_1( \V )$, therefore for all $i > 2$, $\edge{ \op{2}{}, \op{1}{} } \in V_i$ and so $V'_i = V_i$ respects $SCO( \V ) \cup \set[0]{ \edge{ \op{2}{}, \op{1}{} } } = SCO( \V' )$.
		Now $V'_1$ respects $SCO( \V ) \cup \set[0]{ \edge{ \op{2}{}, \op{1}{} } }$ by construction.

		\item $\op{1}{}$, $\op{2}{}$ are writes and $\op{1}{}, \op{2}{} \not \in (\operatorname{w}, 1, *, *)$.
		Then $SCO(\V) = SCO(\V')$ and for each $i \in P$, $V'_i$ respects $SCO(\V')$ since $\edge{ \op{1}{}, \op{2}{} } \not \in SCO(\V)$.
	\end{enumerate}

	So we have shown that $\V'$ certifies a strongly causal replay to be valid for $\R$.
	Since $\edge{ \op{1}{}, \op{2}{} } \in V_1$ and $\edge{ \op{2}{}, \op{1}{} } \in V'_1$, thus $V'_1 \ne V_1$.
	This contradicts the initial assumption that $\R$ is a good record.
	%
\end{proof_of}
\\

\begin{proof_of}{Theorem \ref{theorem strong causal online sufficiency}}
	By Theorem \ref{theorem:RnRStrongCausalConsistencySufficiency}, it follows that $\R$ is a good record of $\V$, so we show that $\R$ can be recorded online.
	Fix a process $i$ and consider an arbitrary time step in the execution when process $i$ observes an operation say $\op{2}{}$.
	Let $\op{1}{} \in (*, *, *, *)$ be the last operation in $V_i$.
	Process $i$ can check if $\edge{ \op{1}{}, \op{2}{} } \in PO$ and also if $\edge{ \op{1}{}, {\op{2}{}} } \in SCO(\V)$.
	To check if $\edge{ \op{1}{}, \op{2}{} } \in SCO_i(\V)$, process $i$ follows the following procedure.
	If $\op{2}{}$ was executed by process $i$, then the edge cannot be in $SCO_i(\V)$.
	If $\op{2}{}$ was not executed by process $i$, then $\edge{ \op{1}{}, \op{2}{} } \in SCO_i(\V)$ if and only if $\edge{ \op{1}{}, \op{2}{} } \in SCO(\V)$.
	Process $i$ records $\edge{ \op{1}{}, \op{2}{} }$ if $\edge{ \op{1}{}, \op{2}{} } \not \in SCO_i( \V ) \cupdot PO$.
	
	Observe that $\edge{ \op{1}{}, \op{2}{} } \in \trans{V}_i$ if and only if, when $\op{2}{}$ is observed by process $i$, $\op{1}{}$ is the last operation in $V_i$.
	Therefore, the above procedure records $\trans{V}_i \setminus \del{ SCO_i( \V ) \cupdot PO }$ at process $i$.
\end{proof_of}
\\

\begin{proof_of} {Theorem \ref{theorem strong causal online necessity}}
	By Theorem \ref{theorem:RnRStrongCausalConsistencyNecessity}, it follows that for any process $i$, $\trans{V}_i \setminus \del{ SCO_i( \V ) \cupdot PO \cupdot B_i( \V ) }$ is necessary to record even in the offline setting.
	We show that an arbitrary process, WLOG process $1$, can not detect if an edge in $\trans{V}_1 \setminus PO \cupdot SCO_1(\V)$ is also in $B_1( \V )$ in an online setting.
	Recall from Definition \ref{definition LRO Bi} that $B_1( \V )$ orders only write operations.
	Suppose, at a given time step in the execution, that process $1$ observes $\opw{}{2} \in (\operatorname{w}, 2, *, *)$, and $\opw{}{1} \in (\operatorname{w}, 1, *, *)$ is the last operation in $V_1$ so that $\edge{ \opw{}{1} , \opw{}{2} } \in B_1( \V ) \cap \trans{V}_1$.
	Assume further that $\edge{ \opw{}{1} , \opw{}{2} } \not \in PO \cupdot SCO_1(\V)$.
	
	Let $\set[0]{ \tilde{V}_i }_{i \ge 2}$ be the (parts of) views of processes $i \ge 2$ that process $1$ is aware of.
	Observe that for each $i \ge 2$, the last operation in $\tilde{V}_i$ was executed by process $i$.
	Let this last operation be $\opw{}{i}$.
	Note that each $\opw{}{i}$ has already been observed by process $1$.
	For $i > 2$, $\edge{ \opw{}{2}, \opw{}{i} } \not \in \tilde{V}_i$, since otherwise $\edge{ \opw{}{2}, \opw{}{i} } \in SCO_1( \V )$, which contradicts the fact that $\opw{}{2}$ is the last operation observed by process $1$.
	Similarly, for $i > 2$, $\edge{ \opw{}{1}, \opw{}{i} } \not \in \tilde{V}_i$.
	Therefore, as far as process $1$ is aware, no process $i > 2$ has observed either $\opw{}{1}$ or $\opw{}{2}$.
	Thus, for each $i > 2$, both $\tilde{V}_i \cup \set[0]{ \edge{ \opw{}{1}, \opw{}{2} } }$ and $\tilde{V}_i \cup \set[0]{ \edge{ \opw{}{2}, \opw{}{1} } }$ are valid for future observation by process $i$.
	So process $1$ cannot decide whether $\edge{ \opw{}{1} , \opw{}{2} } \in B_1( \V )$ or not (see Definition \ref{definition LRO Bi}).
\end{proof_of}

	\section{Some Observations for Section \ref{section optimal record DRO}} \label{section observations}
\begin{observation} \label{observation Ci minimal write}
	Consider a set of views $\V = \set{ V_i }_{i \in P}$ that explain a strongly causal consistent execution, an arbitrary process $i$, and two operations $\op{1}{} \in (*, *, *, *)$ and $\opw{2}{} \in (\operatorname{w}, *, *, *)$ such that $C_i(\V, \op{1}{}, \opw{2}{})$ is non-empty.
	Let $\opw{\operatorname{min}}{i} \in (\operatorname{w}, i, *, *)$ be the minimal (with respect to $PO$) write on process $i$ such that $\op{1}{} \le_{A_i(\V)} \opw{\operatorname{min}}{i}$.
	Then $\opw{\operatorname{min}}{i}$ exists and
	\begin{enumerate}
		\item $C^k_i(\V, \op{1}{}, \opw{2}{}) = C^k_i(\V, \opw{\operatorname{min}}{i}, \opw{2}{})$ for any $k$, and
		
		\item for any two write operations $\opw{3}{} \in (\operatorname{w}, *, *, *)$ and $\opw{4}{} \in (\operatorname{w}, *, *, *)$, if $\edge{\opw{3}{}, \opw{4}{}} \in C_i^1(\V, \op{1}{}, \opw{2}{})$, then $\edge{\opw{3}{}, \opw{\operatorname{min}}{i}} \in C_i^1(\V, \op{1}{}, \opw{2}{})$.
	\end{enumerate}
\end{observation}

\begin{proof}
	The existence of $\opw{\operatorname{min}}{i}$ follows from the assumption that $C_i(\V, \op{1}{}, \opw{2}{})$, and so $C^1_i(\V, \op{1}{}, \opw{2}{})$, is non-empty.
	Therefore, by Definition \ref{definition DRO Ci}, there exists at least one write $w_i \in (\operatorname{w}, i, *, *)$ on process $i$ such that $\op{1}{} \le_{A_i(\V)} w_i$.
	\begin{enumerate} 
	[label=\underline{\arabic*.},wide,topsep=0pt,labelindent=0pt]
		\item
		We proceed via induction on $k$.
		The inductive step for $k > 1$ follows from Definition \ref{definition DRO Ci} by applying the inductive hypothesis $C^{k-1}_i(\V, \op{1}{}, \opw{2}{}) = C^{k-1}_i(\V, \opw{\operatorname{min}}{i}, \opw{2}{})$.
		For the base case, we show the equality for $k = 1$.
		\begin{itemize}
			\item
			$C^1_i(\V, \op{1}{}, \opw{2}{}) \subseteq C^1_i(\V, \opw{\operatorname{min}}{i}, \opw{2}{})$.
			Consider any two operations $\opw{3}{} \in (\operatorname{w}, *, *, *)$ and $\opw{4}{i} \in (\operatorname{w}, i, *, *)$ such that $\edge{\opw{3}{}, \opw{4}{i}} \in C_i^1( \V, \op{1}{}, \opw{2}{} )$.
			Then,
			\begin{itemize}
				\item $\opw{\operatorname{min}}{i} \le_{PO} \opw{4}{i}$, by the minimality of $\opw{\operatorname{min}}{i}$, and
				\item $\opw{3}{} \le_{A_i(\V)} \opw{2}{}$, by Definition \ref{definition DRO Ci}.
			\end{itemize}
			Therefore, by Definition \ref{definition DRO Ci}, $\edge{\opw{3}{}, \opw{4}{i}} \in C_i^1( \V, \opw{\operatorname{min}}{i}, \opw{2}{} )$.
			
			\item
			$C^1_i(\V, \op{1}{}, \opw{2}{}) \supseteq C^1_i(\V, \opw{\operatorname{min}}{i}, \opw{2}{})$.
			Consider any two operations $\opw{3}{} \in (\operatorname{w}, *, *, *)$ and $\opw{4}{i} \in (\operatorname{w}, i, *, *)$ such that $\edge{\opw{3}{}, \opw{4}{}} \in C_i^1(\V, \opw{\operatorname{min}}{i}, \opw{2}{})$.
			Then,
			\begin{itemize}
				\item $\opw{3}{} \le_{A_i(\V)} \opw{2}{}$, by Definition \ref{definition DRO Ci},
				\item $\op{1}{} \le_{A_i(\V)} \opw{\operatorname{min}}{i}$, by the definition of $\opw{\operatorname{min}}{i}$, and
				\item $\opw{\operatorname{min}}{i} \le_{PO} \opw{4}{i}$, by the minimality of $\opw{\operatorname{min}}{i}$.
			\end{itemize}
			Therefore, $\op{1}{} \le_{A_i(\V)} \opw{4}{i}$, by Definition \ref{definition DRO Ai}, and so $\edge{\opw{3}{}, \opw{4}{i}} \in C_i^1(\V, \op{1}{}, \opw{2}{})$ by Definition \ref{definition DRO Ci}.
		\end{itemize}
		
		\item
		Consider any two operations $\opw{3}{} \in (\operatorname{w}, *, *, *)$ and $\opw{4}{i} \in (\operatorname{w}, i, *, *)$ such that $\edge{\opw{3}{}, \opw{4}{i}} \in C_i^1(\V, \op{1}{}, \opw{2}{})$.
		Then,
		\begin{itemize}
			\item $\op{1}{} \le_{A_i(\V)} \opw{\operatorname{min}}{i}$, by the definition of $\opw{\operatorname{min}}{i}$, and
			\item $\opw{3}{} \le_{A_i(\V)} \opw{2}{}$, by Definition \ref{definition DRO Ci}.
		\end{itemize}
		Therefore, by Definition \ref{definition DRO Ci}, $\edge{\opw{3}{}, \opw{\operatorname{min}}{i}} \in C_i^1(\V, \op{1}{}, \opw{2}{})$.
	\end{enumerate}
\end{proof}

\begin{observation} \label{observation Ci sub SWO}
	Consider a set of views $\V = \set{ V_i }_{i \in P}$ that explain a strongly causal consistent execution, an arbitrary process $i$, and two operations $\op{1}{} \in (*, *, *, *)$ and $\opw{2}{} \in (\operatorname{w}, *, *, *)$.
	We have that if $C_i^1(\V, \op{1}{}, \opw{2}{}) \subseteq SWO(\V)$, then
	\begin{enumerate}
		\item $C_i(\V, \op{1}{}, \opw{2}{}) \subseteq SWO(\V)$, and
		\item $\edge{\op{1}{}, \opw{2}{}} \not \in B_i(\V)$.
	\end{enumerate}
\end{observation}

\begin{proof}
	\begin{enumerate}[label=\underline{\arabic*.},wide,topsep=0pt,labelindent=0pt]
		\item
		By induction on $k$, we show that for every positive integer $k$, $C_i^k(\V, \op{1}{}, \opw{2}{}) \subseteq SWO(\V)$.
		The base case follows by assumption.
		For the inductive step, for $k > 1$, consider any two operations $\opw{3}{} \in (\operatorname{w}, *, *, *)$ and $\opw{4}{i'} \in (\operatorname{w}, i', *, *)$ such that $\edge{\opw{3}{}, \opw{4}{i'}} \in C_i^k(\V, \op{1}{}, \opw{2}{})$.
		Then, by Definition \ref{definition DRO Ci}, there exist two write operations $\opw{5}{}, \opw{6}{} \in (\operatorname{w}, *, *, *)$, such that
		\begin{enumerate}
			\item
			$\edge{\opw{5}{}, \opw{6}{}} \in C_i^{k-1}(\V, \op{1}{}, \opw{2}{})$,
			
			\item
			$\opw{3}{} \le_{A_{i'}(\V) \cup C_i^{k-1}(\V, \op{1}{}, \opw{2}{})} \opw{5}{}$, and
			
			\item
			$\opw{6}{} \le_{A_{i'}(\V)} \opw{4}{i'}$.
		\end{enumerate}
		By the inductive hypothesis, we have that $C_i^{k-1}(\V, \op{1}{}, \opw{2}{}) \subseteq SWO(\V)$.
		Therefore $\opw{3}{} \le_{A_{i'}(\V) \cup SWO(\V)} \opw{5}{} \SWO{\V} \opw{6}{} \le_{A_{i'}(\V)} \opw{4}{i'}$, which implies $\edge{ \opw{3}{}, \opw{4}{i'} } \in SWO(\V)$ by Observation \ref{observation DRO Ai same as SWO}.
		
		\item
		Since $C_i(\V, \op{1}{}, \opw{2}{}) \subseteq SWO(\V)$ and, for each process $m \in P$, $A_m(\V) \supseteq SWO(\V)$, thus
		\begin{enumerate}
			\item if $m \ne i$, then ${A_m(\V)} \cupdot C_i( \V, \op{1}{}, \opw{2}{} ) = A_m(\V)$ which is acyclic, and
			\item if $m = i$, then $\del[1]{A_m(\V) \setminus \set[0]{ \edge{\op{1}{}, \opw{2}{}} } } \cupdot C_i( \V, \op{1}{}, \opw{2}{} ) \subseteq A_m(\V)$ which is acyclic.
		\end{enumerate}
		Therefore $\edge{\op{1}{}, \opw{2}{}} \not \in B_i(\V)$ by Definition \ref{definition DRO Bi}.
	\end{enumerate}
\end{proof}

\begin{observation} \label{observation Ci SWO}
	Consider a set of views $\V = \set{ V_i }_{i \in P}$ that explain a strongly causal consistent execution, an arbitrary process $i$, and two write operations $\opw{1}{}, \opw{2}{}, \opw{3}{} \in (\operatorname{w}, *, *, *)$, and $\opw{4}{i'} \in (\operatorname{w}, i', *, *)$.
	We have that if $\edge{\opw{3}{}, \opw{4}{i'}} \in C_i^k( \V, \opw{1}{}, \opw{2}{} )$ for some $k \ge 0$, then $\opw{1}{} \le_{SWO(\V)} \opw{4}{i'}$.
\end{observation}

\begin{proof}
	We proceed by induction on $k$.
	Base case, for $k = 1$, we have $\opw{1}{} \le_{A_i(\V)} \opw{4}{i'}$ and $i' = i$, by Definition \ref{definition DRO Ci}.
	Therefore $\opw{1}{} \le_{SWO(\V)} \opw{4}{i'}$.
	For the inductive step, for $k > 1$, by Definition \ref{definition DRO Ci} there exist $\edge{\opw{5}{}, \opw{6}{}} \in C_i^{k-1}( \V, \opw{1}{}, \opw{2}{} )$ such that $\opw{6}{} \le_{A_{i'}(\V)} \opw{4}{i'}$.
	By the inductive hypothesis, we have that $\opw{1}{} \le_{SWO(\V)} \opw{6}{}$.
	Therefore $\opw{1}{} \le_{SWO(\V)} \opw{6}{} \le_{A_{i'}(\V)} \opw{4}{i'}$, 
	and so $\opw{1}{} \le_{SWO(\V)} \opw{4}{i'}$.
\end{proof}

	\section{Proofs for Section \ref{section optimal record DRO}} \label{section proofs 2}
\begin{lemma} \label{lemma replay preserves SWO and B_i}
	Consider a set of views $\V = \set{ V_i }_{i \in P}$ that explain a strongly causal consistent execution.
	For each process $i \in P$, let $R_i = \hat{A}_i(\V) \setminus \del{ SWO_i(\V) \cupdot PO \cupdot B_i(\V)}$.
	Then, for any set of views $\V'$ that certify a strongly causal consistent replay to be valid under $\R = \set{R_i}_{i \in P}$, we have that
	\begin{enumerate}[label=(\alph*),topsep=0pt]
		\item $SWO(\V') \supseteq SWO(\V)$, and
		\item $V'_i \supseteq B_i(\V)$ for every process $i \in P$.
	\end{enumerate}
\end{lemma}

\begin{proof_of}{Lemma \ref{lemma replay preserves SWO and B_i}(a)}
	Consider any arbitrary set of views $\V' = \set{ V'_i }_{i \in P}$ that certify a strongly causal consistent replay to be valid for $\R$.
	We will call a write operation $\opw{}{}$ \emph{bad} if there exists a write operation ${\opw{}{}}'$ such that $\edge{ {\opw{}{}}', \opw{}{} } \in SWO(\V)$ but $\edge{ {\opw{}{}}', \opw{}{} } \not \in SWO(\V)$.
	Recall from Definition \ref{definition:strongWriteOrder} that $SWO(\V)$ orders only write operations and is a partial order for strongly causal consistent executions.
	Consider any bad write operation, WLOG executed on process $1$, $\opw{2}{1} \in (\operatorname{w}, 1, *, *)$, which is minimal with respect to $SWO(\V)$; i.e. for every write operation ${\opw{}{}}' \SWO{ {\V} } \opw{2}{1}$, we have that ${\opw{}{}}'$ is not bad.
	We proceed via contradiction.
	
	Since $\opw{2}{1}$ is a bad write operation, so there exists a write operation $\opw{1}{}$ such that $\edge{ \opw{1}{}, \opw{2}{1} } \in SWO(\V)$ and $\edge{ \opw{1}{}, \opw{2}{1} } \not \in SWO(\V')$.
	Consider a path $\rho$ from $\opw{1}{}$ to $\opw{2}{1}$ in $\hat{A}_1(\V)$ (such a path must exist since $\opw{1}{} \SWO{ {\V} } \opw{2}{1} \Rightarrow \opw{1}{} <_{A_1(\V)} \opw{2}{1}$ by Observation \ref{observation DRO Ai same as SWO}) given by $\opw{1}{} = \op{\rho, 0}{} \lessdot_{A_1(\V)} \op{\rho, 1}{} \lessdot_{A_1(\V)} \op{\rho, 2}{} \lessdot_{A_1(\V)} \dots \lessdot_{A_1(\V)} \op{\rho, k}{} = \opw{2}{1}$.
	Note that each operation in the path is in the view ${V_1}$ and hence in $V_1'$.
	If $\edge{ \op{\rho, j}{}, \op{\rho, j+1}{} } \in {A_1(\V')}$ for every $j \in \sbr{0, k-1}$, then $\edge{ \opw{1}{}, \opw{2}{1} } \in A_1(\V')$ and so $\edge{ \opw{1}{}, \opw{2}{1} } \in SWO(\V')$ by Observation \ref{observation DRO Ai same as SWO} which is a contradiction.
	So there exists a $j \in \sbr{0, k-1}$ such that $\edge{ \op{\rho, j}{}, \op{\rho, j+1}{} } \not \in {A_1(\V')}$.
	
	Consider the smallest $j \in \sbr{0, k-1}$ such that $\edge{ \op{\rho, j}{}, \op{\rho, j+1}{} } \not \in {A_1(\V')}$.
	Therefore $\opw{1}{} \le_{A_1(\V')} \op{\rho, j}{}$.
	There are $4$ cases to consider.
	\begin{enumerate}
		[label=\underline{Case \arabic*:},wide,topsep=0pt,labelindent=0pt]
		\item $\edge{ \op{\rho, j}{}, \op{\rho, j+1}{} } \in \hat{A}_1(\V) 
		\setminus \del{ SWO(\V) \cupdot PO \cupdot B_1(\V) }$.
		Then $\edge{ \op{\rho, j}{}, \op{\rho, j+1}{} } \in R_1$ and $V'_1$ respects $R_1$ since $\V'$ is a replay of $\R$.
		Thus $\edge{ \op{\rho, j}{}, \op{\rho, j+1}{} } \in DRO(V'_1)$ and so $\edge{ \op{\rho, j}{}, \op{\rho, j+1}{} } \in {A_1(\V')}$, a contradiction.
		
		\item $\edge{ \op{\rho, j}{}, \op{\rho, j+1}{} } \in PO$.
		Then $V'_1$ respects $PO$ due to consistency and $PO$ is independent of executions.
		Thus $\edge{ \op{\rho, j}{}, \op{\rho, j+1}{} } \in {A_1(\V')}$, a contradiction.
		
		\item $\edge{ \op{\rho, j}{}, \op{\rho, j+1}{} } \in SWO_1(\V)$.
		Then both $\op{\rho, j}{}$ and $\op{\rho, j+1}{}$ must be write operations.
		There are now two cases to consider.
		\begin{enumerate}
			[label=\underline{Case \roman*:},topsep=0pt,labelwidth=\widthof{Case 
			ii:a},leftmargin=!,itemindent=\widthof{aaaaaaa}]
			\item
			$j < k-1$.
			Then $\op{\rho, j+1}{} \ne \opw{2}{1}$.
			Observe that $\edge{ \op{\rho, j+1}{}, \opw{2}{1} } \in A_1(\V)$ and so $\edge{ \op{\rho, j+1}{}, \opw{2}{1} } \in SWO(\V)$.
			Therefore, by the minimality of $\opw{2}{1}$, we have that $\op{\rho, j+1}{}$ is not a bad write.
			Thus $\edge{ \op{\rho, j}{}, \op{\rho, j+1}{} } \in SWO(\V')$ and so $\edge{ \op{\rho, j}{}, \op{\rho, j+1}{} } \in {A_1(\V')}$, a contradiction.
			
			\item
			$j = k - 1$.
			So $\op{\rho, j+1}{} = \opw{2}{1}$ and $\edge{ \op{\rho, j}{}, \opw{2}{1} } \in  SWO_{1}({\V})$.
			From Definition \ref{definition:strongWriteOrder} we have that $\opw{2}{1}$ is not executed on process 1, a contradiction to the initial assumption that $\opw{2}{1} \in (*, 1, *, *)$.
		\end{enumerate}
		
		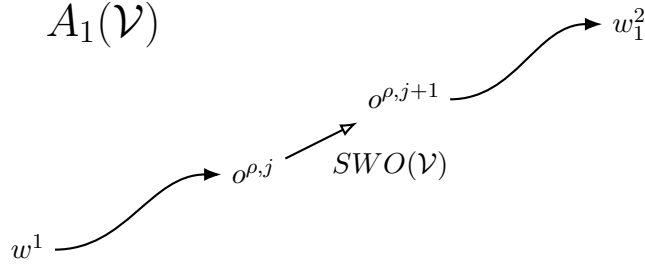
\begin{figure}
			\centering
			{
				\begin{tikzpicture}
				\node at (1, 3) {\LARGE $A_1(\V)$};

				\node (w1) at (0, 0) {$\opw{1}{}$};
				\node (opj) at (3, 1) {$\op{\rho, j}{}$};
				\node (opj+1) at (5, 2) {$\op{\rho, j+1}{}$};
				\node (w21) at (8, 3) {$\opw{2}{1}$};

				\draw[-{Latex}, thick] (w1) to[out=0, in=180] (opj);
				\draw[-{Latex[open]}, thick] (opj) to node[below right] {$SWO(\V)$} (opj+1);
				\draw[-{Latex}, thick] (opj+1) to[out=0, in=180] (w21);
				\end{tikzpicture}
			}
			\caption{Proof of Lemma \ref{lemma replay preserves SWO and B_i}(a) Case 3}
		\end{figure}
		
		\item $\edge{ \op{\rho, j}{}, \op{\rho, j+1}{} } \in B_1(\V)$.
		Then by Definition \ref{definition DRO Bi} we have that $\edge{ \op{\rho, j}{}, \op{\rho, j+1}{} } \in DRO(V_1)$ and that $C_1( \V, \op{\rho, j}{}, \op{\rho, j+1}{} )$ is non-empty.
		Thus, from Observation \ref{observation Ci minimal write}, there exists $\opw{\operatorname{min}}{1} \in (\operatorname{w}, 1, *, *)$ such that $C_1^1( \V, \op{\rho, j}{}, \op{\rho, j+1}{} ) = C_1^1( \V, \opw{\operatorname{min}}{1}, \op{\rho, j+1}{} )$.
		Since $\opw{2}{1} \in (\operatorname{w}, 1, *, *)$, therefore either $\opw{2}{1} \le_{PO} \opw{\operatorname{min}}{1}$ or $\opw{\operatorname{min}}{1} \PO \opw{2}{1}$.
		We consider both cases.
		\begin{enumerate}
			[label=\underline{Case \roman*:},topsep=0pt,labelwidth=\widthof{Case 
			ii:a},leftmargin=!,itemindent=\widthof{aaaaaaa}]
			\item
			$\opw{2}{1} \le_{PO} \opw{\operatorname{min}}{1}$.
			We show that $C^1_1( \V, \op{\rho, j}{}, \op{\rho, j+1}{} ) \subseteq SWO(\V)$ so that, by Observation \ref{observation Ci sub SWO}, $\edge{ \op{\rho, j}{}, \op{\rho, j+1}{} } \not \in B_1(\V)$ which is a contradiction.
			Consider any two writes $\opw{3}{} \in (\operatorname{w}, *, *, *)$ and $\opw{4}{1} \in (\operatorname{w}, 1, *, *)$ such that $\edge{\opw{3}{}, \opw{4}{1}} \in C_1^1( \V, \op{\rho, j}{}, \op{\rho, j+1}{} )$.
			Then
			\begin{itemize}
				\item $\opw{3}{} \le_{A_1(\V)} \op{\rho, j+1}{}$, by Definition \ref{definition DRO Ci},
				\item $\op{\rho, j+1}{} \le_{A_1(\V)} \opw{2}{1}$, since $\op{\rho, j+1}{}$ is on a $\opw{1}{} \opw{2}{1}$-path in $A_1(\V)$,
				\item $\opw{2}{1} \le_{PO} \opw{\operatorname{min}}{1} \le_{PO} \opw{4}{1}$, by assumption and the minimality of $\opw{\operatorname{min}}{1}$, and
				\item $\opw{3}{} \ne \opw{4}{1}$, by Definition \ref{definition DRO Ci}.
			\end{itemize}
			Therefore, we get that $\edge{ \opw{3}{}, \opw{4}{1} } \in {A_1(\V)}$.
			So $\edge{ \opw{3}{}, \opw{4}{1} } \in {SWO(\V)}$, as required.
			
			\begin{figure}
				\centering
				{
					\begin{tikzpicture}
					\node at (1, 3) {\LARGE $A_1(\V)$};
					
					\node (w1) at (0, 0) {$\opw{1}{}$};
					\node (opj) at (3, 1) {$\op{\rho, j}{}$};
					\node (opj+1) at (5, 2) {$\op{\rho, j+1}{}$};
					\node (w21) at (8, 3) {$\opw{2}{1}$};
					\node (wmin) at (10, 3) {$\opw{\operatorname{min}}{1}$};
					\node (w41) at (12, 3) {$\opw{4}{1}$};
					\node (w3) at (0, 2) {$\opw{3}{}$};
					
					\draw[-{Latex}, thick] (w1) to[out=0, in=180] (opj);
					\draw[-{Latex[open]}, thick] (opj) to node[above left] {$B_1(\V)$} (opj+1);
					\draw[-{Latex}, thick] (opj) to[out=330, in=220] (wmin);
					\draw[-{Latex}, thick] (w21) to[out=10, in=170] node[above] {$PO$} (wmin);
					\draw[-{Latex}, thick] (opj+1) to[out=0, in=200] (w21);
					\draw[-{Latex}, thick] (wmin) to node[above] {$PO$} (w41);
					\draw[-{Latex}, thick] (w3) to[out=0, in=150] (opj+1);
					\end{tikzpicture}
				}
				\caption{Proof of Lemma \ref{lemma replay preserves SWO and B_i}(a) Case 4 (i)}
			\end{figure}
			
			\item
			$\opw{\operatorname{min}}{1} \PO \opw{2}{1}$.
			Observe that $\opw{\operatorname{min}}{1} <_{A_1(\V)} \opw{2}{1}$ and therefore $\opw{\operatorname{min}}{1} \SWO{\V} \opw{2}{1}$.
			By the minimality of $\opw{2}{1}$ being a bad write, we have that $\opw{\operatorname{min}}{1}$ is not a bad write.
			Now $\op{\rho, j}{} \le_{A_1(\V)} \opw{\operatorname{min}}{1}$, by Observation \ref{observation Ci minimal write}.
			$\op{\rho, j}{}$ is either a read or a write operation.
			If $\op{\rho, j}{}$ is a read operation, then $\op{\rho, j}{} \in (\operatorname{r}, 1, *, *)$ and so $\edge{ \op{\rho, j}{}, \opw{\operatorname{min}}{1} } \in {PO}$.
			If $\op{\rho, j}{}$ is a write operation, then $\op{\rho, j}{} \in (\operatorname{w}, *, *, *)$ and so $\op{\rho, j}{} \le_{SWO(\V)} \opw{\operatorname{min}}{1}$.
			Since $\opw{\operatorname{min}}{1}$ is not a bad write, thus $\op{\rho, j}{} \le_{SWO(\V')} \opw{\operatorname{min}}{1}$.
			Therefore, in either case $\op{\rho, j}{} \le_{A_1(\V')} \opw{\operatorname{min}}{1}$.
			Furthermore, by the choice of $j$, we have that $\opw{1}{} \le_{A_1(\V')} o_j$.
			So we get that $\opw{1}{} \le_{A_1(\V')} \op{\rho, j}{} \le_{A_1(\V')} \opw{\operatorname{min}}{1} \PO \opw{2}{1}$.
			This implies that $\edge{ \opw{1}{}, \opw{2}{1} } \in {A_1(\V')}$ and so $\edge{ \opw{1}{}, \opw{2}{1} } \in SWO(\V')$.
			This contradicts the initial assumption that $\edge{ \opw{1}{}, \opw{2}{1} } \not \in SWO(\V')$.
			
			\begin{figure}
				\centering
				{
					\begin{tikzpicture}
					\node at (1, 3) {\LARGE $A_1(\V)$};
					
					\node (w1) at (0, 0) {$\opw{1}{}$};
					\node (opj) at (3, 1) {$\op{\rho, j}{}$};
					\node (opj+1) at (5, 2) {$\op{\rho, j+1}{}$};
					\node (wmin) at (5, 3) {$\opw{\operatorname{min}}{1}$};
					\node (w21) at (8, 3) {$\opw{2}{1}$};
					
					\draw[-{Latex}, thick] (w1) to[out=0, in=180] (opj);
					\draw[-{Latex[open]}, thick] (opj) to node[below right] {$B_1(\V)$} (opj+1);
					\draw[-{Latex}, thick] (opj) to[out=90, in=180] (wmin);
					\draw[-{Latex[open]}, thick] (wmin) to[out=10, in=170] node[above] {$PO$} (w21);
					\draw[-{Latex}, thick] (opj+1) to[out=0, in=200] (w21);
					\end{tikzpicture}
				}
				\caption{Proof of Lemma \ref{lemma replay preserves SWO and B_i}(a) Case 4 (ii)}
			\end{figure}
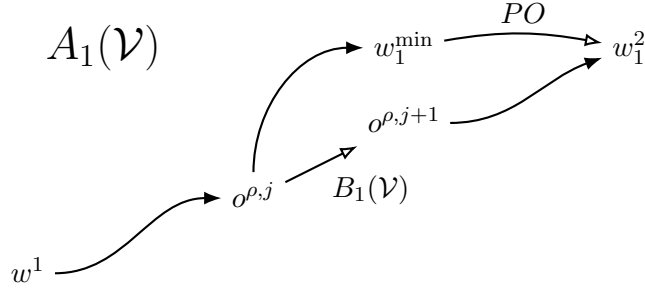
		\end{enumerate}
	\end{enumerate}
	In all cases, we get the desired contradiction.
\end{proof_of}
\\

\begin{proof_of}{Lemma \ref{lemma replay preserves SWO and B_i}(b)}
	Consider any arbitrary set of views $\V' = \set{ V'_i }_{i \in P}$ that certify a strongly causal consistent replay to be valid for $\R$.
	We will call a pair of write operations $\edge{\opw{1}{i}, \opw{2}{}}$, $\opw{1}{i} \in (\operatorname{w}, i, *, *)$ and $\opw{2}{} \in (\operatorname{w}, *, *, *)$, \emph{bad} if
	\begin{enumerate}[topsep=0pt]
		\item
		$\edge{\opw{2}{}, \opw{1}{i}} \in A_i(\V')$, and
		
		\item
		there exists a process $m \in P$ such that either
		\begin{enumerate}[topsep=0pt]
			\item $m \ne i$ and ${A_m(\V)} \cupdot C_i( \V, \opw{1}{i}, \opw{2}{} )$ has a cycle, or
			\item $m = i$ and $\del[1]{A_m(\V) \setminus \set[0]{ \edge{\opw{1}{i}, \opw{2}{}} } } \cupdot C_i( \V, \opw{1}{i}, \opw{2}{} )$ has a cycle.
		\end{enumerate}
	\end{enumerate}
	First note that if there is no bad write pair, then this implies the result as follows.
	We show the contrapositive that if $V_i' \not \supseteq B_i(\V)$, then there exists a bad write pair.
	Suppose there exist two distinct operations $\op{}{} \in (*, *, *, *)$ and $w' \in (\operatorname{w}, *, *, *)$ such that $\edge{ \op{}{}, w' } \in {B_i(\V)}$ but $\edge{ w', \op{}{} } \in V'_i$, for some process $i \in P$.
	Recall from Definition \ref{definition DRO Bi} that $B_i(\V') \subseteq DRO(V'_i)$ .
	Therefore $\edge{ w', \op{}{} } \in DRO(V'_i)$.
	Then, from Observation \ref{observation Ci minimal write}, there exists a write operation $\opw{\operatorname{min}}{i} \in (\operatorname{w}, i, *, *)$ such that $C_i( \V, \op{}{}, w' ) = C_i( \V, \opw{\operatorname{min}}{i}, w' )$ and $\op{}{} \le_{A_i(\V)} \opw{\operatorname{min}}{i}$.
	We show that $\edge{\opw{\operatorname{min}}{i}, w'}$ is a bad write pair.
	Condition 2 follows from the fact that $C_i( \V, \op{}{}, w' ) = C_i( \V, \opw{\operatorname{min}}{i}, w' )$ and $\edge{ \op{}{}, w' } \in B_i(\V)$ by interchanging $C_i( \V, \op{}{}, w' )$ with $C_i( \V, \opw{\operatorname{min}}{i}, w' )$ in Definition \ref{definition DRO Bi}.
	So it is left to show that $\edge{w', \opw{\operatorname{min}}{i}} \in A_i(\V')$.
	If $\op{}{}$ is a read operation then $\op{}{} \in (\operatorname{r}, i, *, *)$ and so $\op{}{} \le_{PO} \opw{\operatorname{min}}{i}$.
	If $\op{}{}$ is a write operation then $\op{}{} \le_{SWO(\V)} \opw{\operatorname{min}}{i}$ (recall that $\op{}{} \le_{A_i(\V)} \opw{\operatorname{min}}{i}$) and, by Lemma \ref{lemma replay preserves SWO and B_i}(a), $\op{}{} \le_{SWO(\V')} \opw{\operatorname{min}}{i}$.
	In either case $\op{}{} \le_{A_i(\V')} \opw{\operatorname{min}}{i}$.
	Therefore $w' \DRO{V'_i} \op{}{} \le_{A_i(\V')} \opw{\operatorname{min}}{i}$ and so $\edge{w', \opw{\operatorname{min}}{i}} \in A_i(\V')$ by Definition \ref{definition DRO Ai}.
	
	We now proceed via contradiction to show that there are no bad write pairs.
	Suppose that there exists at least one bad write pair.
	Recall from Definition \ref{definition:strongWriteOrder} that $SWO(\V)$ orders only write operations and is a partial order for strongly causal consistent executions.
	Consider any bad write pair $\edge{\opw{1}{i}, \opw{2}{}}$ such that $\opw{1}{i} \in (\operatorname{w}, i, *, *)$ is maximal with respect to $SWO(\V)$; i.e. for every write operation ${\opw{}{}}' >_{SWO(\V)} \opw{1}{i}$, we have that there is no $w''$ such that $\edge{w', w''}$ is a bad write pair.
	We also assume that $\opw{2}{}$ is minimal with respect to $V_i$; i.e. for every write operation ${\opw{}{}}' <_{V_i} \opw{2}{}$ we have that $\edge{\opw{1}{i}, {\opw{}{}}'}$ is not a bad write pair.
	Since $\edge{\opw{1}{i}, \opw{2}{}} \in B_i(\V)$, by Definition \ref{definition DRO Bi}, there exists a process, WLOG process $1$, such that either
	\begin{enumerate}[topsep=0pt]
		\item $i \ne 1$ and ${A_1(\V)} \cupdot C_i( \V, \opw{1}{i}, \opw{2}{} )$ has a cycle, or
		\item $i = 1$ and $\del[1]{A_1(\V) \setminus \set[0]{ \edge{\opw{1}{i}, \opw{2}{}} } } \cupdot C_i( \V, \opw{1}{i}, \opw{2}{} )$ has a cycle.
	\end{enumerate}
	We will show that
	\begin{enumerate}[topsep=0pt]
		\item $C_i( \V, \opw{1}{i}, \opw{2}{} ) \subseteq SWO(\V')$, and
		\item ${A_{1}(\V')} \cupdot C_i( \V, \opw{1}{i}, \opw{2}{} )$ has a cycle.
	\end{enumerate}
	This implies that $V'_1 \supseteq A_1(\V')$ does not respect $C_i( \V, \opw{1}{i}, \opw{2}{} ) \subseteq SWO(\V')$, hence giving us the desired contradiction.

	\begin{claim} \label{C_i sub SWO}
		$C_i( \V, \opw{1}{i}, \opw{2}{} ) \subseteq SWO(\V')$.
	\end{claim}

	\begin{proof_claim}
		Suppose, for the sake of contradiction, that $C_i( \V, \opw{1}{i}, \opw{2}{} ) \not \subseteq SWO(\V')$.
		Consider the smallest $\ell \ge 1$ such that $C_i^\ell( \V, \opw{1}{i}, \opw{2}{} ) \not \subseteq SWO(\V')$.
		Consider any two writes $\opw{3}{} \in (\operatorname{w}, *, *, *)$ and $\opw{4}{i'} \in (\operatorname{w}, i', *, *)$ such that $\edge{\opw{3}{}, \opw{4}{i'}} \in C_i^\ell( \V, \opw{1}{i}, \opw{2}{} )$ but $\edge{\opw{3}{}, \opw{4}{i'}} \not \in SWO(\V')$.
		Using $C_i^0( \V, \opw{1}{i}, \opw{2}{} ) = \set[0]{\edge{\opw{2}{}, \opw{1}{i}}}$, by Definition \ref{definition DRO Ci}, we have that there exists a $\opw{3}{} \opw{4}{i'}$-path $\psi$ in $A_{i'}(\V) \cupdot C_i^{\ell-1}( \V, \opw{1}{i}, \opw{2}{} )$ given by $\opw{3}{} = \opw{\psi, 0}{} \le_{A_{i'}(\V)} \opw{\psi, 1}{} <_{C_i^{\ell-1}( \V, \opw{1}{i}, \opw{2}{} )} \opw{\psi, 2}{} \le_{A_{i'}(\V)} \opw{\psi, 3}{} \dots <_{C_i^{\ell-1}( \V, \opw{1}{i}, \opw{2}{} )} \opw{\psi, k-1}{} \le_{A_{i'}(\V)} \opw{\psi, k}{} = \opw{4}{i'}$.
		By the choice of $\ell$, we have that $C_i^{\ell-1}( \V, \opw{1}{i}, \opw{2}{} ) \subseteq SWO(\V')$.
		WLOG, we can assume that $\opw{\psi, j}{} \le_{SWO(\V')} \opw{4}{i'}$ for every $j \in \sbr{1, k-1}$, since otherwise we can consider $\edge{ \opw{\psi, j}{}, \opw{4}{i'} } \in C_i^\ell( \V, \opw{1}{i}, \opw{2}{} )$ instead of $\edge{ \opw{3}{}, \opw{4}{i'} }$.
		Therefore it is sufficient to show that $\opw{3}{} = \opw{\psi, 0}{} \le_{A_{i'}(\V')} \opw{\psi, 1}{}$, since this implies a $\opw{3}{} \opw{4}{i'}$-path in $A_{i'}(\V') \cupdot SWO(\V')$ which is a contradiction since $\edge{\opw{3}{}, \opw{4}{i'}} \not \in SWO(\V')$.
		
		If $\opw{3}{} = \opw{\psi, 1}{}$, then we are done.
		So suppose $\opw{3}{} <_{A_{i'}(\V)} \opw{\psi, 1}{}$.
		Consider a path $\rho$ from $\opw{3}{}$ to $\opw{\psi, 1}{}$ in $\hat{A}_{i'}$ given by $\opw{3}{} = \op{\rho, 0}{} \lessdot_{A_{i'}(\V)} \op{\rho, 1}{} \lessdot_{A_{i'}(\V)} \op{\rho, 2}{} \lessdot_{A_{i'}(\V)} \dots \lessdot_{A_{i'}(\V)} \op{\rho, k'}{} = \opw{\psi, 1}{}$.
		Note that each operation in the path is in the view ${V_{i'}}$ and hence in $V'_{i'}$.
		If $\edge{ \op{\rho, j'}{}, \op{\rho, j'+1}{} } \in {A_{i'}(\V')}$ for every $j' \in \sbr{0, k'-1}$, then $\edge{ \opw{3}{}, \opw{\psi, 1}{} } \in SWO(\V')$ which is a contradiction.
		So there exists a $j' \in \sbr{0, k'-1}$ such that $\edge{ \op{\rho, j'}{}, \op{\rho, j'+1}{} } \not \in {A_{i'}(\V')}$.
		
		Consider any $j' \in \sbr{0, k'-1}$ such that $\edge{ \op{\rho, j'}{}, \op{\rho, j'+1}{} } \not \in {A_{i'}(\V')}$.
		There are $4$ cases to consider.
		\begin{enumerate}
			[label=\underline{Case \arabic*:},wide,topsep=0pt,labelindent=0pt]
			\item $\edge{ \op{\rho, j'}{}, \op{\rho, j'+1}{} } \in \hat{A}_{i'}(\V) 
			\setminus \del{ SWO(\V) \cupdot PO \cupdot B_{i'}(\V) }$.
			Then $\edge{ \op{\rho, j'}{}, \op{\rho, j'+1}{} } \in R_{i'}$ and $V'_{i'}$ respects $R_{i'}$ since $\V'$ is a replay of $\R$. Thus $\edge{ \op{\rho, j'}{}, \op{\rho, j'+1}{} } \in DRO(V'_{i'})$ and so $\edge{ \op{\rho, j'}{}, \op{\rho, j'+1}{} } \in A_{i'}(\V')$, a contradiction.
			
			\item $\edge{ \op{\rho, j'}{}, \op{\rho, j'+1}{} } \in PO$.
			Then $A_{i'}(\V')$ respects $PO$ due to consistency and $PO$ is independent of executions.
			Thus $\edge{ \op{\rho, j'}{}, \op{\rho, j'+1}{} } \in {A_{i'}(\V')}$, a contradiction.
			
			\item $\edge{ \op{\rho, j'}{}, \op{\rho, j'+1}{} } \in SWO_{i'}(\V)$.
			Then $A_{i'}(\V') \supseteq SWO(\V')$ by Definition \ref{definition DRO Ai} and $SWO(\V') \supseteq SWO(\V)$ by Lemma \ref{lemma replay preserves SWO and B_i}(a).
			Thus $\edge{ \op{\rho, j'}{}, \op{\rho, j'+1}{} } \in {A_{i'}(\V')}$, a contradiction.
			
			\item $\edge{ \op{\rho, j'}{}, \op{\rho, j'+1}{} } \in B_{i'}(\V)$.
			Then by Definition \ref{definition DRO Bi} we have that $\edge{ \op{\rho, j'}{}, \op{\rho, j'+1}{} } \in DRO(V_{i'})$ and that $C_{i'}( \V, \op{\rho, j'}{}, \op{\rho, j'+1}{} )$ is non-empty.
			Thus, from Observation \ref{observation Ci minimal write}, there exists $\opw{\operatorname{min}}{i'} \in (\operatorname{w}, i', *, *)$ such that $C_{i'}( \V, \op{\rho, j'}{}, \op{\rho, j'+1}{} ) = C_{i'}( \V, \opw{\operatorname{min}}{i'}, \op{\rho, j'+1}{} )$.
			Since $\opw{4}{i'} \in (\operatorname{w}, i', *, *)$, therefore either $\opw{\operatorname{min}}{i'} \le_{PO} \opw{4}{i'}$ or $\opw{4}{i'} \PO \opw{\operatorname{min}}{i'}$.
			We consider both cases.
			\begin{enumerate}
				[label=\underline{Case \roman*:},topsep=0pt,labelwidth=\widthof{Case 
				ii:a},leftmargin=!,itemindent=\widthof{aaaaaaa}]
				\item $\opw{\operatorname{min}}{i'} \le_{PO} \opw{4}{i'}$.
				Then $\opw{3}{} \le_{A_{i'}(\V)} \op{\rho, j'}{} \le_{A_{i'}(\V)} \opw{\operatorname{min}}{i'} \le_{PO} \opw{4}{i'}$.
				This implies that $\edge{ \opw{3}{}, \opw{4}{i'} } \in SWO(\V)$, and therefore $\edge{ \opw{3}{}, \opw{4}{i'} } \in SWO(\V')$ by Lemma \ref{lemma replay preserves SWO and B_i}(a), a contradiction to the initial assumption that $\edge{ \opw{3}{}, \opw{4}{i'} } \not \in SWO(\V')$.
				
				\begin{figure}
					\centering
					{
						\begin{tikzpicture}
						\node at (1, 3) {\LARGE $A_{i'}(\V)$};
						
						\node (w3) at (0, 0) {$\opw{3}{}$};
						\node (opj) at (3, 1) {$\op{\rho, j'}{}$};
						\node (opj+1) at (5, 2) {$\op{\rho, j'+1}{}$};
						\node (wpsi1) at (8, 3) {$\opw{\psi, 1}{}$};
						\node (wmin) at (8, 5) {$\opw{\operatorname{min}}{i'}$};
						\node (w41) at (10, 5) {$\opw{4}{i'}$};
						
						\draw[-{Latex}, thick] (w3) to[out=0, in=180] (opj);
						\draw[-{Latex[open]}, thick] (opj) to node[below right] {$B_1(\V)$} (opj+1);
						\draw[-{Latex}, thick] (opj) to[out=45, in=180] (wmin);
						\draw[-{Latex}, thick] (opj+1) to[out=0, in=200] (wpsi1);
						\draw[-{Latex}, thick] (wmin) to node[above] {$PO$} (w41);
						\end{tikzpicture}
					}
					\caption{Proof of Claim \ref{C_i sub SWO} Case 4 (i)}
				\end{figure}
				
				\item $\opw{4}{i'} \PO \opw{\operatorname{min}}{i'}$.
				Now $\op{\rho, j'}{} \le_{A_{i'}(\V)} \opw{\operatorname{min}}{i'}$, by Observation \ref{observation Ci minimal write}.
				$\op{\rho, j'}{}$ is either a read or a write operation.
				If $\op{\rho, j'}{}$ is a read operation, then $\op{\rho, j'}{} \in (\operatorname{r}, i', *, *)$ and so $\edge{ \op{\rho, j'}{}, \opw{\operatorname{min}}{i'} } \in {PO}$.
				If $\op{\rho, j'}{}$ is a write operation, then $\op{\rho, j'}{} \in (\operatorname{w}, *, *, *)$ and so $\op{\rho, j'}{} \le_{SWO(\V)} \opw{\operatorname{min}}{i'}$ and by Lemma \ref{lemma replay preserves SWO and B_i}(a) $\op{\rho, j'}{} \le_{SWO(\V')} \opw{\operatorname{min}}{i'}$.
				Therefore, in either case $\op{\rho, j'}{} \le_{A_{i'}(\V')} \opw{\operatorname{min}}{i'}$.
				Now since $\edge{ \op{\rho, j'}{}, \op{\rho, j'+1}{} } \in DRO(V_{i'})$ and $\edge{ \op{\rho, j'}{}, \op{\rho, j'+1}{} } \not \in {A_{i'}(\V')}$, thus $\edge{ \op{\rho, j'+1}{}, \op{\rho, j'}{} } \in DRO(V'_{i'})$ and $\edge{ \op{\rho, j'+1}{}, \op{\rho, j'}{} } \in {A_{i'}(\V')}$.
				It follows that $\op{\rho, j'+1}{} <_{A_{i'}(\V')} \op{\rho, j'}{} \le_{A_{i'}(\V')} \opw{\operatorname{min}}{i'}$.
				Since $\op{\rho, j'+1}{}$ is a write operation (by Definition \ref{definition DRO Bi}), thus $\edge{\opw{\operatorname{min}}{i'}, \op{\rho, j'+1}{}}$ is a bad write pair (recall that $\edge{ \op{\rho, j'}{}, \op{\rho, j'+1}{} } \in B_{i'}(\V)$ and $C_{i'}( \V, \op{\rho, j'}{}, \op{\rho, j'+1}{} ) = C_{i'}( \V, \opw{\operatorname{min}}{i'}, \op{\rho, j'+1}{} )$).

				\begin{figure}
					\centering
					{
						\begin{tikzpicture}
						\node at (1, 3) {\LARGE $A_{i'}(\V)$};
						
						\node (w3) at (0, 0) {$\opw{3}{}$};
						\node (opj) at (3, 1) {$\op{\rho, j'}{}$};
						\node (opj+1) at (5, 2) {$\op{\rho, j'+1}{}$};
						\node (wpsi1) at (8, 3) {$\opw{\psi, 1}{}$};
						\node (w41) at (8, 1) {$\opw{4}{i'}$};
						\node (wmin) at (10, 1) {$\opw{\operatorname{min}}{i'}$};
						\node (w1i) at (5, 0) {$\opw{1}{i}$};
						
						\draw[-{Latex}, thick] (w3) to[out=0, in=180] (opj);
						\draw[-{Latex[open]}, thick] (opj) to node[above left] {$B_{i'}(\V)$} (opj+1);
						\draw[-{Latex}, thick] (opj) to[out=0, in=135] (wmin);
						\draw[-{Latex}, thick] (opj+1) to[out=0, in=180] (wpsi1);
						\draw[-{Latex[open]}, thick] (w41) to node[below] {$PO$} (wmin);
						\draw[-{Latex}, thick] (w1i) to[out=0, in=180] node[below right] {$SWO(\V)$} (w41);
						\end{tikzpicture}
					}
					\caption{Proof of Claim \ref{C_i sub SWO} Case 4 (ii)}
				\end{figure}
				
				Now, by Observation \ref{observation Ci SWO}, we have that $\opw{1}{i} \le_{SWO(\V)} \opw{4}{i'}$ since $\edge{\opw{3}{}, \opw{4}{i'}} \in C_i^\ell( \V, \opw{1}{i}, \opw{2}{} )$.
				So $\opw{1}{i} \le_{SWO(\V)} \opw{4}{i'} \PO \opw{\operatorname{min}}{i'}$ which implies $\edge{ \opw{1}{i},  \opw{\operatorname{min}}{i'} } \in SWO(\V)$.
				Since both $\edge{\opw{\operatorname{min}}{i'}, \op{\rho, j'+1}{}}$ and $\edge{\opw{1}{i}, \opw{2}{}}$ are bad write pairs, thus this contradicts the maximality of $\opw{1}{i}$.
			\end{enumerate}
		\end{enumerate}
		In all cases, we get a contradiction.
		Therefore $\opw{3}{} \le_{A_{i'}(\V')} \opw{\psi, 1}{}$, as required.
	\end{proof_claim}
	
	\begin{claim} \label{claim A_1 C_i cycle}
		${A_{1}(\V')} \cupdot C_i( \V, \opw{1}{i}, \opw{2}{} )$ has a cycle.
	\end{claim}
	\begin{proof_claim}
		Since $\edge{ \opw{1}{i}, \opw{2}{}}$ is a bad write pair, therefore either
		\begin{enumerate}
			\item $i \ne 1$ and ${A_1(\V)} \cupdot C_i( \V, \opw{1}{i}, \opw{2}{} )$ has a cycle, or
			\item $i = 1$ and $\del[1]{A_1(\V) \setminus \set[0]{ \edge{\opw{1}{i}, \opw{2}{}} } } \cupdot C_i( \V, \opw{1}{i}, \opw{2}{} )$ has a cycle.
		\end{enumerate}
		Consider one such cycle $\psi$ given by $\opw{\psi, 0}{} \le_{A_{1}(\V)} \opw{\psi, 1}{} <_{C_i( \V, \opw{1}{i}, \opw{2}{} )} \opw{\psi, 2}{} \le_{A_{1}(\V)} \dots \le_{A_{1}(\V)} \opw{\psi, k-1}{} <_{C_i( \V, \opw{1}{i}, \opw{2}{} )} \opw{\psi, k}{} = \opw{\psi, 0}{}$.
		If $i \ne 1$, then we let $\psi$ be any cycle.
		However, if $i = 1$, then we select $\psi$ to be a cycle with some particular properties.
		If there exists a cycle $\psi$ such that there is no even $j$ with $\opw{\psi, j}{} = \opw{1}{i}$, then we select that cycle.
		Otherwise, we select $\psi$ as follows. 
		Since we can rotate cycles, we assume WLOG that $\opw{\psi, k}{} = \opw{\psi, 0}{} = \opw{1}{i}$.
		We say that $\psi$ has \emph{level} $\ell$ if $\ell$ is the smallest integer such that $\edge{\opw{\psi, k-1}{}, \opw{1}{i}} \in {C^\ell_i( \V, \opw{1}{i}, \opw{2}{} )}$.
		We select $\psi$ such that it has the lowest level $\ell$. 
		The reason behind this choice will become clearer in case 1 below.

		\begin{enumerate}
			[label=\underline{Case \arabic*:},wide,topsep=0pt,labelindent=0pt]
			\item
			$i = 1$ and there exists an even $j \in \sbr{0, k-1}$ such that $\opw{\psi, j}{} = \opw{1}{i}$.
			WLOG we can assume that $\opw{\psi, k}{} = \opw{\psi, 0}{} = \opw{1}{i}$ since we can rotate the cycle $\psi$.
			We first show that the choice of $\psi$ implies that $\edge{\opw{\psi, k-1}{}, \opw{1}{i}} \in {C^1_i( \V, \opw{1}{i}, \opw{2}{} )}$.
			Suppose for the sake of contradiction that the level of $\psi$ is $\ell > 1$ so that $\ell$ is the smallest integer such that $\edge{\opw{\psi, k-1}{}, \opw{1}{i}} \in {C^\ell_i( \V, \opw{1}{i}, \opw{2}{} )}$.
			By Definition \ref{definition DRO Ci} there exists a $\opw{\psi, k-1}{} \opw{1}{i}$-path $\rho$ in $A_1(\V) \cupdot {C^{\ell-1}_i( \V, \opw{1}{i}, \opw{2}{} )}$.
			Then either $\psi \cupdot \rho$ is a cycle or $\rho$ intersects with $\psi$ other than at endpoints.
			In the first case we have found a cycle with level smaller than $\psi$ and in the second case $\psi \cupdot \rho$ has a cycle that does not use $\opw{1}{i}$.
			In either case we have a contradiction with the choice of $\psi$.

			We now show that there exists a path from $\opw{1}{i}$ to $\opw{2}{}$ in $\del[1]{A_1(\V) \setminus \set[0]{ \edge{\opw{1}{i}, \opw{2}{}} } }$.
			Since $\edge{\opw{\psi, k-1}{}, \opw{1}{i}} \in {C^1_i( \V, \opw{1}{i}, \opw{2}{} )}$ we have that $\opw{\psi, k-1}{} \le_{A_{1}(\V)} \opw{2}{}$ by Definition $\ref{definition DRO Ci}$.
			If $k > 2$, then $\edge{\opw{\psi, k-3}{}, \opw{\psi, k-2}{}} \in C_i( \V, \opw{1}{i}, \opw{2}{} )$ and so $\opw{1}{i} \le_{SWO(\V)} \opw{\psi, k-2}{}$ by Observation \ref{observation Ci SWO}.
			If $k = 2$, then $\opw{1}{i} = \opw{\psi, k-2}{}$.
			In either case we get that $\opw{1}{i} \le_{SWO(\V)} \opw{\psi, k-2}{} \le_{A_1(\V)} \opw{\psi, k-1}{} \le_{A_{1}(\V)} \opw{2}{}$.

			Note that $\opw{1}{i} \ne \opw{2}{}$.
			There are 3 cases to consider
			\begin{enumerate}
				[label=\underline{Case \roman*:},topsep=0pt,labelwidth=\widthof{Case 
				ii:a},leftmargin=!,itemindent=\widthof{aaaaaaa}]
				\item $\opw{\psi, k-2}{} = \opw{\psi, k-1}{} = \opw{2}{}$.
				Then $\edge{ \opw{1}{i}, \opw{2}{} } \in SWO(\V)$ and by Lemma \ref{lemma replay preserves SWO and B_i}(a) $\edge{ \opw{1}{i}, \opw{2}{} } \in SWO(\V')$.
				This contradicts with the assumption that $\edge{ \opw{1}{i}, \opw{2}{} }$ is a bad write pair (which implies $\edge{ \opw{2}{}, \opw{1}{i} } \in A_i(\V')$).
				
				\item $\opw{\psi, k-2}{} = \opw{\psi, k-1}{} \ne \opw{2}{}$.
				Since $\edge{\opw{\psi, k-1}{}, \opw{1}{i}} \in {C_i( \V, \opw{1}{i}, \opw{2}{} )}$, by Definition \ref{definition DRO Ci}, we have that $\opw{1}{i} \ne \opw{\psi, k-1}{}$.
				Therefore $\opw{1}{i} <_{SWO(\V)} \opw{\psi, k-2}{} = \opw{\psi, k-1}{} <_{A_{1}(\V)} \opw{2}{}$ is a $\opw{1}{i} \opw{2}{}$-path in $\del[1]{A_1(\V) \setminus \set[0]{ \edge{\opw{1}{i}, \opw{2}{}} } }$.
				
				\item $\opw{\psi, k-2}{} \ne \opw{\psi, k-1}{}$.
				Since, by construction of $\psi$, there is a $\opw{\psi, k-2}{} \opw{\psi, k-1}{}$-path in $\del[1]{A_1(\V) \setminus \set[0]{ \edge{\opw{1}{i}, \opw{2}{}} } }$, therefore there is a $\opw{1}{i} \opw{2}{}$-path in $\del[1]{A_1(\V) \setminus \set[0]{ \edge{\opw{1}{i}, \opw{2}{}} } }$.
			\end{enumerate}
			
			Therefore, there exists a path from $\opw{1}{i}$ to $\opw{2}{}$ in $\del[1]{A_1(\V) \setminus \set[0]{ \edge{\opw{1}{i}, \opw{2}{}} } }$.
			We now show that $\edge{ \opw{1}{i}, \opw{2}{} } \in A_1(\V')$ which contradicts with the assumption that $\edge{ \opw{1}{i}, \opw{2}{} }$ is a bad write pair (which implies $\edge{ \opw{2}{}, \opw{1}{i} } \in A_i(\V')$).
			Since $\hat{A}_1(\V)$ preserves all paths, we can consider the corresponding $\opw{1}{i} \opw{2}{}$-path $\rho$ in $\hat{A}_1(\V)$ given by $\opw{1}{} = \op{\rho, 0}{} \lessdot_{A_{1}(\V)} \op{\rho, 1}{} \lessdot_{A_{1}(\V)} \op{\rho, 2}{} \lessdot_{A_{1}(\V)} \dots \lessdot_{A_{1}(\V)} \op{\rho, k'}{} = \opw{2}{}$.
			Observe that $\rho$ does not use the $\edge{\opw{1}{i}, \opw{2}{} }$ edge (property of transitive reduction).
			If $\edge{ \op{\rho, j'}{}, \op{\rho, j'+1}{} } \in {A_{1}(\V')}$ for every $j' \in \sbr{0, k'-1}$, then $\edge{ \opw{1}{i}, \opw{2}{} } \in A_{1}(\V')$ which is a contradiction.
			So there exists a $j' \in \sbr{0, k'-1}$ such that $\edge{ \op{\rho, j'}{}, \op{\rho, j'+1}{} } \not \in {A_{1}(\V')}$.
			
			Consider the minimum $j' \in \sbr{0, k'-1}$ such that $\edge{ \op{\rho, j'}{}, \op{\rho, j'+1}{} } \not \in {A_{1}(\V')}$.
			Therefore $\opw{1}{i} \le_{A_{1}(\V')} \op{\rho, j'}{}$.
			Similar to proof of Claim \ref{C_i sub SWO}, the interesting case is when $\edge{ \op{\rho, j'}{}, \op{\rho, j'+1}{} } \in B_{1}(\V)$.
			Then by Definition \ref{definition DRO Bi} we have that $\edge{ \op{\rho, j'}{}, \op{\rho, j'+1}{} } \in DRO(V_1)$ and that $C_{1}( \V, \op{\rho, j'}{}, \op{\rho, j'+1}{} )$ is non-empty.
			Thus, from Observation \ref{observation Ci minimal write}, there exists $\opw{\operatorname{min}}{1} \in (\operatorname{w}, 1, *, *)$ such that $C_{1}( \V, \op{\rho, j'}{}, \op{\rho, j'+1}{} ) = C_{1}( \V, \opw{\operatorname{min}}{1}, \op{\rho, j'+1}{} )$.
			Now $\op{\rho, j'}{} \le_{A_{1}(\V)} \opw{\operatorname{min}}{1}$, by Observation \ref{observation Ci minimal write}.
			$\op{\rho, j'}{}$ is either a read or a write operation.
			If $\op{\rho, j'}{}$ is a read operation, then $\op{\rho, j'}{} \in (\operatorname{r}, 1, *, *)$ and so $\edge{ \op{\rho, j'}{}, \opw{\operatorname{min}}{1} } \in {PO}$.
			If $\op{\rho, j'}{}$ is a write operation, then $\op{\rho, j'}{} \in (\operatorname{w}, *, *, *)$ and so $\op{\rho, j'}{} \le_{SWO(\V)} \opw{\operatorname{min}}{1}$ and by Lemma \ref{lemma replay preserves SWO and B_i}(a) $\op{\rho, j'}{} \le_{SWO(\V')} \opw{\operatorname{min}}{1}$.
			Therefore, in either case $\op{\rho, j'}{} \le_{A_{1}(\V')} \opw{\operatorname{min}}{1}$.
			Now since $\edge{ \op{\rho, j'}{}, \op{\rho, j'+1}{} } \in DRO(V_1)$ and $\edge{ \op{\rho, j'}{}, \op{\rho, j'+1}{} } \not \in {A_{1}(\V')}$, thus $\edge{ \op{\rho, j'+1}{}, \op{\rho, j'}{} } \in DRO(V'_1)$ and $\edge{ \op{\rho, j'+1}{}, \op{\rho, j'}{} } \in {A_{1}(\V')}$.
			It follows that $\op{\rho, j'+1}{} <_{A_{1}(\V')} \op{\rho, j'}{} \le_{A_{1}(\V')} \opw{\operatorname{min}}{1}$.
			Since $\op{\rho, j'+1}{}$ is a write operation (by Definition \ref{definition DRO Bi}), thus $\edge{\opw{\operatorname{min}}{1}, \op{\rho, j'+1}{}}$ is a bad write pair (recall that $\edge{ \op{\rho, j'}{}, \op{\rho, j'+1}{} } \in B_{1}(\V)$ and $C_{1}( \V, \op{\rho, j'}{}, \op{\rho, j'+1}{} ) = C_{1}( \V, \opw{\operatorname{min}}{1}, \op{\rho, j'+1}{} )$).
			
			\begin{figure}
				\centering
				{
					\begin{tikzpicture}
					\node at (1, 3) {\LARGE $A_1(\V)$};
					
					\node (w1i) at (0, 0) {$\opw{1}{i}$};
					\node (opj) at (3, 1) {$\op{\rho, j'}{}$};
					\node (opj+1) at (5, 2) {$\op{\rho, j'+1}{}$};
					\node (wmin) at (6, 0.5) {$\opw{\operatorname{min}}{1}$};
					\node (w2) at (8, 3) {$\opw{2}{}$};
					
					\draw[-{Latex}, thick] (w1i) to[out=0, in=180] (opj);
					\draw[-{Latex[open]}, thick] (opj) to node[above left] {$B_1(\V)$} (opj+1);
					\draw[-{Latex}, thick] (opj) to[out=0, in=180] (wmin);
					\draw[-{Latex}, thick] (opj+1) to[out=0, in=180] (w2);
					\end{tikzpicture}
				}
				\caption{Proof of Claim \ref{claim A_1 C_i cycle} Case 1}
			\end{figure}
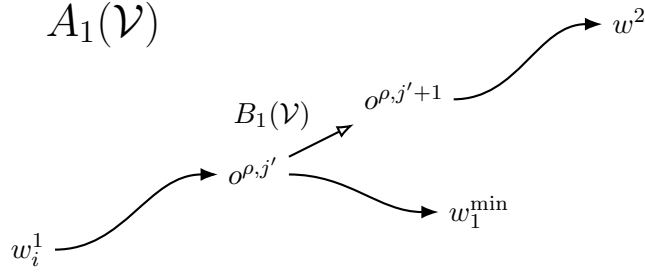
			
			Now,
			\begin{itemize}
				\item $\opw{1}{i} \le_{A_{1}(\V')} \op{\rho, j'}{}$ by choice of $j'$, and
				\item $\op{\rho, j'}{} \le_{A_{1}(\V)} w^{\operatorname{min}}_{1}$ by Observation \ref{observation Ci minimal write}.
			\end{itemize}
			Therefore we have that $\opw{1}{i} \le_{SWO(\V)} w^{\operatorname{min}}_{1}$.
			There are two cases to consider.
			\begin{enumerate}
				[label=\underline{Case \roman*:},topsep=0pt,labelwidth=\widthof{Case 
				ii:a},leftmargin=!,itemindent=\widthof{aaaaaaa}]
				\item $\edge{ \opw{1}{i}, w^{\operatorname{min}}_{1} } \in {SWO(\V)}$.
				This contradicts the maximality of $\opw{1}{i}$ since both $\edge{\opw{\operatorname{min}}{1}, \op{\rho, j'+1}{}}$ and $\edge{\opw{1}{i}, \opw{2}{}}$ are bad write pairs.
				
				\item $\opw{1}{i} = \op{\rho, j'}{} = w^{\operatorname{min}}_{1}$.
				Since $\rho$ is a path in $\del[1]{A_1(\V) \setminus \set[0]{ \edge{\opw{1}{i}, \opw{2}{}} } }$ and $\opw{1}{i} = \op{\rho, j'}{}$, thus $\op{\rho, j' + 1}{} \ne \opw{2}{}$ and by the minimality of $\opw{2}{}$, we have that $\edge{ w^{\operatorname{min}}_{1}, \op{\rho, j' + 1}{} }$ is not a bad write pair, a contradiction.
			\end{enumerate}
			
			\item
			Either $i \ne 1$ or there does not exist an even $j \in \sbr{0, k-1}$ such that $\opw{\psi, j}{} = \opw{1}{i}$.
			We show that for every even $j \in \sbr{0, k-1}$ we have that $\opw{\psi, j}{} \le_{A_{1}(\V')} \opw{\psi, j+1}{}$.
			It follows that ${A_{1}(\V')} \cupdot C_i( \V, \opw{1}{i}, \opw{2}{} )$ has a cycle and we are done.
			
			Consider any even $j \in \sbr{0, k-1}$.
			WLOG $j = 0$ since we can rotate the cycle.
			If $\opw{\psi, 0}{} = \opw{\psi, 1}{}$, then we are done.
			So assume $\opw{\psi, 0}{} <_{A_{1}(\V)} \opw{\psi, 1}{}$.
			Suppose for the sake of contradiction that $\opw{\psi, 0}{} \not <_{A_{1}(\V)} \opw{\psi, 1}{}$.
			Consider a path $\rho$ from $\opw{\psi, 0}{}$ to $\opw{\psi, 1}{}$ in $\hat{A}_{1}$ given by $\opw{\psi, 0}{} = \op{\rho, 0}{} \lessdot_{A_{1}(\V)} \op{\rho, 1}{} \lessdot_{A_{1}(\V)} \op{\rho, 2}{} \lessdot_{A_{1}(\V)} \dots \lessdot_{A_{1}(\V)} \op{\rho, k'}{} = \opw{\psi, 1}{}$.
			Note that each operation in the path is in the view ${V_{1}}$ and hence in $V'_{1}$.
			If $\edge{ \op{\rho, j'}{}, \op{\rho, j'+1}{} } \in {A_{1}(\V')}$ for every $j' \in \sbr{0, k'-1}$, then $\edge{ \opw{\psi, 0}{}, \opw{\psi, 1}{} } \in A_{1}(\V')$ which is a contradiction.
			So there exists a $j' \in \sbr{0, k'-1}$ such that $\edge{ \op{\rho, j'}{}, \op{\rho, j'+1}{} } \not \in {A_{1}(\V')}$.
			
			Consider the minimum $j' \in \sbr{0, k'-1}$ such that $\edge{ \op{\rho, j'}{}, \op{\rho, j'+1}{} } \not \in {A_{1}(\V')}$.
			Therefore $\opw{\psi, 0}{} \le_{A_{1}(\V')} \op{\rho, j'}{}$.
			Similar to proof of Claim \ref{C_i sub SWO}, the interesting case is when $\edge{ \op{\rho, j'}{}, \op{\rho, j'+1}{} } \in B_{1}(\V)$.
			Then by Definition \ref{definition DRO Bi} we have that $\edge{ \op{\rho, j'}{}, \op{\rho, j'+1}{} } \in DRO(V_1)$ and that $C_{1}( \V, \op{\rho, j'}{}, \op{\rho, j'+1}{} )$ is non-empty.
			Thus, from Observation \ref{observation Ci minimal write}, there exists $\opw{\operatorname{min}}{1} \in (\operatorname{w}, 1, *, *)$ such that $C_{1}( \V, \op{\rho, j'}{}, \op{\rho, j'+1}{} ) = C_{1}( \V, \opw{\operatorname{min}}{1}, \op{\rho, j'+1}{} )$.
			Similar to Case 1, we get that $\edge{\opw{\operatorname{min}}{1}, \op{\rho, j'+1}{}}$ is a bad write pair.
			
			Now,
			\begin{itemize}
				\item $\opw{1}{i} \le_{SWO(\V)} \opw{\psi, 0}{}$, by Observation \ref{observation Ci SWO} since $\opw{\psi, k}{} = \opw{\psi, 0}{}$ and $\edge{\opw{\psi, k-1}{}, \opw{\psi, k}{}} \in C_i[\V, \opw{1}{i}, \op{2}{}]$,
				\item $\opw{\psi, 0}{} \le_{A_{1}(\V')} \op{\rho, j'}{}$ by choice of $j'$, and
				\item $\op{\rho, j'}{} \le_{A_{1}(\V)} w^{\operatorname{min}}_{1}$ by Observation \ref{observation Ci minimal write}.
			\end{itemize}
			Therefore we have that $\opw{1}{i} \le_{SWO(\V)} w^{\operatorname{min}}_{1}$.
			There are two cases to consider.
			\begin{enumerate}
				[label=\underline{Case \roman*:},topsep=0pt,labelwidth=\widthof{Case 
				ii:a},leftmargin=!,itemindent=\widthof{aaaaaaa}]
				\item $\opw{1}{i} <_{SWO(\V)} w^{\operatorname{min}}_{1}$.
				This contradicts the maximality of $\opw{1}{i}$ since both $\edge{\opw{\operatorname{min}}{1}, \op{\rho, j'+1}{}}$ and $\edge{\opw{1}{i}, \opw{2}{}}$ are bad write pairs.
				
				\item $\opw{1}{i} = \opw{\psi, 0}{} = \op{\rho, j'}{} = 
				w^{\operatorname{min}}_{1}$.
				Then $i = 1$ and $\opw{\psi, 0}{} = w^{\operatorname{min}}_{1}$.
				This contradicts the assumption that either $i \ne 1$ or there does not exist an even $j \in \sbr{0, k-1}$ such that $\opw{\psi, j}{} = \opw{1}{i}$.
			\end{enumerate}
			In both cases, we get a contradiction.
			Therefore for every even $j \in \sbr{0, k-1}$ we have that $\opw{\psi, j}{} \le_{A_{1}(\V')} \opw{\psi, j+1}{}$ and so ${A_{1}(\V')} \cupdot C_i( \V, \opw{1}{i}, \opw{2}{} )$ has a cycle, as required.
		\end{enumerate}
	\end{proof_claim}
\end{proof_of}
\\

\begin{proof_of}{Theorem \ref{theorem DRO RnRStrongCausalConsistencySufficiency}}
	Consider any arbitrary set of views $\V'$ that certify a strongly causal consistent replay to be valid for $\R$.
	We show that for any process $i$ and any two operations $\op{1}{}, \op{2}{} \in (*, *, *, *)$ such that $\edge{ \op{1}{}, \op{2}{} } \in {DRO(V_i)}$ we must have that $\edge{ \op{1}{}, \op{2}{} } \in {V'_i}$.
	Consider any arbitrary process $i$.
	We have that
	\begin{itemize}
		\item $V'_i$ respects ${R_i}$, since $V'_i$ certifies a replay to be valid for $\R$;
		\item $V'_i$ respects $SWO_i(\V) \cup \del{PO | ( *, i, *, * ) \cup ( w, *, *, * )} \cup B_i(\V)$ due to consistency and Lemma \ref{lemma replay preserves SWO and B_i}.
	\end{itemize}
	Consider the $\op{1}{} \op{2}{}$-path $\rho$ in $\hat{A}_i$ given by $\op{1}{} = \op{\rho, 0}{} \lessdot_{A_{i}(\V)} \op{\rho, 1}{} \lessdot_{A_{i}(\V)} \op{\rho, 2}{} \lessdot_{A_{i}(\V)} \dots \lessdot_{A_{i}(\V)} \op{\rho, k}{} = \op{2}{}$.
	By construction of $\hat{A}_i$, each edge is either a $R_i$ edge or a ${PO}$ edge or a $SWO_i(\V)$ edge or a $B_i(\V)$ edge.
	Thus $\op{1}{} = \op{\rho, 0}{} <_{V'_i} \op{\rho, 1}{} <_{V'_i} \op{\rho, 2}{} <_{V'_i} \dots <_{V'_i} \op{\rho, k}{} = \op{2}{}$ and $\edge{ \op{1}{}, \op{2}{} } \in {V'_i}$, as required.
\end{proof_of}
\\

We extend Definition \ref{definition LRO strong causal order} of strong causal order to be applicable to a set of partial orders as follows.

\begin{definition} \label{definition extended strong causal order}
	Given a set of partial orders $\pV = \set{ U_i }_{i \in P}$, two writes, $\opw{1}{} \in (\operatorname{w}, *, *, *)$ and $\opw{2}{i} \in (\operatorname{w}, i, *, *)$, are ordered $\edge{ \opw{1}{}, \opw{2}{i} } \in { SCO( U_i ) }$, if $\edge{ \opw{1}{}, \opw{2}{i} } \in  { U_i } $.
	Furthermore, $SCO( \pV ) = \bigcup_{i \in P} SCO( U_i )$.
\end{definition}

\begin{lemma} \label{lemma sufficient to reproduce DRO}
	Given a set of partial orders $\pV = \set{U_i}_{i \in P}$ such that for each process $i \in P$, $U_i$ is a partial order on $( *, i, *, * ) \cup ( w, *, *, * )$ that satisfies transitivity and respects $SCO(\pV) \cup \del{PO | ( *, i, *, * ) \cup ( w, *, *, * )}$.
	Then there exists a strongly causal consistent execution $\V = \set{V_i}_{i \in P}$ such that each $V_i \supseteq U_i$.
\end{lemma}

\begin{proof}
	We extend $\pV$ to $\V$ iteratively.
	Let $U_i^t$ be the partial order after $t$ steps.
	Initially, $U_i^0 = U_i$.
	After some finite number of steps, $U_i^t$ will be a total order and we set $V_i = U_i^t$ at that step.
	We first order all the write operations for each process $i$ and then add edges for reads appropriately.
	At each step $t$, we consider two write operations $w_1 \in (\operatorname{w}, 1, *, *)$ and $\opw{2}{2} \in (\operatorname{w}, 2, *, *)$.
	\begin{enumerate}
		\item
		If $\opw{1}{1}, \opw{2}{2}$ are not related in $U_1^{t-1}$, then we set $U_1^t = U_1^{t-1} \cup \set[0]{ \edge{\opw{1}{1}, \opw{2}{2}} }$.
		
		\item
		If $\opw{1}{1}, \opw{2}{2}$ are not related in $U_2^{t-1}$, then we set $U_2^t = U_2^{t-1} \cup \set{ \edge{\opw{2}{2}, \opw{1}{1}} }$.
		
		\item
		For every process $k \ne 1, 2$, if $\opw{1}{1}, \opw{2}{2}$ are not related in $U_k^{t-1}$, then we do the following.
		If $SCO( U_k^{t-1} ) \cup \set{ \edge{\opw{1}{1}, \opw{2}{2}} } ] = SCO( U_k^{t-1} )$, then we set $U_k^{t} = U_k^{t-1} \cup \set{ \edge{\opw{1}{1}, \opw{2}{2}} }$. Otherwise, we set $U_k^{t} = U_k^{t-1} \cup \set{ \edge{\opw{2}{2}, \opw{1}{1}} }$.
	\end{enumerate}
	After processing all pairs of write operations, we add edges for read operations as follows.
	If both operations are reads, then they are already related by $PO$.
	For each read $r \in (\operatorname{r}, i, *, *)$ and write $w \in (\operatorname{w}, *, *, *)$ such that $r, w$ are not related in $U_i^{t-1}$, set $U_i^t = U_i^{t-1} \cup \set{\edge{w, r}}$.
	At the end we set $V_i = U_i^t$ for each process $i$.

	We now show the correctness of the above procedure.
	First observe that we add two type of edges 1) for operations that are not already related, and 2) the edges implied by transitivity.
	Therefore, each $V_i$ is acyclic.
	Thus, each $V_i$ is a total order on $(*, i, *, *) \cup (\operatorname{w}, *, *, *)$ by construction.
	Now note that each $V_i \supseteq U^0_i \supseteq \del{PO | ( *, i, *, * ) \cup ( w, *, *, * )}$ by construction.
	So we show that each $V_i$ respects $SCO(\V)$.
	We proceed via induction and show that at each step $t$, each $U_i^t$ respects $SCO(\V^t)$.
	For the base case $U^0_i$ respects $SCO(\pV^0)$ by construction.
	For the inductive step, we show that $SCO(\pV^t) = SCO(\pV^{t-1})$.
	This implies the result since each $V_i^t \supseteq U_i^{t-1}$ and $U_i^{t-1}$ respects $SCO(\pV^{t-1})$ by the inductive hypothesis.

	If at step $t$ we considered two write operations, then we have $3$ cases to consider.
	\begin{enumerate}
		\item
		$\opw{1}{1} \in (\operatorname{w}, 1, *, *)$ and $\opw{2}{2} \in (\operatorname{w}, 2, *, *)$ are not related in $U_1^{t-1}$ and we set $U_1^t = U_1^{t-1} \cup \set{ \edge{\opw{1}{1}, \opw{2}{2}} }$.
		We show, via contradiction, that $SCO(U_1^t) \setminus SCO(U_1^{t-1})$ is empty and so there are no new $SCO$ edges in this case.
		Suppose $\del{w', w''} \in SCO(U_1^t) \setminus SCO(U_1^{t-1})$.
		Then $w' \le_{U_1^{t-1}} \opw{1}{1}$, $\opw{2}{2} <_{U_1^{t-1}} w''$, and $w'' \in (\operatorname{w}, 1, *, *)$.
		Therefore, $\opw{1}{1}$ and $w''$ are related by $\del{PO | ( *, i, *, * ) \cup ( w, *, *, * )}$.
		If $w'' \PO \opw{1}{1}$, then $w'' <_{U_1^{t-1}} \opw{1}{1}$ and so $\opw{2}{2} <_{U_1^{t-1}} w'' <_{U_1^{t-1}} \opw{1}{1}$, which contradicts the initial assumption that $\opw{1}{1}$ and $\opw{2}{2}$ are not related in $U_1^{t-1}$.
		Thus $\opw{1}{1} \PO w''$.
		This implies that $\opw{1}{1} <_{U_1^{t-1}} w''$ and so $w' \le_{U_1^{t-1}} \opw{1}{1} <_{U_1^{t-1}} w''$.
		Thus $\del{w', w''} \in SCO(U_1^{t-1})$, which contradicts the initial assumption that $\del{w', w''} \in SCO(U_1^t) \setminus SCO(U_1^{t-1})$.

		\item
		$\opw{1}{1} \in (\operatorname{w}, 1, *, *)$ and $\opw{2}{2} \in (\operatorname{w}, 2, *, *)$ are not related in $U_1^{t-1}$ and we set $U_2^t = U_2^{t-1} \cup \set{ \edge{\opw{2}{2}, \opw{1}{1}} }$.
		This is the same as Case 1.

		\item
		For every process $k \ne 1, 2$ such that $\opw{1}{1}, \opw{2}{2}$ are not related in $U_k^{t-1}$, we do the following.
		If $SCO( U_k^{t-1} \cup \set{ \edge{\opw{1}{1}, \opw{2}{2}} } ) = SCO( U_k^{t-1} )$, then we set $U_k^{t} = U_k^{t-1} \cup \set{ \edge{\opw{1}{1}, \opw{2}{2}} }$. Otherwise, we set $U_k^{t} = U_k^{t-1} \cup \set{ \edge{\opw{2}{2}, \opw{1}{1}} }$.
		We proceed via contradiction to show that either $SCO( U_k^{t-1} \cup \set{ \edge{\opw{1}{1}, \opw{2}{2}} } ) \setminus SCO(U_k^{t-1})$ or $SCO( U_k^{t-1} \cup \set{ \edge{\opw{2}{2}, \opw{1}{1}} } ) \setminus SCO(U_k^{t-1})$ is empty and so there are no new $SCO$ edges in this case.
		Suppose $\edge{\opw{3}{}, \opw{4}{}} \in SCO( U_k^{t-1} \cup \set{ \edge{\opw{1}{1}, \opw{2}{2}} } ) \setminus SCO(U_k^{t-1})$ and $\edge{\opw{5}{}, \opw{6}{}} \in SCO( U_k^{t-1} \cup \set{ \edge{\opw{2}{2}, \opw{1}{1}} } ) \setminus SCO(U_k^{t-1})$.
		It follows that $\opw{4}{}, \opw{6}{} \in (\operatorname{w}, k, *, *)$ and therefore, are related by $\del{PO | ( *, k, *, * ) \cup ( w, *, *, * )}$.
		There are two cases to consider.
		\begin{enumerate}[label=\roman*),topsep=0pt]
			\item
			$\opw{6}{} \le_{PO} \opw{4}{}$.
			Since $\edge{\opw{3}{}, \opw{4}{}} \in SCO( U_k^{t-1} \cup \set{ \edge{\opw{1}{1}, \opw{2}{2}} } ) \setminus SCO(U_k^{t-1})$, so $\opw{3}{} \le_{U_k^{t-1}} \opw{1}{1}$ and $\opw{2}{2} \le_{U_1^{t-1}} \opw{4}{}$.
			Since $\edge{\opw{5}{}, \opw{6}{}} \in SCO( U_k^{t-1} \cup \set{ \edge{\opw{2}{2}, \opw{1}{1}} } ) \setminus SCO(U_k^{t-1})$, so $\opw{5}{} \le_{U_k^{t-1}} \opw{2}{2}$ and $\opw{1}{1} \le_{U_1^{t-1}} \opw{6}{}$.
			Therefore $\opw{3}{} \le_{U_k^{t-1}} \opw{1}{1} \le_{U_1^{t-1}} \opw{6}{} \le_{PO} \opw{4}{}$ and thus $\edge{\opw{3}{}, \opw{4}{}} \in SCO(U_k^{t-1})$.
			This contradicts the initial assumption that $\edge{\opw{3}{}, \opw{4}{}} \in SCO( U_k^{t-1} \cup \set{ \edge{\opw{1}{1}, \opw{2}{2}} } ) \setminus SCO(U_k^{t-1})$.
			
			\item
			$\opw{4}{} \le_{PO} \opw{6}{}$.
			This is the same as Case i with the role of $\opw{4}{}$ and $\opw{6}{}$ switched.
		\end{enumerate}
	\end{enumerate}
	
	Now, if step $t$ considered read operations, then all write operations have already been ordered by each $V_i^{t-1}$ and therefore $SCO(\V^t) = SCO(\V^{t-1})$.
	This completes the proof that $SCO(\V^t) = SCO(\V^{t-1})$.
\end{proof}
\\

\begin{proof_of} {Theorem \ref{theorem DRO RnRStrongCausalConsistencyNecessity}}
	Assume for the sake of contradiction that there exists a good record $\R$ of $\V$ such that there exists a process, WLOG process $1$, and two operations $\op{1}{}, \op{2}{} \in (*, 1, *, *) \cup (*, \operatorname{w}, *, *)$ such that
	$\edge{ \op{1}{}, \op{2}{} } \in \trans{A}_1(\V) \setminus PO \cupdot SWO_1( \V ) \cupdot B_1(\V)$ and $\edge{ \op{1}{}, \op{2}{} } \not \in R_1$.
	%
	Then, we construct a set of views $\V'$, such that $DRO(V'_1) \ne DRO(V_1)$ but $\V'$ certifies a strongly causal consistent replay to be valid for $\R$, i.e. $\V'$ explains a strongly causal execution and extends the record $\R$.
	This violates the definition of a good record.
	We use Lemma \ref{lemma sufficient to reproduce DRO} and construct, for each process $i$, a partial order $U_i \supseteq R_i \cup SCO(\V) \cup \del{PO | ( *, i, *, * ) \cup ( w, *, *, * )}$ such that $\edge{ \op{2}{}, \op{1}{} } \in U_1$.
	From Lemma \ref{lemma sufficient to reproduce DRO}, it follows that there exists a strongly causal execution $\V'$ such that for each process $i$, $V_i \supseteq U_i \supseteq R_i$ (and therefore a replay of $\R$) and $\edge{ \op{2}{}, \op{1}{} } \in V_1$.
	Observe that since $\edge{ \op{1}{}, \op{2}{} } \in \trans{A}_1(\V) \setminus \del{PO \cupdot SWO_1( \V ) \cupdot B_1(\V)}$, therefore $\edge{ \op{1}{}, \op{2}{} } \in DRO(V_1)$, $\edge{ \op{2}{}, \op{1}{} } \in DRO(V'_1)$, and so $DRO(V'_1) \ne DRO(V_1)$.
	
	We construct $\pV$ from $\A(\V)$ as follows.
	We slightly abuse notation to set $C_1(\V, \op{1}{}, \op{2}{}) = \emptyset$ if $\op{2}{}$ is a read operation (recall that $C_1(\V, \op{1}{}, \op{2}{})$ is only defined when $\op{2}{}$ is a write operation in Definition \ref{definition DRO Ci}).
	Let $U_1 := \del[1]{ A_1(\V) \setminus \set{\edge{ \op{1}{}, \op{2}{} }} } \cup \set{\edge{ \op{2}{}, \op{1}{} }} \cup C_1(\V, \op{1}{}, \op{2}{})$.
	For each $i > 1$, set $U_i = A_i(\V) \cup C_1(\V, \op{1}{}, \op{2}{})$ (see Definition \ref{definition DRO Ci}).
	
	For correctness we have to show that each $U_i$
	\begin{enumerate}[topsep=0pt]
		\item is a partial order, and
		\item respects $SCO(\pV)$.
	\end{enumerate}

	We first consider the case when $\op{2}{}$ is a read operation.
	We claim that $SCO(\pV) \setminus SCO( \A( \V ) )$ is empty.
	If not, then there exist two write operation $\opw{3}{} \in ( \operatorname{w}, *, *, * )$ and $\opw{4}{1} \in ( \operatorname{w}, 1, *, * )$ such that $\opw{3}{} \le_{A_{1}(\V)} \op{2}{}$ and $\op{1}{} \le_{A_{1}(\V)} \opw{4}{1}$.
	Since $\op{2}{} \in (\operatorname{r}, 1, *, * )$, therefore either $\edge{ \op{2}{}, \opw{4}{1} } \in PO$ or $\edge{ \opw{4}{1}, \op{2}{} } \in PO$.
	In the first case $\edge{ \op{3}{}, \opw{4}{1} } \in SCO( \A( \V ) )$.
	In the second case if $\op{1}{} = \opw{4}{1}$, then $\edge{ \opw{1}{}, \op{2}{} } \in PO$.
	So $\op{1}{} <_{A_1( \V )} \opw{4}{1} <_{ A_1( \V ) } \op{2}{}$ which contradicts the fact that $\edge{ \op{1}{}, \op{2}{} } \in \trans{A}_1(\V)$.
	In either case, we have the desired contradiction.
	Therefore, each $U_i$ is a partial order that respects $SCO(\pV)$.
	
	We now consider the case when $\op{2}{}$ is a write operation.
	Since $\edge{\op{1}{}, \op{2}{}} \not \in B_i(\V)$, therefore each $U_i$ is a partial order by Definition \ref{definition DRO Bi}.
	So we show that for any process $i$, $SCO(U_i) \setminus SCO(A_i(\V)) \subseteq C_1( \V, \op{1}{}, \op{2}{} )$.
	Consider any $\edge{\opw{3}{}, \opw{4}{i}} \in SCO(U_i) \setminus SCO(A_i(\V))$.
	Then there exists a $\opw{3}{} \opw{4}{i}$-path $\rho$ in $A_i \cupdot C_1( \V, \op{1}{}, \op{2}{} )$, given by $\opw{3}{} = \op{\rho, 0}{} \le_{A_{i}(\V)} \op{\rho, 1}{} <_{C_1( \V, \op{1}{}, \op{2}{} )} \op{\rho, 2}{} \le_{A_{i}(\V)} \dots \op{\rho, k-1}{} \le_{A_{i}(\V)} \op{\rho, k}{} = \opw{4}{i}$ with $k > 2$.
	It follows by Definition \ref{definition DRO Ci} that $\edge{\opw{3}{}, \opw{4}{i}} \in C_1( \V, \op{1}{}, \op{2}{} )$.

	So we have shown that $\pV$ meets the conditions of Lemma \ref{lemma sufficient to reproduce DRO} and therefore, we can find a strongly causal consistent replay $\V'$ of $\R$ such that $DRO(V'_1) \ne DRO(V_1)$.
	This contradicts the initial assumption that $\R$ is a good record.
	%
\end{proof_of}
\end{document}